\def\llncs{0}
\def\fullpage{1}
\def\anonymous{0}
\def\authnote{0}
\def\notxfont{0}
\def\submission{0}

\ifnum\submission=1
\def\llncs{1}
\fi

\ifnum\llncs=1
\documentclass[envcountsect,a4paper,runningheads,10pt]{llncs}
\else
	\documentclass[letterpaper,hmargin=1.05in,vmargin=1.05in,11pt]{article}
			\ifnum\fullpage=1
		\usepackage{fullpage}
		\fi
\fi

\usepackage[%
  colorlinks=true,
  citecolor=blue,
  pagebackref=true
]{hyperref}

\usepackage{amsmath, amsfonts, amssymb, mathtools,amscd}

\usepackage{amsthm}

\usepackage{lmodern}
\usepackage[T1]{fontenc}
\usepackage[utf8]{inputenc}

\usepackage{etex}

\usepackage{arydshln} 
\usepackage{url}
\usepackage{ifthen}
\usepackage{bm}
\usepackage{multirow}
\usepackage[dvips]{graphicx}

\usepackage[usenames]{color}
\usepackage{xcolor,colortbl} 
\usepackage{threeparttable}
\usepackage{comment}
\usepackage{paralist,verbatim}
\usepackage{cases}
\usepackage{booktabs}
\usepackage{braket}
\usepackage{cancel} 
\usepackage{ascmac} 
\usepackage{framed}
\usepackage{authblk}
\usepackage{pifont}
\usepackage{qcircuit}
\usepackage{tikz}
\usetikzlibrary{cd}
\definecolor{darkblue}{rgb}{0,0,0.6}
\definecolor{darkgreen}{rgb}{0,0.5,0}
\definecolor{maroon}{rgb}{0.5,0.1,0.1}
\definecolor{dpurple}{rgb}{0.2,0,0.65}

\usepackage[capitalise,noabbrev]{cleveref}
\usepackage[absolute]{textpos}
\usepackage[final]{microtype}
\usepackage[absolute]{textpos}
\usepackage{everypage}
\DeclareMathAlphabet{\mathpzc}{OT1}{pzc}{m}{it}

\usepackage{algorithmic}
\usepackage{algorithm}
\usepackage{here}

\usepackage[normalem]{ulem}

\newtheoremstyle{thicktheorem}%
{\topsep}
{\topsep}
{\itshape}{}%
{\bfseries}%
{.}
{ }%
{\thmname{#1}\thmnumber{ #2}%
		\thmnote{ (#3)}%
}

\newtheoremstyle{remark}
{\topsep}
{\topsep}
	{}
	{}
	{}
	{.}
	{ }
	{\textit{\thmname{#1}}\thmnumber{ #2}
			\thmnote{ (#3)}%
	}

\ifnum\llncs=0
	\theoremstyle{thicktheorem}
	\newtheorem{theorem}{Theorem}[section]
	\newtheorem{lemma}[theorem]{Lemma}

	\newtheorem{definition}[theorem]{Definition}
	
    \newtheorem{fact}[theorem]{Fact}

	\theoremstyle{remark}
	
	\newtheorem{remark}[theorem]{Remark}

\else
\fi

\Crefname{MyClaim}{Claim}{Claims}

	\crefname{theorem}{Theorem}{Theorems}
	\crefname{assumption}{Assumption}{Assumptions}
	\crefname{construction}{Construction}{Constructions}
	\crefname{corollary}{Corollary}{Corollaries}
	\crefname{conjecture}{Conjecture}{Conjectures}
	\crefname{definition}{Definition}{Definitions}
	\crefname{exmaple}{Example}{Examples}
	\crefname{experiment}{Experiment}{Experiments}
	\crefname{counterexample}{Counterexample}{Counterexamples}
	\crefname{lemma}{Lemma}{Lemmata}
	\crefname{observation}{Observation}{Observations}
	\crefname{proposition}{Proposition}{Propositions}
	\crefname{remark}{Remark}{Remarks}
	\crefname{claim}{Claim}{Claims}
	\crefname{fact}{Fact}{Facts}
	\crefname{note}{Note}{Notes}

\ifnum\llncs=1
 \crefname{appendix}{App.}{Appendices}
 \crefname{section}{Sec.}{Sections}
\else
\fi

\ifnum\llncs=1
\renewcommand*{\backref}[1]{}
\else
	\renewcommand*{\backref}[1]{(Cited on page~#1.)}
	\ifnum\notxfont=1
	\else
		\usepackage{newtxtext}
	\fi
\fi

\usepackage{fancyhdr}

\ifnum\authnote=0  
\newcommand{\mor}[1]{}
\newcommand{\shogo}[1]{}
\newcommand{\takashi}[1]{}
\newcommand{\fuyuki}[1]{}
\newcommand{\minki}[1]{}

\else
\newcommand{\mor}[1]{$\ll$\textsf{\color{red} Tomoyuki: { #1}}$\gg$}
\newcommand{\takashi}[1]{$\ll$\textsf{\color{orange} Takashi: { #1}}$\gg$}
\newcommand{\shogo}[1]{$\ll$\textsf{\color{darkgreen} Shogo: { #1}}$\gg$}
\newcommand{\fuyuki}[1]{$\ll$\textsf{\color{darkblue} Fuyuki: { #1}}$\gg$}
\newcommand{\minki}[1]{$\ll$\textsf{\color{darkblue} Minki: { #1}}$\gg$}

\DeclareRobustCommand{\Erase}{\bgroup\markoverwith{\textcolor{red}{\rule[.5ex]{2pt}{0.4pt}}}\ULon}

\fi

\newcommand{\SWAP}{\mathrm{SWAP}}

\newcommand{\Tr}{\mathrm{Tr}}

\newcommand{\Sim}{\algo{Sim}}

\newcommand{\Image}{\algo{Im}}

\newcommand{\cA}{\mathcal{A}}
\newcommand{\cB}{\mathcal{B}}
\newcommand{\cC}{\mathcal{C}}
\newcommand{\cD}{\mathcal{D}}
\newcommand{\cE}{\mathcal{E}}
\newcommand{\cF}{\mathcal{F}}

\newcommand{\cK}{\mathcal{K}}

\newcommand{\cM}{\mathcal{M}}
\newcommand{\cN}{\mathcal{N}}
\newcommand{\cO}{\mathcal{O}}

\newcommand{\cS}{\mathcal{S}}

\newcommand{\cU}{\mathcal{U}}

\newcommand{\cX}{\mathcal{X}}
\newcommand{\cY}{\mathcal{Y}}

\newcommand{\identitymap}{\mathrm{id}}

\def\makeuppercase#1{
\expandafter\newcommand\csname tl#1\endcsname{\widetilde{#1}}
}

\def\makelowercase#1{
\expandafter\newcommand\csname tl#1\endcsname{\widetilde{#1}}
}

\newcommand{\N}{\mathbb{N}}

\newcommand{\R}{\mathbb{R}}
\newcommand{\C}{\mathbb{C}}

\newcommand{\Unitaries}{\mathbb{U}}

\newcommand{\States}{\mathbb{S}}

\newcommand{\regC}{\mathbf{C}}

\newcommand{\regB}{\mathbf{B}}
\newcommand{\regA}{\mathbf{A}}

\newcommand{\secp}{\lambda}

\newcommand*{\algo}[1]{\ensuremath{\mathsf{#1}}}

\newcommand{\str}[1]{\mathsf{str}(#1)}

\newenvironment{boxfig}[2]{\begin{figure}[#1]\fbox{\begin{minipage}{0.97\linewidth}
                        \vspace{0.2em}
                        \makebox[0.025\linewidth]{}
                        \begin{minipage}{0.95\linewidth}
            {{
                        #2 }}
                        \end{minipage}
                        \vspace{0.2em}
                        \end{minipage}}}{\end{figure}}

\newcommand{\bit}{\{0,1\}}

\newcommand{\TD}{\algo{TD}}

\newcommand{\PRF}{\algo{PRF}}

\newcommand{\PRP}{\algo{PRP}}
\newcommand{\PRPinverse}{\algo{PRP}^{-1}}

\newcommand{\negl}{{\mathsf{negl}}}

\newcommand{\poly}{{\mathrm{poly}}}

\DeclareMathOperator*{\Exp}{\mathbb{E}}

\newcommand{\Span}{\mathrm{span}}

\newcommand{\id}{\mathrm{id}}

\newcommand\ovec[1]{\overrightarrow{#1}}

\newcommand{\Wg}{\mathrm{Wg}}

\usetikzlibrary{decorations.markings}
\tikzset{
  cross/.style={
    postaction={decorate,decoration={markings,
    mark=at position 0.45 with {\draw[-,line width=1pt] (-10pt,-10pt) -- (10pt,10pt);\draw[-,line width=1pt] (-10pt,10pt) -- (10pt,-10pt);}}}
  }
}

\makeatletter
\DeclareRobustCommand
  \myvdots{\vbox{\baselineskip4\p@ \lineskiplimit\z@
    \hbox{.}\hbox{.}\hbox{.}}}
\makeatother

\newcommand{\ketbra}[2][\relax]{\ifx\relax#1 \ket{#2}\bra{#2}\else\ket{#1}\bra{#2}\fi}

\newtheorem{myclaim}{Claim}

\makeatletter
\def\ketbra#1{\@ifnextchar\bgroup{\ketbrahelp{#1}}{\ketbrahelp{#1}{#1}}}
\makeatother
\def\ketbrahelp#1#2{\ket{#1}\bra{#2}}

\title{Pseudorandom Function-like States from Common Haar Unitary}

\ifnum\anonymous=1
\ifnum\llncs=1
\author{\empty}\institute{\empty}
\else
\author{}
\fi
\else
\ifnum\llncs=1
\author{
Tomoyuki Morimae\inst{1} \and Takashi Yamakawa\inst{1,2} 
}
\institute{
	Yukawa Institute for Theoretical Physics, Kyoto University, Kyoto, Japan \and NTT Social Informatics Laboratories, Tokyo, Japan 
}
\else
\author[1]{Minki Hhan\footnote{This work was done in part while the first author was in KIAS, Korea.}}
\author[2]{Shogo Yamada}
\affil[1]{{\small The University of Texas at Austin, Texas, USA}
\authorcr{\small minki.hhan@austin.utexas.edu}}
\affil[2]{{\small Yukawa Institute for Theoretical Physics, Kyoto University, Kyoto, Japan}
\authorcr{\small shogo.yamada@yukawa.kyoto-u.ac.jp}}

\fi 
\fi

\date{}

\begin{document}

\thispagestyle{fancy}
\rhead{YITP-24-144}

\maketitle
\begin{abstract}
    Recent active studies have demonstrated that cryptography without one-way functions (OWFs) could be possible in the quantum world.
    Many fundamental primitives that are natural quantum analogs of OWFs or pseudorandom generators (PRGs) have been introduced, and their mutual relations and applications have been studied.
    Among them, pseudorandom function-like state generators (PRFSGs) [Ananth, Qian, and Yuen, Crypto 2022] are one of the most important primitives.
    PRFSGs are a natural quantum analogue of pseudorandom functions (PRFs), and imply many applications such as IND-CPA secret-key encryption (SKE) and EUF-CMA message authentication code (MAC).
    However, only known constructions of (many-query-secure) PRFSGs are ones from OWFs or pseudorandom unitaries (PRUs).

    In this paper, we construct classically-accessible adaptive secure PRFSGs in the invertible quantum Haar random oracle (QHRO) model which is introduced in [Chen and Movassagh, Quantum].
    The invertible QHRO model is an idealized model where any party can access a public single Haar random unitary and its inverse, which can be considered as a quantum analog of the random oracle model.
    Our PRFSG constructions resemble the classical Even-Mansour encryption based on a single permutation, and are secure against any unbounded polynomial number of queries to the oracle and construction.
    To our knowledge, this is the first application in the invertible QHRO model without any assumption or conjecture.
    The previous best construction in the idealized model is PRFSGs secure up to $o(\secp/\log \secp)$ queries in the common Haar state model [Ananth, Gulati, and Lin, TCC 2024].

    We develop new techniques on Haar random unitaries to prove the selective and adaptive security of our PRFSGs. 
    For selective security, we introduce a new formula, which we call the Haar twirl approximation formula. 
    For adaptive security, we show the unitary reprogramming lemma and the unitary resampling lemma. These have their own interest, and
    may have many further applications.
    In particular, by using the approximation formula, we give an alternative proof of the non-adaptive security of the PFC ensemble [Metger, Poremba, Sinha, and Yuen, FOCS 2024] as an additional result.

    Finally, we prove that our construction is not PRUs or quantum-accessible non-adaptive PRFSGs by presenting quantum polynomial time attacks. Our attack is based on generalizing the hidden subgroup problem where the relevant function outputs quantum states.
\end{abstract}

\ifnum\submission=1
\else
\clearpage
\newpage
\setcounter{tocdepth}{2}
\tableofcontents
\fi
\newpage

\section{Introduction}

\noindent
In classical cryptography, one-way functions (OWFs) are the minimal assumption~\cite{FOCS:ImpLub89}, because
many primitives, such as pseudorandom generators (PRGs), pseudorandom functions (PRFs), 
secret-key encryption (SKE), message authentication code (MAC), digital signatures, and commitments, 
are all existentially equivalent to OWFs.
Moreover, almost all primitives (including important applications such as public-key encryption and multiparty computations) 
imply OWFs.

In the quantum world, on the other hand, 
OWFs are not necessarily the minimum assumption~\cite{TQC:Kre21,C:MorYam22,C:AnaQiaYue22}.
Many fundamental primitives have been introduced such as pseudorandom unitaries (PRUs)~\cite{C:JiLiuSon18,MaHsi24}, 
pseudorandom function-like state generators (PRFSGs)~\cite{C:AnaQiaYue22,TCC:AGQY22},
pseudorandom state generators~(PRSGs)~\cite{C:JiLiuSon18}, one-way state generators (OWSGs)~\cite{C:MorYam22}, 
one-way puzzles~(OWPuzzs)~\cite{STOC:KhuTom24}, unpredictable state generators (UPSGs) \cite{MorYamYam24}, and EFI pairs~\cite{ITCS:BraCanQia23}. 
Although they are believed to be weaker than OWFs~\cite{TQC:Kre21,STOC:KQST23,STOC:LomMaWri24},
they still imply many useful applications such as private-key quantum money schemes~\cite{C:JiLiuSon18}, SKE~\cite{C:AnaQiaYue22}, 
MAC~\cite{C:AnaQiaYue22}, 
digital signatures~\cite{C:MorYam22}, commitments~\cite{C:MorYam22,C:AnaQiaYue22}, 
and multiparty computations~\cite{C:MorYam22,C:AnaQiaYue22}.

Among them, pseudorandom function-like state generators (PRFSGs)~\cite{C:AnaQiaYue22,TCC:AGQY22} are one of the most important primitives.
PRFSGs are a natural quantum analogue of pseudorandom functions (PRFs).
A PRFSG is a quantum polynomial-time (QPT) algorithm $G$ that takes a classical key $k$ and a bit string $x$ as input,
and outputs a quantum state $|\phi_k(x)\rangle$. 
Roughly speaking, the security requires that no QPT adversary can distinguish whether it is querying to $G(k,\cdot)$ with a random $k$ or
an oracle that outputs Haar random states.\footnote{More precisely, the oracle works as follows. If it gets $x$ as input and $x$ was not queried before,
it samples a Haar random state $\psi_x$ and returns it. If $x$ was queried before, it returns the same state $\psi_x$ that was sampled before when $x$ was queried for the first time.}
PRFSGs imply almost all known primitives such as UPSGs, PRSGs, OWSGs, OWPuzzs, and EFI pairs.
PRFSGs also imply useful applications such as IND-CPA SKE, EUF-CMA MAC, 
private-key quantum money schemes, commitments, multi-party computations, (bounded-poly-time-secure) digital signatures, etc.
However, all known constructions of (multi-query-secure) PRFSGs are ones from OWFs or PRUs~\cite{C:AnaQiaYue22,TCC:AGQY22}.

In classical cryptography,
some idealized setups where parties can access some public source of randomness 
are often introduced, such as the common random string model~\cite{blum2019non} or the random function or permutation oracle model~\cite{CCS:BelRog93,EM97}.
These idealized setups reflect the reality of random source or hash functions and naturally provide the practical instantiations of basic cryptographic primitives, and they serve as a testbed for new analysis tools in classical and post-quantum settings~\cite{EC:DunKelSha12,alagic2022post}.

It is natural to consider their quantum counterparts: public sources of quantum states and unitaries.
In fact, various quantum analogue of the setup models have already been introduced~\cite{CM_2024,cryptoeprint:2022/435,C:MorYam24,C:Qian24,chen2024power,ananth2024cryptography,ITCS:AdamBillUmesh20}, 
and several primitives have been constructed including commitments, PRSGs, and restricted-copy secure PRFSGs.
In particular, \cite{ananth2024cryptography} recently constructed bounded-query PRFSGs in the common Haar state (CHS) model, which was shown to be optimal by the authors.

However, most previous works focus on the idealized model where the parties have access to the common \emph{states}, except for \cite{ITCS:AdamBillUmesh20,CM_2024}. The idealized model for common unitaries must be much more useful than common states and perhaps connect the practical and heuristic constructions of quantum cryptographic objects and theory in the near future, as in the random oracle and ideal cipher models in the classical and post-quantum world.
In particular, the limitation of the CHS model motivates the following question:
\begin{center}
    \emph{Are multi-copy secure PRFSGs achievable if a common random unitary is given?}
\end{center}

\subsection{Our Results}

\paragraph{PRFSGs in the invertible QHRO Model.}
The main result of the present paper is
a construction of classically-accessible adaptive secure PRFSGs in the invertible quantum Haar random oracle (QHRO) model which is a quantum analog of the random oracle model. In the invertible 
QHRO model, which is introduced in \cite{CM_2024} and considered in \cite{ITCS:AdamBillUmesh20}, any party can query the same Haar random unitaries $\cU\coloneqq\{U_\secp\}_{\secp\in\N}$ and their inverses $\cU^\dag\coloneqq\{U_\secp^\dag\}_{\secp\in\N}$, where $U_\secp$ is a $\secp$-qubit Haar random unitary.
\footnote{
In \cite{CM_2024,ITCS:AdamBillUmesh20}, they consider that anyone has access to $\secp$-qubit unitary $U_\secp$ and its inverse for specific $\secp$. On the other hand, in this work, we consider that any party has access to Haar random unitaries $\{U_\secp\}_\secp$ and their inverses, where $U_\secp$ is $\secp$-qubit Haar random unitary for each $\secp\in\N$.}

\begin{theorem}[Informal]\label{intro_thm:PRFS_in_QHRO}
Classically-accessible adaptive secure PRFSGs exist in the invertible QHRO model. 
\end{theorem}
More precisely, given the common Haar unitaries $\cU$,
our construction of a PRFSG $G^{\cU}(k,x)\to|\phi_k(x)\rangle$ is the following one:
For any $x,k,k'\in\bit^\secp$,
\begin{align}
|\phi_k(x)\rangle\coloneqq 
X^{k'}U_{\secp}X^k\ket{x}, 
\end{align}
where $X^k$ applies Pauli $X^{k_i}$ on $i$th qubit for each $i\in[\secp]$. This $XUX$ construction resembles the Even-Mansour encryption, the simplest encryption scheme based on a single permutation~\cite{EM97,EC:DunKelSha12}.

Our PRFSG is classically-accessible adaptive secure. Roughly speaking, it means that, for each $x$, $\ket{\phi_k(x)}$ looks like an independent Haar random state even given access to $\cU$ and $\cU^\dag$.
The precise meaning is as follows: let $\cA^{(\cdot,\cdot,\cdot)}$ be an {\it unbounded} adversary such that
\begin{itemize}
    \item $\cA^{(\cdot,\cdot,\cdot)}$ can query the first oracle only classically but adaptively at most $\poly(\secp)$ times.
    \item $\cA^{(\cdot,\cdot,\cdot)}$ can query the second and third oracle quantumly and adaptively at most $\poly(\secp)$ times.
\end{itemize}
Then, for any such $\cA^{(\cdot,\cdot,\cdot)}$,
\begin{align}
            \bigg|\Pr_{\cU\gets\mu,k\gets\bit^\secp}[1\gets\cA^{\cO^{\cU}_{\text{PRFS}}(k,\cdot),\cU,\cU^\dag}]-\Pr_{U\gets\mu,\cO_{\text{Haar}}}[1\gets\cA^{\cO_{\text{Haar}}(k,\cdot),\cU,\cU^\dag}]\bigg|\le\negl(\secp),
\end{align}
where $\cU=\{U_\secp\}_{\secp\in\N}\gets\mu$ means that, for each natural number $\secp$, $U_\secp$ is sampled from the Haar measure over $\secp$-qubit unitary group.
Here, $\cO^{\cU}_{\text{PRFS}}$ and $\cO_{\text{Haar}}$ are defined as follows:
    \begin{itemize}
            \item $\cO^\cU_{\text{PRFS}}(k,\cdot)$: It takes $x\in\bit^\secp$ as input and outputs $G^\cU(k,x)=\ket{\phi_k(x)}$.
            \item $\cO_{\text{Haar}}(\cdot)$: It takes $x\in\bit^\secp$ as input and outputs $\ket{\psi_x}$, where $\ket{\psi_x}$ is sampled from the Haar measure over all $\secp$-qubit pure states for each $x\in\bit^\secp$.
    \end{itemize}

\paragraph{PRSGs in the invertible QHRO Model.}

As in the plain model, PRFSGs trivially imply PRSGs in the invertible QHRO model. 
As an cororally of \cref{intro_thm:PRFS_in_QHRO}, we also have the following result:
\begin{theorem}[Informal]
    PRSGs exist in the invertible QHRO model. 
\end{theorem}
The construction is the trivial one, namely, $G^\cU(k)$ outputs $|\phi_k\rangle\coloneqq U_\secp|k\rangle$, where $k\in\bit^\secp$.
The security means that 
for any polynomial $t$, $\ket{\phi_k}^{\otimes t}$ is statistically indistinguishable even given access to $\cU$ and $\cU^\dag$.

The only known previous construction of PRSGs in the invertible QHRO model~\cite{ITCS:AdamBillUmesh20} is more complicated and requires high query depth to the common unitary. Namely, their construction has to query a common Haar unitary $U_\secp$ $\poly(\secp)$ times. On the other hand, our construction queries a common Haar unitary $U_\secp$ only at once, which is simpler construction than \cite{ITCS:AdamBillUmesh20}.

\paragraph{Our PRFSG is not quantum-accessible secure (and therefore not PRUs).}
We complement this result by proving that our $XUX$ construction is \emph{not} secure quantum-accessible PRFSGs, even non-adaptively and without accessing inverse Haar unitary oracles. In particular, this implies that the construction is not PRUs. 
Concretely, our attack learns the secret keys in polynomial time\footnote{A concurrent paper \cite{AnaBosGulYao24} proves that a single depth is insufficient to construct PRUs inspired by \cite{chen2024power}. However, their attack is information-theoretic, thus two results are incomparable.} given non-adaptive access to $U_\secp$ and $U_\secp P$ using a variant of Simon's algorithm~\cite{C:Simon96} for quantum states, inspired by the quantum attack on the Even-Mansour encryption~\cite{KM12}.

We also prove that a similar attack can break $UP$, a naturally strengthened variant of $UX$, but using random Pauli $P$ instead of random $X$. 
We believe a similar attack breaks the quantum-accessible security of the $PUP$ construction.


\paragraph{Haar Twirl Approximation Formula.}
To show a special case of \cref{intro_thm:PRFS_in_QHRO}, we introduce 
a new formula, which we call {\it Haar twirl approximation formula},
which is our technical contribution.
The formula is written as follows:\footnote{\cite{metger2024simple} implicitly showed a similar result, but our formula is simpler. 
Moreover, our formula is true for any state $\rho$, while their result holds only for specific states $\rho$ on the distinct subspace.} 
\begin{lemma}\label{intro_lem:approximation_formula_for_Haar_k-fold}
   Let $k,d\in\N$ such that $d>\sqrt{6}k^{7/4}$. Define $S_k$ to be the set of all permutations over $k$ elements. 
   Let $\regA$ be a $d^k$-dimentional register, and $\regB$ be any register. Then, for any quantum state $\rho$ on the registers $\regA\regB$,
    \begin{align}
        \left\|
        (\cM^{(k)}_{\text{Haar},\regA}\otimes\identitymap_\regB)(\rho_{\regA\regB})-\sum_{\sigma\in S_k}\frac{1}{d^k}R^\dag_{\sigma,\regA}\otimes\Tr_{\regA}[(R_{\sigma,\regA}\otimes I_\regB)\rho_{\regA\regB}]
        \right\|_1
        \le O\left(\frac{k^2}{d}\right),
    \end{align}
    where $\cM^{(k)}_{\text{Haar}}(\cdot)\coloneqq\Exp_{U\gets\mu_d}U^{\otimes k}(\cdot)U^{\dag\otimes k}$, $\mu_d$ is the Haar measure over $d$-dimensional unitaries, and 
    $R_\pi$ is the permutation unitary that acts $R_\pi\ket{x_1,...,x_k}=\ket{x_{\pi^{-1}(1)},...,x_{\pi^{-1}(k)}}$ for all $x_1,...,x_k\in[d]$
    for each $\pi\in S_k$. 
\end{lemma}

We show \cref{intro_lem:approximation_formula_for_Haar_k-fold} based on Weingarten calculus \cite{collins2006integration}. However, for applications, we do not need any complicated facts about Wingarden calculus and Haar measure because \cref{intro_lem:approximation_formula_for_Haar_k-fold} is stated based on only permutation unitaries. 

\paragraph{Alternative Proof of \cite{metger2024simple}.}
The above formula should be of independent interest, and will have many other applications.
In fact, by using the formula, we show
alternative proof of the non-adaptive security of PFC emsemble~\cite{metger2024simple}.
In their proof, they used the Schur-Weyl duality, but our \cref{intro_lem:approximation_formula_for_Haar_k-fold} is based on the Weingarten calculus \cite{collins2006integration}. 
This new approach will be useful in other applications.

\paragraph{Unitary reprogramming and resampling lemma.}
To prove \cref{intro_thm:PRFS_in_QHRO} with the adaptive security, we follow the post-quantum security proof of the Even-Mansour encryption \cite{alagic2022post}.
Along the way, we develop the unitary variants of their main lemmas, the (arbitrary) reprogramming lemma and resampling lemma. 
The unitary (arbitrary) reprogramming lemma can be understood as an adaptive version of the generalization of the lower bound of Grover's search \cite{EC:AMRS20}, which was used in various applications including the post-quantum security of MAC.
The unitary resampling lemma can be thought of as a unitary variant (and generalization) of the adaptive reprogramming lemma \cite{AC:GHHM21}, which was widely used in e.g., Fiat-Shamir signature and transform \cite{EC:KilLyuSch18,EC:SaiXagYam18} or in some of the first quantum applications of random oracles \cite{C:Unruh14,EC:Unruh14,ES15}. We believe the unitary variant presented in this paper must have further applications in quantum cryptography.

\subsection{Related work}

\paragraph{Comparison with PRFSGs in the CHS Model.}
A recent work \cite{ananth2024cryptography} constructed bounded-copy PRFSGs in the CHS model.
Compared with their PRFSGs, our PRFSGs have an important advantage in that
the number of queries allowed for the adversary is not limited: it is an unbounded polynomial time.
PRFSGs in the CHS model \cite{ananth2024cryptography} allows only $o(\secp/\log\secp)$ number of copies, and it was shown to be optimal.
Our result overcomes the barrier by considering the invertible QHRO model.

\paragraph{Comparison with previous works about invertible QHRO model.}
As mentioned, the invertible QHRO model was considered in \cite{CM_2024,ITCS:AdamBillUmesh20}. In \cite{CM_2024}, they conjectured the Gap-Local-Hamiltonian problem has a succinct argument in the invertible QHRO model. In \cite{ITCS:AdamBillUmesh20}, they provide an idea of how to construct PRSGs and a security proof sketch in the invertible QHRO model based on some (unproven but very plausible) claim that might be proven through the Weingarten calculus.
We give formal security proof of PRFSGs without any conjectures. Moreover, our result immediately implies the existence of PRSGs in the invertible QHRO model, which supersedes \cite{ITCS:AdamBillUmesh20}.

\paragraph{Comparison with the concurrent work \cite{AnaBosGulYao24}.}
Ananth, Bostanci, Gulati and Lin independently and concurrently show similar results in \cite{AnaBosGulYao24}. They consider the inverseless QHRO model in which an adversary can query common Haar random unitary but cannot query its inverse, and construct PRUs, classically-accessible adaptive secure PRFSGs, and PRSGs in the inverseless QHRO model.
The strength of their result is to construct PRUs.  They also prove that the query depth 1 construction cannot be information-theoretic secure PRUs by suggesting the polynomial query attack.
We do not construct PRUs, but our classically-accessible adaptive PRFSGs are secure even if an adversary has access to not only the common Haar random unitary but also its inverse. 

\paragraph{The state hidden subgroup problem} We consider a variant of hidden subgroup problems when breaking the quantum-accessible security. 
{Two concurrent works \cite{BouGiuWri24,MutZha24} observe and use the quantum state version of the hidden subgroup problem in different contexts.}

\subsection{Open Problems}
\begin{itemize}
    \item Can we construct quantumly-accessible adaptive secure PRFSGs in the invertible QHRO model?
    \item Can we construct PRUs and strong PRUs \cite{MaHsi24}\footnote{Strong PRUs are efficiently implementable unitaries which are computationally indistinguishable from Haar random unitaries even given access to them and their inverses.} in the invertible QHRO model? As mentioned, the recent concurrent work \cite{AnaBosGulYao24} shows that (inverseless) PRUs in the inverseless QHRO model. However, their construction is broken when the inverse queries are allowed using the attack presented in the same paper for the single query construction. Concretely, is $X UX U X$ strong PRUs? This candidate deviates from the known impossibility.
    \item For the $XUX$ construction, can we use the same key for two $X$ operators? Or, can we prove stronger security of the construction, e.g., secure PRUs with pure state inputs?
    \item Can we find further applications of the new techniques presented in this paper? 
    Our tools are quite different from the tools used in the recent studies of the random unitaries; the Schur-Weyl duality \cite{metger2024simple} or the path-recording technique \cite{bostanci2024efficient,AnaBosGulYao24} developed in \cite{MaHsi24}.
    The Haar Twirl approximation formula may be useful in the application of PRUs. The classical counterparts or relatives of the unitary reprogramming and resampling lemmas are one of the main tools in the post-quantum security analysis.
\end{itemize}

\subsection{Technical Overviews}
Our construction uses only single $U_\secp$. Since each $U_\secp$ is sampled independently, it suffices to consider the case an adversary queries the same $U_\secp$. We write it just $U$ for notational simplicity. For $k,x\in\bit^\secp$, we construct PRFSGs $\ket{\phi_{k,k'}(x)}$ as 
\begin{align}
    \ket{\phi_{k}(x)}\coloneqq UX^k\ket{x}\text{ or }X^{k'}UX^k\ket{x},
\end{align}
where $X^k$ is the $\secp$-qubit Pauli operator defined by $\bigotimes_i X^{k_i}$. Here $X^{k_i}$ acts on $i$th qubit.
We show non-adaptive security for $UX^k\ket{x}$ and adaptive security for $X^{k'}UX^k\ket{x}$ based on independent techniques, which we will explain below.

\subsubsection{Non-Adaptive Security}

First, let us consider when an adversary can query $U$ only non-adaptively. This precisely means that 
for any {\it unbounded} adversary $\cA$,
any polynomial $t$, and any bit strings
$x_1,...,x_{\ell(\secp)}\in\bit^\secp$ with any polynomial $\ell$,
$\Pr[\top\gets\cC]\le1/2+\negl(\secp)$ in the following security game.
\begin{enumerate}
    \item
    $\cA$ can apply $U$ on its state.
    Note that $n_i$ can be equal to $n_j$ for any $i\neq j$.
    \item $\cC$ samples $b\gets\bit$. 
    If $b=0$, $\cC$ chooses $k\gets\bit^\secp$ and runs $\ket{\phi_k(x_i)}\gets G^U(k,x_i)$ $t(\secp)$ times for each $i\in[\ell(\secp)]$. 
    Then $\cC$ sends $\ket{\phi_k(x_1)}^{\otimes t}\otimes...\otimes\ket{\phi_k(x_\ell)}^{\otimes t}$ to $\cA$. 
    If $b=1$, $\cC$ sends $\ket{\psi_1}^{\otimes t}\otimes...\otimes\ket{\psi_\ell}^{\otimes t}$ to $\cA$, where $\ket{\psi_i}$ is a Haar random $\secp$-qubit state for each $i\in[\ell(\secp)]$.
    \item $\cA$ returns $b'\in\bit$. Note that $\cA$ cannot query $U$ after receiving the challenge state.
    \item 
    $\cC$ outputs $\top$ if and only if $b=b'$.
\end{enumerate}
As we can see in the above definition of the security game, the adversary $\cA$ can query $U$ {\it only before}
it receives the challenge state.

Since our construction uses only single $U$, it suffices to claim that the trace norm between the following two states is at most negligible in $\secp$ for any polynomial $t,\ell$, any bit strings $x_1,...,x_\ell\in\bit^\secp$, and any quantum state $\rho$:
\begin{align}
    &\Exp_{\substack{ U\gets\mu_{2^\secp},\\k\gets\bit^\secp}}\bigotimes^\ell_{i=1}(U X^k\ket{x_i}\bra{x_i}X^{k\dag} U^\dag)^{\otimes t}_\regA\otimes U^{\otimes m}_{\regB}\rho_{\regB\regC} U^{\dag\otimes m}_{\regB}\label{eq:eq1_proof_overview_of_non-adaptive_security}\\
    &\bigotimes^\ell_{i=1}\bigg(\Exp_{\ket{\psi_i}\gets\mu^{s}_{2^\secp}}\ket{\psi_i}\bra{\psi_i}^{\otimes t}\bigg)_\regA\otimes\Exp_{U\gets\mu_{2^\secp}} U^{\otimes m}_{\regB}\rho_{\regB\regC} U^{\dag\otimes m}_{\regB}.\label{eq:eq2_proof_overview_of_non-adaptive_security}
\end{align}
Here $\mu_{2^\secp}$ denotes the Haar measure over all $\secp$-qubit unitary, and $\mu^{s}_{2^\secp}$ denotes the Haar measure over all $\secp$-qubit pure states.

The main challenge is to compare the above two states. One possible way is to calculate the expectation of the Haar random unitary $U$ by invoking the Schur-Weyl duality as in \cite{metger2024simple}. However, this approach encounters challenges due to the limitations of Schur-Weyl duality in facilitating comparisons between different moments of the Haar measure. In our situation, a $(t\ell+m)$th moment of the Haar measure appears in \cref{eq:eq1_proof_overview_of_non-adaptive_security}, but a $m$th moment of the Haar measure appears in \cref{eq:eq2_proof_overview_of_non-adaptive_security}.

\paragraph{Solution: Approximation of the Haar Twirl.}
In order to overcome this challenge, we show and invoke the following approximation formula for the Haar twirl as we state it in \cref{intro_lem:approximation_formula_for_Haar_k-fold}: for any quantum state $\rho$,
\begin{align}
        \left\|
        (\cM^{(k)}_{\text{Haar},\regA}\otimes\identitymap_\regB)(\rho_{\regA\regB})-\sum_{\sigma\in S_k}\frac{1}{d^k}R^\dag_{\sigma,\regA}\otimes\Tr_{\regA}[(R_{\sigma,\regA}\otimes I_\regB)\rho_{\regA\regB}]
        \right\|_1
        \le O\left(\frac{k^2}{d}\right),\label{eq:eq3_proof_overview_of_non-adaptive_security}
    \end{align}
    where $\cM^{(k)}_{\text{Haar}}(\cdot)\coloneqq\Exp_{U\gets\mu_d}U^{\otimes k}(\cdot)U^{\dag\otimes k}$, $\mu_d$ is the $d$-dimensional Haar measure, and 
    $R_\pi$ is the unitary that acts $R_\pi\ket{x_1,...,x_k}=\ket{x_{\pi^{-1}(1)},...,x_{\pi^{-1}(k)}}$ for all $x_1,...,x_k\in[d]$
    for each $\pi\in S_k$. 

From \cref{eq:eq3_proof_overview_of_non-adaptive_security}, the Haar twirl can be approximated as the summation of permutation unitary $R_\pi$. This helps us to compare \cref{eq:eq1_proof_overview_of_non-adaptive_security} with \cref{eq:eq2_proof_overview_of_non-adaptive_security}.
By \cref{eq:eq3_proof_overview_of_non-adaptive_security} and the property of random Pauli operator, we can show \cref{eq:eq1_proof_overview_of_non-adaptive_security} is statistically close to \cref{eq:eq2_proof_overview_of_non-adaptive_security}.

\subsubsection{Adaptive Security}
The quantitative security of the PUP construction is as follows.
\begin{theorem}
    Suppose that $\cA$ makes $p$ classical queries to the PRFSG oracle and $q$ queries to the (invertible) $n$-qubit Haar random unitary oracle $U$. Then
    it holds that
    \begin{align}
            \left|\Pr_{U\gets \mu, (k,k') \gets \bit^n}\left[
                \cA^{X^{k'} U X^k,U} \to 1
            \right] - \Pr_{U\gets \mu, W \gets \mu}\left[
                \cA^{W,U} \to 1
            \right]\right|
        = O\left(\sqrt{\frac{p^3+p^2q^2}{2^n}}\right).
    \end{align}
\end{theorem}
The proof of the above theorem closely resembles the \emph{post-quantum} security proof of the Even-Mansour cipher~\cite{alagic2022post}, which uses the standard hybrid arguments by changing the oracles. To this end, we develop the following resampling and reprogramming lemmas for the unitary oracles.

\begin{lemma}[Unitary reprogramming lemma, informal]
    Consider the following experiment:
    \begin{description}
        \item[Phase 1:] $\cD$ outputs a unitary $F_0=F$ over $m$-qubit, and a quantum algorithm $\cC$ that decides how to reprogram $F$.
        \item[Phase 2:] 
        $\cC$ is executed and reprogram $F$ on the subspace $S$ and outputs the reprogrammed unitary $F'$.
        A random $b\in \bit$ is chosen and $\cD$ receives oracle access (with only forward queries) to $F$ (if $b=0$) or $F'$ (if $b=1$).
        \item[Phase 3:] $\cD$ loses access to the oracle access, and is noticed how $F$ is reprogrammed, and outputs a bit $b'$.
    \end{description}
    Then, it holds that $
    \left|
        \Pr\left[
            \cD \to 1 | b=1
        \right]-
        \Pr\left[
            \cD \to 1 | b=0
        \right]
    \right|\le q \cdot \sqrt{2\epsilon},$
    where $\epsilon$ is the maximum overlap of any quantum state and the reprogrammed space.
\end{lemma}

\begin{lemma}[Unitary resampling lemma, informal]
    Consider the following experiment:
    \begin{description}
        \item[Phase 1:] $\cD$ specifies two distributions of quantum states $D_0,D_1$.
        $\cD$ makes $q$ forward or inverse queries to a $m$-dimensional Haar random unitary $U$.
        \item[Phase 2:] 
        A random $b\in \bit$ is chosen, and $\cD$ makes an arbitrary many queries to $U$ (if $b=0$) or $U'$ (if $b=1$) where $U'$ is defined by $U\circ \SWAP_{{\mu_0},{\mu_1}}$ for $\ket{\mu_0}\gets D_0,\ket{\mu_1}\gets D_1$.\footnote{$\SWAP_{x,y}$ maps $\alpha\ket{x}+\beta \ket{y} \mapsto \alpha\ket{y} + \beta\ket{y}$ in the span of $\{\ket{x},\ket{y}\}.$}
    \end{description}
    Then, the following holds:$
        \left|
            \Pr\left[
                b'=1 | b=0
            \right] -
            \Pr\left[
                b'=1 | b=1
            \right]
        \right|\le 2\sqrt{\frac{6q}{2^m}}$ given that $D_0,D_1$ satisfy some uniformity conditions.
\end{lemma}

We write $X^{k'} UX^k =: V_k^U$ for simplicity. The real-world experiment where the algorithm is given oracle access to $V_k^U,U$ can be denoted by ${\bf H}_0$ below:
\begin{align}
    {\bf H}_{0}:& {U},{V_k^U,U,V_k^U,U,....}
\end{align}
where each $U$ may include multiple forward and inverse queries to the unitary $U$, and $V_k^U$ denotes a single classical-query to the PRFSG oracle $V_k^U=X^{k'} UX^k$. 

We consider the following hybrid experiments, which are close because of the unitary resampling lemma:
\begin{align}
    {\bf H}_{0}:& \underbrace{U}_{\textbf{Phase 1}},\underbrace{V_k^U,U,V_k^U,U,....}_{\textbf{Phase 2}}\\
    {\bf H}_{0}':& \underbrace{U}_{\textbf{Phase 1}},\underbrace{V_k^{U'},U',V_k^{U'},U',....}_{\textbf{Phase 2}}
\end{align}
where in ${\bf H_0}'$, we \emph{resample} the random unitary $U$ into $U'$ so that $V_k^{U'}$ maps the first query to almost random state. More concretely, for the first query $x_1$, we define the distribution $D_0$ samples $x_1 \oplus k$ for the random key $k$, and $D_1$ samples a random state. Then, it holds that 
\begin{align}
    V_k^{U'} \ket{x_1}= X^{k'} U' X^k\ket{x_1} = X^{k'} U \circ \SWAP_{x_1\oplus k , \mu}\ket{x_1 \oplus k} = X^{k'} U\ket{\mu} =: \ket{\nu}
\end{align}
for random state $\ket{\mu}$. In turn, the output state $\ket{\nu}$ also looks random. 

We define $W$ as a random unitary that maps $\ket{x_1}$ to $\ket{\nu}$. Then, with some calculations, we can prove that the hybrid experiment ${\bf H}_{0}'$ is close to the following hybrid:
\begin{align}
    {\bf H}_{0}'':& {U},{W,U',V_k^{U},U',....}
\end{align}
Intuitively, the first query to the PRFSG oracle is identical, and the later steps may differ slightly due to the replacement of oracles. The actual proof requires multiple hybrid experiments in between.

Then, the unitary reprogramming lemma for $F=\ketbra 0 \otimes U + \ketbra 1 \otimes U^\dagger$ proves the following two hybrids are close:
\begin{align}
    {\bf H}_{0}'':& \underbrace{U,W}_{\textbf{Phase 1}},\underbrace{U'}_{\textbf{Phase 2}},\underbrace{V_k^{U},U',....}_{\textbf{Phase 3}}\\
    {\bf H}_{1}:& \underbrace{U,W}_{\textbf{Phase 1}},\underbrace{U}_{\textbf{Phase 2}},\underbrace{V_k^{U},U',....}_{\textbf{Phase 3}}\\
\end{align}
then the proof continues to reach
\begin{align}
    {\bf H}_{p}:& {U},{W,U,W,U,....}
\end{align}
that corresponds to the ideal-world experiment. By the standard hybrid argument, we conclude that ${\bf H}_0$ and ${\bf H}_p$ are close, i.e., the $p$-query algorithm cannot distinguish the oracle $X^{k'} UX^k$ and $W$ with a high probability.

\paragraph{Simulation of Unitary Oracles.} A careful reader may notice (multiple) problems in the above arguments. In particular, if the adversary queries the same input $x$ to the PRFSG oracle twice, some intermediate hybrid answers them using one by $V_k^U$ and the other by $W$. Another subtle problem is that reprogramming and resampling may require the knowledge of some pure quantum states related to the oracle input/outputs.

We detour these problems by considering a \emph{stronger} oracle algorithm. When an oracle $W$ is always queried by classical inputs by $A$, we consider the simulation oracle $\Sim(W)$ and corresponding the simulation algorithm $\Sim(A)$ that work roughly as follows:
\begin{itemize}
    \item When $A$ makes the classical input query $x$ to $W$ and obtain a pure state $\ket{\phi_x} = W\ket{x}$, $\Sim(A)$ queries $x$ to $\Sim(W)$ and obtains a perfect classical description of $\ket{\phi_x}$ as an answer. Then it constructs $\ket{\phi_x}$ by itself and proceeds as $A$. If the same query $x$ is made by $A$ later, $\Sim(A)$ does not make any query and constructs $\ket{\phi_x}$ by itself again. The other behavior of $\Sim(A)$ is identical to $A$.
\end{itemize}
Note that even if the oracle $W$ is changed during the experiment, $\Sim(A)$ does the same behavior. Therefore, even if the oracle is replaced in the middle, the same input to the oracle is answered by the same output state. 

By simulating the PRFSG oracles, we can extract the query information of the PRFSG oracles, thereby reprogramming and resampling relevant spaces without affecting the algorithm's success probability. This resolves the subtle problems mentioned above and completes the security proof.
\section{Preliminaries}\label{sec:preliminaries}

\subsection{Basic Notations}
\label{sec:basic_notations}
This paper uses the standard notations of quantum computing and cryptography. 
We use $\secp$ as the security parameter.
$[n]$ means the set $\{1,2,...,n\}$.
For any set $S$, $x\gets S$ means that an element $x$ is sampled uniformly at random from the set $S$.
We write $\negl$ as a negligible function
and $\poly$ as a polynomial. 
QPT stands for quantum polynomial-time.
For an algorithm $A$, $y\gets A(x)$ means that the algorithm $A$ outputs $y$ on input $x$.

The identity operator over $d$-dimensional space is denoted by $I_d$. When the dimension is clear from the context, we sometimes write $I$ for simplicity.
We use $X$, $Y$ and $Z$ as Pauli operators.
For a bit string $x$, $X^x\coloneqq\bigotimes_iX^{x^i}$. We use $Y^y$ and $Z^z$ similarly.
The $n$-qubit Pauli group is defined as $\{X^xZ^z\}_{x,z\in\bit^n}$.

For a vector $\ket{\psi}$, we define its norm as $\|\ket{\psi}\|\coloneqq\sqrt{\braket{\psi|\psi}}$.
For two density matrices $\rho$ and $\sigma$, the trace distance is defined as
${\rm TD}(\rho,\sigma)\coloneqq 
\frac{1}{2}\|\rho-\sigma\|_1 =
\frac{1}{2} \mathrm{Tr}\left[\sqrt{(\rho - \sigma)^2}\right]$,
where $\|\cdot\|_1$ is the trace norm.
For any matrix $A$, we define the Frobenius norm $\|A\|_2$ as $\sqrt{\Tr[A^\dag A]}$.
For any matrix $A$, the operator norm $\|\cdot\|_\infty$ is defined as $\|A\|_\infty\coloneqq\max_{\ket{\psi}}\sqrt{\bra{\psi}A^\dag A\ket{\psi}}$, where the maximization is taken over all pure state $\ket{\psi}$.
$\identitymap$ denotes the identity channel, i.e., $\identitymap(\rho)=\rho$ for any state $\rho$.
For two channels $\cE$ and $\cF$ that take $d$ dimensional states, we say $\|\cE-\cF\|_\diamond\coloneqq\max_{\ket{\psi}}\|(\mathrm{id}\otimes\cE)(\ket{\psi}\bra{\psi})-(\mathrm{id}\otimes\cF)(\ket{\psi}\bra{\psi})\|_1$ is the diamond norm between $\cE$ and $\cF$, where the maximization is taken over all $d^2$ dimensional pure states.

The set (or group) of $d$-dimensional unitary matrices and states are denoted by $\Unitaries(d)$ and $\States(d)$.
$\mu_d$ and $\mu_d^s$ denotes the Haar measure over $\Unitaries(d)$ and $\States(d)$, respectively.
$S_k$ denotes the permutation group over $k$ elements. For $\pi\in S_k$ and $\sigma\in S_\ell$, $(\pi,\sigma)\in S_{k+\ell}$ is the permutation that permutates the first $k$ elements with respect to $\pi$ and the last $\ell$ elemensts with respect to $\sigma$.
For $\pi\in S_k$, we define $d^k\times d^k$ permutation unitary $R_\pi$ that satisfies $R_\pi\ket{x_1,...,x_k}=\ket{x_{\pi^{-1}(1)},...,x_{\pi^{-1}(k)}}$ for all $x_1,...,x_k\in[d]$.

\subsection{Haar twirl and Unitary Design}
We define the Haar twirling map.

\begin{definition}\label{def: twirl}
     Let $k,d\in\N$ and $\nu$ be a distribution over $\Unitaries(d)$. We define the $k$-wise twirl with respect to $\nu$ $\cM^{(k)}_{\nu}$ as
    \begin{align}
        \cM^{(k)}_{\nu}(\cdot)&\coloneqq\Exp_{U\gets\nu}U^{\otimes k}(\cdot)U^{\dag\otimes k}.
    \end{align}
    In particular, for the Haar measure $\mu_d$ over $\Unitaries(d)$, we call it the Haar $k$-wise twirl, and write it by $\cM^{(k)}_{\text{\rm Haar}}$
\end{definition}

\begin{definition}[Unitary $k$-Design \cite{mele2024introduction}]\label{def:unitary_1-design}
    Let $k,d\in\N$.
    We say that a distribution $\nu$ over $\Unitaries(d)$ is a unitary $k$-design if the action of the $k$-wise twirl is the same as that of the Haar $k$-wise twirl. Namely,
    for any $d^k$-dimensional state $\rho$,
    \begin{align}
        \cM^{(k)}_{\nu}(\rho)=\cM^{(k)}_{\text{\rm Haar}}(\rho).
    \end{align}
\end{definition}

Note that an action on any unitary $1$-design $\nu$ can be written
\begin{align}
    \cM^{(1)}_\nu=\frac{I}{d}
\end{align}
since $\cM^{(1)}_{\text{\rm Haar}}(\rho)=\frac{I}{d}$ for any state $\rho$.

The following lemma follows from the straightforward calculation.
\begin{lemma}\label{lem:Pauli_is_1-design}
    For any $n\in\N$, the uniform distribution over $n$-qubit Pauli group is unitary $1$-design.
\end{lemma}

\subsection{Useful Lemmas}

\begin{theorem}[Theorem 5.17 in \cite{Mec}]
    \label{thm:Haar_concentration}
Let $\mu_N$ be the Haar measure over $\Unitaries(N)$. Given $N_1,\ldots,N_k \in \N$, let $X = \Unitaries(N_1)\times\cdots\times\Unitaries(N_k)$. Let $\mu = \mu_{N_1} \times \cdots \times \mu_{N_k}$ be the product of Haar measures on $X$. Suppose that $f: X \to \R$ is $L$-Lipschitz in the $\ell^2$-sum of Frobenius norm, i.e., for any $U=(U_1,...,U_k)\in X$ and $V=(V_1,...,V_k)\in X$, 
we have $|f(U)-V(U)|\le L\sqrt{\sum_i\|U_i-V_i\|^2_2}$. Then for every $\delta > 0$,
\begin{align}
\Pr_{U\gets \mu}\left[f(U)\ge\Exp_{V \gets \mu}[f(V)]+\delta\right] \le \exp\left(-\frac{N\delta^2}{24L^2}\right),
\end{align}
where $N \coloneqq \min\{N_1,\ldots,N_k\}$.
\end{theorem}

We use that the probability that an algorithm given access to $U$ and its inverse $U^\dag$ is a Lipschitz function concerning $U$. The case when an algorithm queries only $U$ is shown in \cite{TQC:Kre21}. Since the proof for the case when given access to $U$ and $U^\dag$ is the same, we have the following.

\begin{lemma}[\cite{TQC:Kre21}]\label{lem:Lipschitz}
    Let $\cA^{U,U^\dag}$ be a quantum algorithm that makes $T$ queries to $U\in\Unitaries(N)$ and its inverse. Then, $f(U)=\Pr[1\gets\cA^{U,U^\dag}]$ is $2T$-Lipschitz in the Frobenius norm, i.e., $|f(U)-f(V)|\le2T\|U-V\|_2$ for all $U,V\in\Unitaries(N)$. 
\end{lemma}


\begin{lemma}[Gentle Measurement Lemma \cite{winter1999coding,watrous2018theory}]\label{lem:gentle_measurement}
    Let $\rho$ be a quantum state, and $0\le\epsilon\le1$. Let $M$ be a matrix such that $0\le M\le I$ and
    \begin{align}
        \Tr[M\rho]\ge1-\epsilon.
    \end{align}
    Then,
    \begin{align}
        \bigg\|\rho-\frac{\sqrt{M}\rho\sqrt{M}}{\Tr[M\rho]}\bigg\|_1\le \sqrt{\epsilon}.
    \end{align}
\end{lemma}

\begin{lemma}[Quantum Union Bound \cite{gao2015quantum,OV22}]\label{lem:quantum_union}
    Apply the two-outcome projective $(P_1,I-P_1),...,(P_m,I-P_m)$ sequentially on $\rho$, and let $\rho_m$ be the final state conditioned on the outcome $P_1,...,P_m$ all occurring. 
    Suppose that $\Tr[P_i\rho]\ge 1-\epsilon_i$ for $i=1,...,m$,
    then, it holds that 
    \begin{align}
        \|\rho - \rho_m\|_1 \le \sqrt{\epsilon_1 + ... + \epsilon_m}.
    \end{align}
\end{lemma}
We also use the following lemma.

\begin{lemma}\label{lem:trace_inequality_for_three_matrices}
    Let $A,B$ and $C$ be square matrices of the same size such that 1: $A$ is hermitian, 2: $AB=BA$, 3: $B$ and $C$ are positive. Then,
    \begin{align}
        |\Tr[ABC]|\le\|A\|_\infty\Tr[BC].
    \end{align}
\end{lemma}

\begin{proof}[Proof of \cref{lem:trace_inequality_for_three_matrices}]
    Since $AB=BA$, $A$ and $B$ has the spectral decomposition with the same basis $\{\ket{\psi_i}\}_i$: $A=\sum_ia_i\ket{\psi_i}\bra{\psi_i},B=\sum_ib_i\ket{\psi_i}\bra{\psi_i}$. Note that $b_i\ge0$ for any $i$ since $B$ is positive. Therefore,
    \begin{align}
        |\Tr[ABC]|=|\sum_ia_ib_i\bra{\psi_i}C\ket{\psi_i}|\le\sum_i|a_i|b_i\bra{\psi_i}C\ket{\psi_i}\le(\max_i|a_i|)\sum_ib_i\bra{\psi_i}C\ket{\psi_i}=\|A\|_\infty\Tr[BC],
    \end{align}
    where, in the first inequality, we have used the triangle inequality and $b_i\bra{\psi_i}C\ket{\psi_i}\ge0$ since $b_i\ge0$ and $C$ is positive.
\end{proof}

The following two lemmas follow from the straightforward calculation.

\begin{lemma}\label{lem:1-norm_for_general_hermitian_matrix}
    For any Hermitian matrix $A$,
    \begin{align}
        \|A\|_1=2\max_{M:0\le M\le I}\Tr[MA]-\Tr[A].
    \end{align}
\end{lemma}

\begin{lemma}\label{lem:partial_trace}
    Let $\regA$ and $\regB$ be registers such that the dimension of $\regB$ is larger than that of $\regA$. Then, for any unitary $V$ on $\regA$ and pure state $\ket{\phi}_{\regA,\regB}$ over the registers $\regA$ and $\regB$, there exist a unitary $W$ on $\regB$ and a quantum state $\xi$ on $\regB$ such that
    \begin{align}
        \Tr_{\regA}[(V_\regA\otimes I_\regB)\ket{\phi}\bra{\phi}_{\regA\regB}]=\sqrt{\xi}_\regB W(V^\Gamma\otimes I)W^\dag \sqrt{\xi}_\regB,
    \end{align}
    where $\Gamma$ denotes the transpose with respect to computational basis.
\end{lemma}

For $k,d\in\N$, define 
\begin{align}
    \Pi^{(d,k)}_{\text{sym}}\coloneqq\frac{1}{k!}\sum_{\sigma\in S_k}R_\sigma.
\end{align}
$\Pi^{(d,k)}_{\text{sym}}$ satisfies the following which are
well-known facts. For its detail, see \cite{harrow2013church,mele2024introduction}.

\begin{lemma}\label{lem:symmetric_subspace}
    Let $k,d\in\N$. Then, $\Pi^{(d,k)}_{\text{sym}}$ is the projection. Moreover,
    \begin{align}
        \Tr\Pi^{(d,k)}_{\text{sym}}=\binom{d+k-1}{k}
    \end{align}
    and
    \begin{align}
        \Exp_{\ket{\psi}\gets\mu^{s}_d}(\ket{\psi}\bra{\psi})^{\otimes \ell}=\frac{\Pi^{(d,k)}_{\text{sym}}}{\Tr\Pi^{(d,k)}_{\text{sym}}},
    \end{align}
    where $\mu^{s}_d$ is the Haar measure over all $d$-dimensional states.
\end{lemma}

\begin{lemma}\label{lem:sum_of_R_pi_otimes_R_pi^dag_is_projection}
    Let $k,d\in\N$. Then,
    \begin{align}
        \frac{1}{k!}\sum_{\sigma\in S_k}R_\sigma\otimes R_\sigma=\Pi^{(d^2,k)}_{\text{sym}}.
    \end{align}
\end{lemma}

\begin{proof}[Proof of \cref{lem:sum_of_R_pi_otimes_R_pi^dag_is_projection}]
    Let $\regA\coloneqq\regA_1...\regA_k$ and $\regB\coloneqq\regB_1...\regB_k$, where $\regA_i$ and $\regB_i$ are $d$-dimensional registers for each $i\in[k]$. We define $\regC_i\coloneqq\regA_i\regB_i$ for each $i\in[k]$ and $\regC\coloneqq\regC_1...\regC_k$. Then, $R_{\sigma,\regA}\otimes R_{\sigma,\regB}=R_{\sigma,\regC}$ for any $\sigma\in S_k$. Therefore \cref{lem:sum_of_R_pi_otimes_R_pi^dag_is_projection} follows from \cref{lem:symmetric_subspace}.
\end{proof}

\section{Definition of PRSGs and PRFSGs in QHRO Model}
\label{sec:def_of_QHRO_model}
In this section, we introduce the quantum Haar random oracle (QHRO) model and define pseudorandom state generators (PRSGs) and pseudorandom function-like state generators (PRFSGs) in the QHRO model.
In the QHRO model, any party is given oracle access to a family $\cU=\{U_n\}_{n\in\mathbb{N}}$ of Haar random unitaries, where
$U_n$ is a Haar random unitary acting on $n$ qubits. For simplicity, $U\gets\mu$ denotes $U_\secp\gets\mu_{2^\secp}$ for each $\secp\in\N$. 

We consider the following two different types of oracle accesses:
\begin{itemize}
    \item Any party can query $U$ but cannot query its inverse $U^\dag\coloneqq\{U_n^\dag\}_{n\in\mathbb{N}}$.
    \item Any party can query both $U$ and $U^\dag$.
\end{itemize}
We call the former the inverseless QHRO model, and the latter the invertible QHRO model.



\subsection{PRSGs in the QHRO Model}
The pseudorandom state generators in the plain model were defined in \cite{C:JiLiuSon18}.
Here we define PRSGs in the QHRO model as follows.
\begin{definition}[Pseduorandom States Generators (PRSGs) in the QHRO Model]\label{def:PRSG_in_QHRO}
    We define that an algorithm $G^{(\cdot,\cdot)}$ is a pseudorandom state generator (PRSG) in the QHRO model if it satisfies the following: 
    \begin{itemize}
        \item Efficient generation: Let $\secp\in\N$ be the security parameter and $\cK_\secp$ be a key space over at most $\poly(\secp)$ bits.
        $G^{\cU,\cU^\dag}$ is a QPT algorithm that takes a key $k\in\cK_\secp$ as input, and outputs a quantum state $\ket{\phi_k}$.

        \item Pseudorandomness in the invertible QHRO model: For any polynomial-query adversary $\cA^{(\cdot,\cdot)}$, and any polynomial $t(\secp)$, there exists a negligible function $\negl$ such that
        \begin{align}
            \bigg|\Pr_{\cU\gets\mu,k\gets\cK_\secp}[1\gets\cA^{\cU,\cU^\dag}(\ket{\phi_k}^{\otimes t})]-\Pr_{\cU\gets\mu,\ket{\psi}\gets\mu^{S}_{2^\secp}}[1\gets\cA^{\cU,\cU^\dag}(\ket{\psi}^{\otimes t})]
            \bigg|\le\negl(\secp).
        \end{align}
    \end{itemize}
    If both the generation algorithm $G$ and the adversary $\cA$ are only allowed to query $\cU$ non-adaptively before receiving challenge states, we say it is PRSGs in the \emph{non-adaptive inverseless} QHRO model.
\end{definition}

\subsection{PRFSGs in the QHRO Models}

We give the definition of PRFSGs in the QHRO model and the invertible QHRO model. PRFSGs in the plain model were defined in \cite{C:AnaQiaYue22,TCC:AGQY22}. As a security, we can consider selective security, classically-accessible adaptive security, and quantumly-accessible adaptive security. In this work, we focus on selective security in the QHRO model and classically-accessible adaptive security in the invertible QHRO model.

\begin{definition}[Slectively Secure Pseudorandom Function-like State Generators (PRFSGs) in the QHRO Model]\label{def:PRFSG_in_QHRO}
    We define that an algorithm $G^{(\cdot,\cdot)}$ is a selectively secure pseudorandom function-like state generator (PRFSG) in the QHRO model if it satisfies the following: 
    \begin{itemize}
        \item Efficient generation:  Let $\secp\in\N$ be the security parameter and $\cK_\secp$ be a key space at most $\poly(\secp)$ bits.
        $G^{\cU,\cU^\dag}$ is a QPT algorithm that takes a key $k\in\cK_\secp$ and a bit string $x\in\bit^\secp$ as input, and outputs a quantum state $\ket{\phi_k(x)}$.

        \item Selective security in the invertible QHRO model: For any {\it unbounded} adversary $\cA$, any polynomial $t$, and any bit strings $x_1,...,x_{\ell(\secp)}\in\bit^\secp$ with any polynomial $\ell$,
        \begin{align}
            \bigg|\Pr_{\cU\gets\mu,k\gets\cK_\secp}&[1\gets\cA^{\cU,\cU^\dag}(\ket{\phi_k(x_1)}^{\otimes t},...,\ket{\phi_k(x_\ell)}^{\otimes t})]\\
            &-\Pr_{\cU\gets\mu,\ket{\psi_1},...,\ket{\psi_\ell}\gets\mu^{S}_{2^\secp}}[1\gets\cA^{\cU,\cU^\dag}(\ket{\psi_1}^{\otimes t},...,\ket{\psi_\ell}^{\otimes t})]
            \bigg|\le\negl(\secp).
        \end{align}
    \end{itemize}
\end{definition}
If both the generation algorithm $G$ and the adversary $\cA$ are only allowed to query $\cU$ non-adaptively before receiving challenge states, we say it is PRFSGs in the \emph{non-adaptive inverseless} QHRO model.

\begin{definition}[Classically-accesible Adaptively Secure PRFSGs in the invertible QHRO Model]
    Let $G^{(\cdot)}$ be a QPT algorithm that satisfies the efficient generation property in \cref{def:PRFSG_in_QHRO}. If it satisfies the following, we say it is a classically-accessible adaptive secure PRFSG in the invertible QHRO model.
    \begin{itemize}
        \item Classically-accessible adaptive security in the invertible QHRO model: For any unbounded adversary $\cA^{(\cdot,\cdot,\cdot)}$ that queries each oracle at most $\poly(\secp)$ and can query the first oracle only classically,
        \begin{align}
            \bigg|\Pr_{\cU\gets\mu,k\gets\cK_\secp}[1\gets\cA^{\cO^{\cU,\cU^\dag}_{\text{PRFS}}(k,\cdot),\cU,\cU^\dag}]-\Pr_{\cU\gets\mu,\cO_{\text{Haar}}}[1\gets\cA^{\cO_{\text{Haar}}(k,\cdot),\cU,\cU^\dag}]\bigg|\le\negl(\secp).
        \end{align}
        Here, $\cO^{\cU,\cU^\dag}_{\text{PRFS}}$ and $\cO_{\text{Haar}}$ are defined as follows:
        \begin{itemize}
            \item $\cO^{\cU,\cU^\dag}_{\text{PRFS}}(k,\cdot)$: It takes $x\in\bit^\secp$ as input and outputs $G^{\cU,\cU^\dag}(k,x)=\ket{\phi_k(x)}$.
            \item $\cO_{\text{Haar}}(\cdot)$: It takes $x\in\bit^\secp$ as input and outputs $\ket{\psi_x}$, where $\ket{\psi_x}\gets\mu^{S}_{2^\secp}$ for each $x\in\bit^\secp$.
        \end{itemize}
    \end{itemize}
\end{definition}

\section{Haar Twirl Approximation Formula}\label{sec:approx_formula}

In this section, we derive the Haar twirl approximation formula, \cref{intro_lem:approximation_formula_for_Haar_k-fold},
which plays a crucial role in proving the non-adaptive security of PRFSGs in the QHRO model.
First, we intuitively explain why the approximation formula holds in \cref{subsec:intuition_of_approx_formula}. 
Next, we introduce the Weingarten matrix and the Weingarten function in \cref{subsec:Weingarten_matrix_and_function}, and give some lemmas of the Weingarten function in \cref{subsec:lemmas_of_Wg_function}. 
Finally, we give a proof of the approximation formula in \cref{subsec:proof_of_approx_formula}. 

\subsection{Intuition for Approximation Formula}\label{subsec:intuition_of_approx_formula}

Our goal of this section is to show the following approximation formula, which we call Haar twirl approximation formula: let $\rho$ be a quantum state on the register $\regA\regB$, where the dimension of $\regA$ is $d^k$ and $\regB$ is some fixed register. Then, for the Haar twirl $\cM^{(k)}_{\text{\rm Haar}}(\cdot) = \Exp_{U\gets \mu_d} U^{\otimes k} (\cdot) U^{\deg \otimes k}$ defined in \cref{def: twirl},
\begin{align}
    \left\|(\cM^{(k)}_{\text{\rm Haar},\regA}\otimes\identitymap_\regB)(\rho_{\regA\regB})-
    \sum_{\pi\in S_k}\frac{1}{d^k}R_{\pi,\regA}\otimes \Tr_\regA[(R_{\pi,\regA}^\dag\otimes I_\regB)\rho_{\regA\regB}]
    \right\|_1
    \le O\bigg(\frac{k^2}{d}\bigg),\label{eq:eq1_proof_sketch_of_approx_formula}
\end{align}
where $S_k$ is the set of all permutations over $k$ elements, and $R_\pi$ is the permutation unitary that acts $R_\pi\ket{x_1,...,x_k}=\ket{x_{\pi^{-1}(1)},...,x_{\pi^{-1}(k)}}$ for all $x_1,...,x_k\in[d]$.

Intuitively, the above formula is derived as follows; first, from Weingarten calculus \cite{collins2006integration},
\begin{align}
    (\cM^{(k)}_{\text{\rm Haar},\regA}\otimes\identitymap_\regB)(\rho_{\regA\regB})=\sum_{\sigma,\tau\in S_k}\Wg(\tau\sigma^{-1};d)R_{\sigma\regA}^\dag\otimes \Tr_\regA[\rho_{\regA\regB}(R_{\tau\regA}\otimes I_\regB)],\label{eq:eq2_proof_sketch_of_approx_formula}
\end{align}
where $\Wg(~\cdot~;d)$ is called the Weingarten function that maps an element of $S_k$ to a real number.\footnote{$\Wg$ depends on $k$, but for simplicitly we omit $k$ here.} We give its definition in \cref{subsec:Weingarten_matrix_and_function}. The Weingarten function has the following nice property\footnote{For the case when $\pi$ is the identity, see \cref{lem:approximation_formula_for_Weingarten_function}. For other cases, we do not use it explicitly but it is shown in \cite{collins2017weingarten}.};
\begin{align}
    \Wg(\pi;d)\approx
    \begin{cases}
        d^{-k} \:&\text{if }\pi\text{ is the identity,}\\
        O(d^{-k-1})\:&\text{otherwise.}
    \end{cases}
\end{align}
Therefore, if we ignore all terms such that $\tau\neq\sigma$ in \cref{eq:eq2_proof_sketch_of_approx_formula}, \cref{eq:eq1_proof_sketch_of_approx_formula} seems to hold. We show this formally in \cref{subsec:proof_of_approx_formula}.

\subsection{Weingerten Calculus}\label{subsec:Weingarten_matrix_and_function}
In this subsection, we review Weingarten calculus~\cite{collins2022weingarten}. 

We first introduce the Weingarten function as follows.
Let us assume $k$ and $d$ are positive integers such that $k\le d$. Recall that $S_k$ is the permutation group over $k$ elements, and, for $\pi\in S_k$, $R_\pi$ is $d^k\times d^k$ permutation unitary that satisfies $R_\pi\ket{x_1,...,x_k}=\ket{x_{\pi^{-1}(1)},...,x_{\pi^{-1}(k)}}$ for all $x_1,...,x_k\in[d]$.
We define $k!\times k!$ matrix $G(d)$ whose matrix elements are specified by two permutations $\sigma,\tau\in S_k$ such that
\begin{align}
    G(d)_{\sigma,\tau}\coloneqq\Tr[R_{\tau\sigma^{-1}}]=d^{k-|\tau\sigma^{-1}|}.
\end{align}
Here, for $\pi\in S_k$, $|\pi|$ is defined by the minimum number of transpositions to represent $\pi$ as a product of those transpositions. Note that $G(d)_{\sigma\pi,\tau\pi}=G(d)_{\sigma,\tau}$ for any $\pi,\sigma,\tau\in S_k$. Let $\Wg(d)$, which is called the Weingarten matrix, be a $k!\times k!$ matrix as the pseudo-inverse of $G(d)$.
We define the Weingarten function $\Wg(\cdot;d):S_k\to\R$ such that
\begin{align}
    \Wg(\sigma\tau^{-1};d)\coloneqq\Wg(d)_{\sigma,\tau}.
\end{align}
This is well-defined since $\Wg(d)_{\sigma,\tau}=\Wg(d)_{\sigma\pi,\tau\pi}$ for any $\pi,\sigma,\tau\in S_k$, where it follows from $G(d)_{\sigma\pi,\tau\pi}=G(d)_{\sigma,\tau}$.

Weingarten calculus is the following lemma:
\begin{lemma}[Corollary 2.4 of \cite{collins2006integration}]\label{lem:Weingarten_calculas}
    Let $k,d\in\N$. Let $\mu_d$ be the Haar measure over $\Unitaries(d)$ and $S_k$ be the set of all permutations over $[k]$. Let $i\coloneqq(i_1,...,i_k),j\coloneqq(j_1,...,j_k),i'\coloneqq(i'_1,...,i'_k),j'\coloneqq(j'_1,...,j'_k)\in[d]^k$. Then,
    \begin{align}
        \Exp_{U\gets\mu_d}U_{i_1j_1}...U_{i_kj_k}\overline{U}_{i'_1j'_1}...\overline{U}_{i'_kj'_k}=\sum_{\sigma,\tau\in S_k}\delta_{i,\sigma(i')}\delta_{j,\tau(j')}\Wg(\tau\sigma^{-1};d),
    \end{align}
    where,
    for $\ell\coloneqq(\ell_1,...,\ell_k),\in[d]^k$ and $\pi\in S_k$, $\pi(\ell)\coloneqq(\ell_{\pi(1)},...,\ell_{\pi(k)})$.
\end{lemma}

From \cref{lem:Weingarten_calculas} and the straightforward calculation, we have the following lemma.

\begin{lemma}\label{lem:Weingarten_calculas_of_matrix_form}
    Let $k,d\in\N$. Let $\regA$ denote the $d^k$-dimensional register, and $\regB$ denote any dimensional register.
    Let $M_{\regA\regB}$ be a matrix. Then, 
    \begin{align}
        (\cM^{(k)}_{\text{\rm Haar},\regA}\otimes\identitymap_\regB)(M_{\regA\regB})=\sum_{\sigma,\tau\in S_k}\Wg(\tau\sigma^{-1};d)R_{\sigma\regA}^\dag\otimes\Tr_{\regA'}[M_{\regA\regB}(R_{\tau\regA}\otimes I_\regB)].
    \end{align}
\end{lemma}

\begin{proof}[Proof of \cref{lem:Weingarten_calculas_of_matrix_form}]
    For any $d^k\times d^k$ matrix $N$, we have 
    \begin{align}
        \cM^{(k)}_{\text{\rm Haar}}(N)=\sum_{\sigma,\tau\in S_k}\Wg(\tau\sigma^{-1};d)\Tr[NR_\tau]R_\sigma^\dag.\label{eq:Weingarten_calculas_of_matrix_form}
    \end{align}
    We give its proof later. From \cref{eq:Weingarten_calculas_of_matrix_form}, 
    \begin{align}
        (\cM^{(k)}_{\text{\rm Haar},\regA}\otimes\identitymap_\regB)(M_{\regA\regB})
        =&\bigg(I_\regA\otimes\sum_{i}\ket{i}\bra{i}_\regB\bigg)
        (\cM^{(k)}_{\text{\rm Haar},\regA}\otimes\identitymap_\regB)(M_{\regA\regB})\bigg(I_\regA\otimes\sum_{j}\ket{j}\bra{j}_\regB\bigg)
        \\
        =&\sum_{i,j}\cM^{(k)}_{\text{\rm Haar},\regA}(M^{(i,j)}_\regA)\otimes\ket{i}\bra{j}_\regB\\
        =&\sum_{i,j}\sum_{\sigma,\tau\in S_k}\Wg(\tau\sigma^{-1};d)\Tr[M^{(i,j)}R_\tau]R_{\sigma\regA}^\dag\otimes\ket{i}\bra{j}_\regB\\
        =&\sum_{\sigma,\tau\in S_k}\Wg(\tau\sigma^{-1};d)R_{\sigma\regA}^\dag\otimes\bigg(\sum_{i,j}\Tr[M^{(i,j)}R_\tau]\ket{i}\bra{j}_\regB\bigg),
    \end{align}
    where $M^{(i,j)}_\regA\coloneqq(I_\regA\otimes\bra{i}_\regB)M_{\regA\regB}(I_\regA\otimes\ket{j}_\regB)$. From the standard calculation, we have
    \begin{align}
        \sum_{i,j}\Tr[M^{(i,j)}R_\tau]\ket{i}\bra{j}_\regB=\Tr_{\regA'}[M_{\regA\regB}(R_{\tau\regA}\otimes I_\regB)].
    \end{align}
    Thus, we have
    \begin{align}
        (\cM^{(k)}_{\text{\rm Haar},\regA}\otimes\identitymap_\regB)(M_{\regA\regB})=&\sum_{\sigma,\tau\in S_k}\Wg(\tau\sigma^{-1};d)R_{\sigma\regA}^\dag\otimes \Tr_{\regA'}[M_{\regA\regB}(R_{\tau\regA}\otimes I_\regB)].
    \end{align}
    
    To conclude the proof, we show \cref{eq:Weingarten_calculas_of_matrix_form}. From \cref{lem:Weingarten_calculas}, we have
    \begin{align}
        \cM^{(k)}_{\text{\rm Haar}}(N)
        &=\sum_{i,j\in[d]^k}\bigg(\cM^{(k)}_{\text{\rm Haar}}(N)\bigg)_{i,j}\ket{i}\bra{j}\\
        &=\sum_{i,j,\ell,m\in[d]^k}\Exp_{U\gets\mu_d}U_{i_1\ell_1}...U_{i_k\ell_k}N_{\ell,m}\overline{U}_{j_1m_1}...\overline{U}_{j_km_k}\ket{i}\bra{j}\\
        &=\sum_{i,j,\ell,m\in[d]^k}\sum_{\sigma,\tau\in S_k}N_{\ell,m}\delta_{i,\sigma(j)}\delta_{\ell,\tau(m)}\Wg(\tau\sigma^{-1};d)\ket{i}\bra{j}\\
        &=\sum_{\sigma,\tau\in S_k}\Wg(\tau\sigma^{-1};d)\left(\sum_{\ell,m\in[d]^k}N_{\ell,m}\delta_{\ell,\tau(m)}\right)\left(\sum_{i,j\in[d]^k}\delta_{i,\sigma(j)}\ket{i}\bra{j}\right).
    \end{align}
    From the definition of $R_\sigma$,
    \begin{align}
        \sum_{i,j\in[d]^k}\delta_{i,\sigma(j)}\ket{i}\bra{j}
        =\sum_{j\in[d]^k}\ket{\sigma(j)}\bra{j}=R_\sigma^\dag.
    \end{align}
    On the other hand,
    \begin{align}
        \sum_{\ell,m\in[d]^k}N_{\ell,m}\delta_{\ell,\tau(m)}=\sum_{\ell,m\in[d]^k}N_{\ell,m}(R_\tau)_{m,\ell}=\Tr[NR_\tau].
    \end{align}
    From the above three equations, \cref{eq:Weingarten_calculas_of_matrix_form} follows.
\end{proof}

\subsection{Useful Lemmas of Weingarten Function}
\label{subsec:lemmas_of_Wg_function}
We use some properties of the Weingarten function. First, the following lemma follows from the fact that the Weingarten matrix is the pseudo-inverse of the Gram matrix.

\begin{lemma}\label{lem:Wg_is_invariant_under_conjucation}
    Let $k,d\in\N$. Then, for any $\pi,\sigma\in S_k$, $\Wg(\pi;d)=\Wg(\pi^{-1};d)$ and $\Wg(\pi\sigma\pi^{-1};d)=\Wg(\sigma;d)$.
\end{lemma}

\begin{proof}[Proof of \cref{lem:Wg_is_invariant_under_conjucation}]
    First, let us show the former.
    Since the Gram matrix is symmetric by its definition, the Weingartne matrix $\Wg(d)$ is also symmetric since it is the pseudoinverse of the Garm matrix. Thus, we have $\Wg(\pi;d)=\Wg(d)_{\pi,e}=\Wg(d)_{e,\pi}=\Wg(\pi^{-1};d)$.
    
    Next, let us prove the latter.
    Since $\Wg(\sigma\tau^{-1};d)=\Wg(d)_{\sigma,\tau}$, it suffices to show $\Wg(d)_{\pi\sigma,\pi\tau}=\Wg(d)_{\sigma,\tau}$ for any $\pi,\sigma,\tau\in S_k$. 
    For $\pi\in S_k$, let us define $k!\times k!$ matrix $\Wg^{(\pi)}(d)$ such that $\Wg^{(\pi)}(d)_{\sigma,\tau}\coloneqq\Wg(d)_{\pi\sigma,\pi\tau}$. Since $G(d)_{\pi\rho,\pi\tau}=d^{k-|\pi\rho\tau^{-1}\pi^{-1}|}=d^{k-|\sigma\tau^{-1}|}=G(d)_{\sigma,\tau}$, $\Wg^{(\pi)}(d)$ is also the pseudo-inverse of $G(d)$. From the uniqueness of the pseudo-inverse matrix, we have $\Wg^{(\pi)}(d)=\Wg(d)$ for any $\pi\in S_k$, which implies $\Wg(d)_{\pi\sigma,\pi\tau}=\Wg(d)_{\sigma,\tau}$ for any $\pi,\sigma,\tau\in S_k$.
\end{proof}

The following are lemmas about a summation of the Weingarten function.
The first lemma is shown in section 3.1.1. of \cite{collins2012integration}.

\begin{lemma}[\cite{collins2012integration}]\label{lem:sum_of_Weingarten_function}
    \begin{align}
        \sum_{\pi\in S_k}\Wg(\pi;d)=\frac{1}{d(d+1)\cdots(d+k-1)}.
    \end{align}
\end{lemma}

\begin{lemma}[Lemma 6 in \cite{aharonov2022quantum}]\label{lem:sum_of_absolute_value_of_Weingarten_function}
    \begin{align}
        \sum_{\pi\in S_k}|\Wg(\pi;d)|=\frac{1}{d(d-1)\cdots(d-k+1)}.
    \end{align}
\end{lemma}

The following lemma is a corollary of Theorem 3.2 in \cite{collins2017weingarten}.

\begin{lemma}[\cite{collins2017weingarten}]\label{lem:approximation_formula_for_Weingarten_function}
    Let $k,d\in\N$ such that $d>\sqrt{6}k^{7/4}$. Then,
    \begin{align}
        \frac{1}{d^k}\bigg(1-O\bigg(\frac{k}{d^2}\bigg)\bigg)
        \le\Wg(e;d)
        \le\frac{1}{d^k}\bigg(1-O\bigg(\frac{k^{7/2}}{d^2}\bigg)\bigg).
    \end{align}
\end{lemma}

\subsection{Proof of Haar Twirl Approximation Formula}\label{subsec:proof_of_approx_formula}
Now we are ready to prove the approximation formula.

\begin{lemma}\label{lem:approximation_formula_for_Haar_k-fold}
   Let $k,d\in\N$ such that $d>\sqrt{6}k^{7/4}$. Let $\regA$ be a $d^k$-dimentional register, and $\regB$ be some fixed register. Then, for any quantum state $\rho$ on the registers $\regA\regB$,
    \begin{align}
        \left\|
        (\cM^{(k)}_{\text{\rm Haar},\regA}\otimes\identitymap_\regB)(\rho_{\regA\regB})-\sum_{\sigma\in S_k}\frac{1}{d^k}R^\dag_{\sigma,\regA}\otimes\Tr_{\regA}[(R_{\sigma,\regA}\otimes I_\regB)\rho_{\regA\regB}]
        \right\|_1
        \le O\left(\frac{k^2}{d}\right).
    \end{align}
\end{lemma}

\begin{proof}[Proof of \cref{lem:approximation_formula_for_Haar_k-fold}]
    From the concavity of $1$-norm, it suffices to show the case when $\rho$ is a pure state $\ket{\psi}\bra{\psi}$. In the following, we often write $\ket{\psi}\bra{\psi}$ just as $\psi$ for the notational simplicity. 

    It is clear that both matrices are hermitian.
    Thus, from \cref{lem:1-norm_for_general_hermitian_matrix},
    \begin{align}
        &\left\|
        (\cM^{(k)}_{\text{\rm Haar},\regA}\otimes\identitymap_\regB)(\psi_{\regA\regB})-\sum_{\sigma\in S_k}\frac{1}{d^k}R_{\sigma,\regA}^\dag\otimes\Tr_{\regA'}[(R_{\sigma,\regA'}\otimes I_\regB)\psi_{\regA'\regB}]
        \right\|_1\\
        =&2\max_{M:0\le M\le I}\Tr\left[M\bigg(
        (\cM^{(k)}_{\text{\rm Haar},\regA}\otimes\identitymap_\regB)(\psi_{\regA\regB})-\sum_{\sigma\in S_k}\frac{1}{d^k}R_{\sigma,\regA}^\dag\otimes\Tr_{\regA'}[(R_{\sigma,\regA'}\otimes I_\regB)\psi_{\regA'\regB}]\bigg)
        \right] \notag
        \\
        &-\Tr\left[
        (\cM^{(k)}_{\text{\rm Haar},\regA}\otimes\identitymap_\regB)(\psi_{\regA\regB})-\sum_{\pi\in S_k}\frac{1}{d^k}R_{\sigma,\regA}^\dag\otimes\Tr_{\regA'}[(R_{\sigma,\regA'}\otimes I_\regB)\psi_{\regA'\regB}]
        \right].\label{eq:approx_formula}
    \end{align}
    Thus, it suffices to show that both the first term and the second term are at most $O(k^2/d)$.
    
    \paragraph{Estimation of the first term in \cref{eq:approx_formula}.}
    To show the first term is at most $O(k^2/d)$, we show
    \begin{align}
        \Bigg|\Tr\bigg[M_{\regA\regB}\bigg((\cM^{(k)}_{\text{\rm Haar},\regA}\otimes\identitymap_\regB)(\psi_{\regA\regB})-\sum_{\sigma\in S_k}\frac{1}{d^k}R_{\sigma,\regA}^\dag\otimes\Tr_{\regA'}[(R_{\sigma,\regA'}\otimes I_\regB)\psi_{\regA'\regB}]\bigg)
        \bigg]\Bigg|
        \le O\bigg(\frac{k^2}{d}\bigg)
    \end{align}
    for any $0\le M\le I$. Without loss of generality, we can assume the dimension of $\regB$ is larger than that of $\regA$. Otherwise, we add some register $\regC$ such that the dimension of $\regB\regC$ is larger than that of $\regA$, and consider $M'_{\regA\regB\regC}\coloneqq M_{\regA\regB}\otimes I_\regC$ and $\ket{\psi'}_{\regA\regB\regC}\coloneqq\ket{\psi}_{\regA\regB}\ket{0...0}_\regC$. Thus, for any permutation $\sigma\in S_k$, there exists a quantum state $\xi$ such that
    \begin{align}
        \Tr_{\regA'}[(R_{\sigma,\regA'}\otimes I_\regB)\psi_{\regA'\regB}]=\sqrt{\xi}_\regB R_{\sigma,\regB}^\Gamma\sqrt{\xi}_\regB=\sqrt{\xi}_\regB R_{\sigma,\regB}^\dag\sqrt{\xi}_\regB\label{eq:eq1_approx_formula}
    \end{align}
    from \cref{lem:partial_trace}. Here, for simplicity, we write that $R_{\sigma,\regB}$ is a unitary which acts as $R_\sigma$ on the subregister of $\regB$ whose dimension is the same as that of $\regA$, and as the identity on the residual subregister of $\regB$. By using this, we have
    \begin{align}
        \sum_{\sigma\in S_k}\frac{1}{d^k}R_{\sigma,\regA}^\dag\otimes\Tr_{\regA'}[(R_{\sigma,\regA'}\otimes I_\regB)\psi_{\regA'\regB}]=\sum_{\sigma\in S_k}\frac{1}{d^k}R_{\sigma,\regA}^\dag\otimes\sqrt{\xi}_\regB R_{\sigma,\regB}^\dag\sqrt{\xi}_\regB.\label{eq:eq2_approx_formula}
    \end{align}
    By combing \cref{eq:eq1_approx_formula} and \cref{lem:Weingarten_calculas_of_matrix_form},  the $k$-fold Haar twirl can be rewritten as follows;
    \begin{align}
        (\cM^{(k)}_{\text{\rm Haar},\regA}\otimes\identitymap_\regB)(\psi_{\regA\regB})
        =&\sum_{\sigma,\tau\in S_k}\Wg(\tau\sigma^{-1};d)R_{\sigma,\regA}^\dag\otimes\Tr_{\regA'}[(R_{\tau,\regA'}\otimes I_\regB)\psi_{\regA'\regB}]\\
        =&\sum_{\sigma,\pi\in S_k}\Wg(\pi;d)R_{\sigma,\regA}^\dag\otimes\Tr_{\regA'}[(R_{\pi\sigma,\regA'}\otimes I_\regB)\psi_{\regA'\regB}]\\
        =&\sum_{\sigma,\pi\in S_k}\Wg(\pi;d)R_{\sigma,\regA}^\dag\otimes\sqrt{\xi}_\regB R_{\pi\sigma,\regB}^\dag\sqrt{\xi}_\regB,\label{eq:eq3_approx_formula}
    \end{align}
    where we replaced the summention of $\tau$ with $\pi$ that satisfies $\tau=\pi\sigma$. Then, we have
    \begin{align}
        &\Bigg|\Tr\bigg[M_{\regA\regB}\bigg((\cM^{(k)}_{\text{\rm Haar},\regA}\otimes\identitymap_\regB)(\psi_{\regA\regB})
        -\sum_{\sigma\in S_k}\frac{1}{d^k}R_{\sigma,\regA}^\dag\otimes\Tr_{\regA'}[(R_{\sigma,\regA'}\otimes I_\regB)\psi_{\regA'\regB}]\bigg)\bigg]\Bigg|\\
        =&\Bigg|\Tr\bigg[M_{\regA\regB}\bigg(\bigg(\Wg(e;d)-\frac{1}{d^k}\bigg)\sum_{\sigma\in S_k}R_{\sigma,\regA}^\dag\otimes \sqrt{\xi}_\regB R_{\sigma,\regB}^\dag\sqrt{\xi}_\regB
        +\sum_{\substack{\sigma,\pi\in S_k,\\
        \pi\neq e}}\Wg(\pi;d)R_{\sigma,\regA}^\dag\otimes \sqrt{\xi}_\regB R_{\pi\sigma,\regB}^\dag\sqrt{\xi}_\regB\bigg)\bigg]\Bigg|\\
        \le&\Bigg|\Wg(e;d)-\frac{1}{d^k}\Bigg|
        \Bigg|\Tr\bigg[M_{\regA\regB}\sum_{\sigma\in S_k}R_{\sigma,\regA}^\dag\otimes\sqrt{\xi}_\regB R_{\sigma,\regB}^\dag\sqrt{\xi}_\regB\bigg]\Bigg|
        +\Bigg|\Tr\bigg[M_{\regA\regB}\sum_{\substack{\sigma,\pi\in S_k,\\
        \pi\neq e}}\Wg(\pi;d)R_{\sigma,\regA}^\dag\otimes \sqrt{\xi}_\regB R_{\pi\sigma,\regB}^\dag\sqrt{\xi}_\regB\bigg]\Bigg|,\label{eq:eq4_approx_formula}
    \end{align}
    where we have used \cref{eq:eq1_approx_formula,eq:eq2_approx_formula} in the equality, and the inequality follows from the triangle inequality. In the following, we show both the first term and the second term are at most $O(k^2/d)$.

    The trace of the first term in \cref{eq:eq4_approx_formula} can be estimated as follows:
    \begin{align}
        \Bigg|\Tr\bigg[M_{\regA\regB}\bigg(\sum_{\sigma\in S_k}R_{\sigma,\regA}^\dag\otimes\sqrt{\xi}_\regB R_{\sigma,\regB}^\dag\sqrt{\xi}_\regB\bigg)\bigg]\Bigg|
        \le&\Bigg|\Tr\bigg[\sum_{\sigma\in S_k}R_{\sigma,\regA}^\dag\otimes\sqrt{\xi}_\regB R_{\sigma,\regB}^\dag\sqrt{\xi}_\regB\bigg]\Bigg|\\
        \le&\sum_{\sigma\in S_k}\bigg|\Tr[R_{\sigma,\regA}^\dag]\Tr[\sqrt{\xi}_\regB R_{\sigma,\regB}\sqrt{\xi}_\regB]\bigg|\\
        \le&\sum_{\sigma\in S_k}\bigg|\Tr[R_{\sigma,\regA}^\dag]\bigg|\\
        =&k!\binom{d+k-1}{k},\label{eq:eq5_approx_formula}
    \end{align}
    where we have used
    \begin{itemize}
        \item the facts that $0\le M\le I$ and $\sum_{\sigma\in S_k}R_{\sigma,\regA}^\dag\otimes\sqrt{\xi}_\regB R_{\sigma,\regB}^\dag\sqrt{\xi}_\regB$ is positive in the first ineqaulity. The latter follows from \cref{lem:sum_of_R_pi_otimes_R_pi^dag_is_projection};
        
        \item the riangle inequality and $\Tr[A]\Tr[B]=\Tr[A\otimes B]$ for any matrix $A$ and $B$ in the second equality;

        \item $|\Tr[\sqrt{\xi}_\regB R_{\sigma,\regB}^\dag\sqrt{\xi}_\regB]|=|\Tr[\xi_\regB R_{\sigma,\regB}]|\le 1$ for any $\sigma\in S_k$ since $\xi$ is a quantum state and $R_\sigma$ is unitary in the third inequality;

        \item $\Tr[R_\sigma^\dag]=\Tr[R_\sigma]\ge0$ for any $\sigma\in S_k$ and \cref{lem:symmetric_subspace} in the last equality.
    \end{itemize}
    By combing \cref{eq:eq5_approx_formula} and \cref{lem:approximation_formula_for_Weingarten_function}, we have\footnote{$\frac{k^{7/2}}{d^2}=\frac{k^2}{d}\frac{k^{3/2}}{d}\le\frac{k^2}{d}$ since $d>\sqrt{6}k^{7/4}$.}
    \begin{align}
        \bigg|\Wg(e;d)-\frac{1}{d^k}\bigg|\Bigg|\Tr\bigg[M\bigg(\sum_{\sigma\in S_k}R_{\sigma,\regA}^\dag\otimes\sqrt{\xi}_\regB R_{\sigma,\regB}^\dag\sqrt{\xi}_\regB\bigg)\bigg]\Bigg|
        \le&\bigg|\Wg(e;d)-\frac{1}{d^k}\bigg|k!\binom{d+k-1}{k}\\
        \le&O\left(\frac{k^{7/2}}{d^{k+2}}\right)k!\binom{d+k-1}{k}\\
        =&O\left(\frac{k^{7/2}}{d^{2}}\right)\\
        \le&O\left(\frac{k^2}{d}\right).\label{eq:eq6_in_approx_formula}
    \end{align}

    Next, let us estimate the second term in \cref{eq:eq4_approx_formula};
    \begin{align}
        &\Bigg|\Tr\bigg[M_{\regA\regB}\sum_{\substack{\sigma,\pi\in S_k,\\
        \pi\neq e}}\Wg(\pi;d)R_{\sigma,\regA}^\dag\otimes \sqrt{\xi}_\regB R_{\pi\sigma,\regB}^\dag\sqrt{\xi}_\regB\bigg]\Bigg|\\
        =&\Bigg|\Tr\bigg[\sum_{\sigma,\pi\in S_k;
        \pi\neq e}\Wg(\pi;d)(R_{\sigma,\regA}^\dag\otimes  R_{\pi,\regB}^\dag R_{\sigma,\regB}^\dag )(I_\regA\otimes \sqrt{\xi}_\regB)M_{\regA\regB}(I_\regA\otimes \sqrt{\xi}_\regB)\bigg]\Bigg|\\
        =&\Bigg|\Tr\bigg[\bigg(\sum_{\pi\in S_k;\pi\neq e}\Wg(\pi;d)I_\regA\otimes R_{\pi,\regB}\bigg)\bigg(\sum_{\sigma\in S_k}R_{\sigma,\regA}\otimes  R_{\sigma,\regB}\bigg)(I_\regA\otimes \sqrt{\xi}_\regB)M_{\regA\regB}(I_\regA\otimes \sqrt{\xi}_\regB)\bigg]\Bigg|.\label{eq:eq7_approx_formula}
    \end{align}
    Here, 
    \begin{itemize}
        \item we have used $R_{\pi\sigma}=R_\sigma R_\pi$ in the first equality;
        
        \item we replaced the summation of $\pi$ with $\pi^{-1}$ and that of $\sigma$ with $\sigma^{-1}$ in the second equality.
    \end{itemize}
    We want to apply \cref{lem:trace_inequality_for_three_matrices} to \cref{eq:eq7_approx_formula}.
    Note that
    \begin{itemize}
        \item $\sum_{\pi\in S_k;\pi\neq e}\Wg(\pi;d)I_\regA\otimes R_{\pi,\regB}$ is hermitian:
        \begin{align}
            \bigg(\sum_{\pi\in S_k;\pi\neq e}\Wg(\pi;d)I_\regA\otimes R_{\pi,\regB}\bigg)^\dag
            &=\sum_{\pi\in S_k;\pi\neq e}\Wg(\pi;d)I_\regA\otimes R_{\pi^{-1},\regB}\\
            &=\sum_{\pi\in S_k;\pi\neq e}\Wg(\pi^{-1};d)I_\regA\otimes R_{\pi,\regB}\\
            &=\sum_{\pi\in S_k;\pi\neq e}\Wg(\pi;d)I_\regA\otimes R_{\pi',\regB},
        \end{align}
        where we have replaced the summation of $\pi$ with $\pi^{-1}$ in the second equality, and we have used \cref{lem:Wg_is_invariant_under_conjucation} in the last equality;

        \item $\sum_{\sigma\in S_k}R_{\sigma,\regA}\otimes R_{\sigma,\regB}$ and $(I_\regA\otimes \sqrt{\xi}_\regB)M_{\regA\regB}(I_\regA\otimes \sqrt{\xi}_\regB)$ are positive, where the former follows from \cref{lem:sum_of_R_pi_otimes_R_pi^dag_is_projection};

        \item $\sum_{\pi\in S_k;\pi\neq e}\Wg(\pi;d)I_\regA\otimes R_{\pi,\regB}$ is commutive with $\sum_{\sigma\in S_k}R_{\sigma,\regA}\otimes R_{\sigma,\regB}$ as follows;
    \begin{align}
        &\bigg(\sum_{\pi\in S_k;\pi\neq e}\Wg(\pi;d)I_\regA\otimes R_{\pi,\regB}\bigg)\bigg(\sum_{\sigma\in S_k}R_{\sigma,\regA}\otimes R_{\sigma,\regB}\bigg)\\
        =&\sum_{\sigma\in S_k}(R_{\sigma,\regA}\otimes R_{\sigma,\regB})\bigg(\sum_{\pi\in S_k;\pi\neq e}\Wg(\pi;d)I_\regA\otimes R_{\sigma}^\dag R_\pi R_{\sigma,\regB}\bigg)\\
        =&\sum_{\sigma\in S_k}(R_{\sigma,\regA}\otimes R_{\sigma,\regB})\bigg(\sum_{\pi\in S_k;\pi\neq e}\Wg(\pi;d)I_\regA\otimes R_{\sigma\pi\sigma^{-1},\regB}\bigg)\\
        =&\sum_{\sigma\in S_k}(R_{\sigma,\regA}\otimes R_{\sigma,\regB})\bigg(\sum_{\pi'\in S_k;\pi'\neq e}\Wg(\sigma^{-1}\pi'\sigma;d)I_\regA\otimes R_{\pi',\regB}\bigg)\\
        =&\bigg(\sum_{\sigma\in S_k}R_{\sigma,\regA}\otimes R_{\sigma,\regB}\bigg)\bigg(\sum_{\pi'\in S_k;\pi'\neq e}\Wg(\pi';d)I_\regA\otimes R_{\pi',\regB}\bigg),
    \end{align}
    where
    \begin{itemize}
        \item we replaced the summation of $\pi$ with $\pi'\coloneqq\sigma\pi\sigma^{-1}$ in the third equality;

        \item we have used \cref{lem:Wg_is_invariant_under_conjucation} in the last equality.
    \end{itemize}
    \end{itemize}
    Thus, we can apply \cref{lem:trace_inequality_for_three_matrices} to \cref{eq:eq7_approx_formula};
    \begin{align}
        &\Bigg|\Tr\bigg[M_{\regA\regB}\sum_{\substack{\sigma,\pi\in S_k,\\
        \pi\neq e}}\Wg(\pi;d)R_{\sigma,\regA}^\dag\otimes \sqrt{\xi}_\regB R_{\pi\sigma,\regB}^\dag\sqrt{\xi}_\regB\bigg]\Bigg|\\
        \le&\bigg\|\sum_{\pi\in S_k;\pi\neq e}\Wg(\pi;d)I_\regA\otimes R_{\pi,\regB}\bigg\|_\infty
        \Tr\bigg[\bigg(\sum_{\sigma\in S_k}R_{\sigma,\regA}\otimes  R_{\sigma,\regB}\bigg)(I_\regA\otimes \sqrt{\xi}_\regB)M_{\regA\regB}(I_\regA\otimes \sqrt{\xi}_\regB)\bigg]\\
        =&\bigg\|\sum_{\pi\in S_k;\pi\neq e}\Wg(\pi;d)I_\regA\otimes R_{\pi,\regB}\bigg\|_\infty
        \Tr\bigg[M_{\regA\regB}\bigg(\sum_{\sigma\in S_k}R_{\sigma,\regA}\otimes  \sqrt{\xi}_\regB R_{\sigma,\regB}\sqrt{\xi}_\regB\bigg)\bigg].\label{eq:eq8_approx_formula}
    \end{align}
    We have already estimated the trace of \cref{eq:eq8_approx_formula} in \cref{eq:eq5_approx_formula}. Hence, it suffices to estimate the operator norm in \cref{eq:eq8_approx_formula};
    \begin{align}
        \bigg\|\sum_{\pi\in S_k;\pi\neq e}\Wg(\pi;d)I_\regA\otimes R_{\pi,\regB}\bigg\|_\infty=&\bigg\|\sum_{\pi\in S_k;\pi\neq e}\Wg(\pi;d)R_{\pi,\regB}\bigg\|_\infty\\
        \le&\sum_{\pi\in S_k;\pi\neq e}|\Wg(\pi;d)|\\
        =&\frac{1}{d(d-1)\cdots(d-k+1)}-|\Wg(e;d)|\\
        \le&\frac{1}{d^k}\left(1+O\left(\frac{k^2}{d}\right)\right)-\frac{1}{d^k}\left(1-O\left(\frac{k}{d^2}\right)\right)\\
        =&O\left(\frac{k^2}{d^{k+1}}\right),\label{eq:eq9_approx_formula}
    \end{align}
    where we have used
    \begin{itemize}
        \item $\|A\otimes B\|_\infty=\|A\|_\infty\|B\|_\infty$ for any matrix $A$ and $B$ in the first equality;
        \item triangle inequality and $\|R_\pi\|_\infty\le1$ for any $\pi\in S_k$ in the first inequality;
        \item \cref{lem:sum_of_absolute_value_of_Weingarten_function} in the second equality, and \cref{lem:approximation_formula_for_Weingarten_function} in the last inequality.
    \end{itemize}
    From \cref{eq:eq5_approx_formula,eq:eq8_approx_formula,eq:eq9_approx_formula}, we can bund the second term in \cref{eq:eq4_approx_formula} as follows;
    \begin{align}
        \Bigg|\Tr\bigg[M_{\regA\regB}\sum_{\substack{\sigma,\pi\in S_k,\\
        \pi\neq e}}\Wg(\pi;d)R_{\sigma,\regA}^\dag\otimes \sqrt{\xi}_\regB R_{\pi\sigma,\regB}^\dag\sqrt{\xi}_\regB\bigg]\Bigg|
        \le&k!\binom{d+k-1}{k}O\left(\frac{k^2}{d^{k+1}}\right)\le O\left(\frac{k^2}{d}\right).\label{eq:eq10_approx_formula}
    \end{align}
    Therefore, from \cref{eq:eq4_approx_formula,eq:eq6_in_approx_formula,eq:eq10_approx_formula}, we have
    \begin{align}
        \Bigg|\Tr\bigg[M_{\regA\regB}\bigg((\cM^{(k)}_{\text{\rm Haar},\regA}\otimes\identitymap_\regB)(\psi_{\regA\regB})
        -\sum_{\sigma\in S_k}\frac{1}{d^k}R_{\sigma,\regA}^\dag\otimes\Tr_{\regA'}[(R_{\sigma,\regA'}\otimes I_\regB)\psi_{\regA'\regB}]\bigg)\bigg]\Bigg|
        \le O\left(\frac{k^2}{d}\right).\label{eq:eq11_approx_formula}
    \end{align}
    for any $0\le M\le I$. This implies that the first term in \cref{eq:approx_formula} is at most $O(k^2/d)$.

    \paragraph{Estimation of the second term in \cref{eq:approx_formula}.}
    To conclude the proof, let us estimate the second term in \cref{eq:approx_formula}.
    To do so, it suffices to show the following;
    \begin{align}
        \Bigg|\Tr\bigg[
        (\cM^{(k)}_{\text{\rm Haar},\regA}\otimes\identitymap_\regB)(\psi_{\regA\regB})-\sum_{\sigma\in S_k}\frac{1}{d^k}R_{\sigma,\regA}^\dag\otimes\Tr_{\regA'}[(R_{\sigma,\regA'}\otimes I_\regB)\psi_{\regA'\regB}]
        \bigg]\Bigg|\le O\bigg(\frac{k^2}{d}\bigg)
    \end{align}
    It is clear $\Tr[(\cM^{(k)}_{\text{\rm Haar},\regA}\otimes\identitymap_\regB)(\psi_{\regA\regB})]=1$. On the other hand,
    \begin{align}
        &\Tr\bigg[\sum_{\sigma\in S_k}\frac{1}{d^k}R_{\sigma,\regA}^\dag\otimes\Tr_{\regA'}[(R_{\sigma,\regA'}\otimes I_\regB)\psi_{\regA'\regB}]
        \bigg]\\
        =&\frac{1}{d^k}\Tr\Bigg[I_\regA\otimes\Tr_{\regA'}[\psi_{\regA'\regB}]\bigg]+\Tr\bigg[\sum_{\sigma\in S_k;\sigma\neq e}\frac{1}{d^k}R_{\sigma,\regA}^\dag\otimes\Tr_{\regA'}[(R_{\sigma,\regA'}\otimes I_\regB)\psi_{\regA'\regB}]
        \bigg]\\
        =&1+\frac{1}{d^k}\sum_{\sigma\in S_k;\sigma\neq e}\Tr[R_{\sigma,\regA}]\Tr[(R_{\sigma,\regA'}\otimes I_\regB)\psi_{\regA'\regB}]\label{eq:eq12_approx_formula}
    \end{align}
    Thus,
    \begin{align}
        &\Bigg|\Tr\bigg[
        \Exp_{U\gets\mu_d}(U^{\otimes k}_\regA\otimes I_\regB) \psi_{\regA\regB}(U^{\dag\otimes k}_\regA\otimes I_\regB)-\sum_{\pi\in S_k}\frac{1}{d^k}R_{\sigma,\regA}^\dag\otimes\Tr_{\regA'}[(R_{\sigma,\regA'}\otimes I_\regB)\psi_{\regA'\regB}]
        \bigg]\Bigg|\\
        =&\frac{1}{d^k}\Bigg|\sum_{\sigma\in S_k;\sigma\neq e}\Tr[R_{\sigma,\regA}]\Tr[(R_{\sigma,\regA'}\otimes I_\regB)\psi_{\regA'\regB}]\Bigg|\\
        \le&\frac{1}{d^k}\sum_{\sigma\in S_k;\sigma\neq e}|\Tr[R_{\sigma,\regA}]\Tr[(R_{\sigma,\regA'}\otimes I_\regB)\psi_{\regA'\regB}]|\\
        \le&\frac{1}{d^k}\sum_{\sigma\in S_k;\sigma\neq e}|\Tr[R_{\sigma,\regA}]|\\
        =&\frac{1}{d^k}\bigg(\sum_{\sigma\in S_k}\Tr[R_{\sigma,\regA}]-\Tr[I_\regA]\bigg)\\
        =&\frac{1}{d^k}\bigg(k!\Tr[\Pi^{(d,k)}_{\text{sym},\regA}]-\Tr[I_\regA]\bigg)\\
        =&\frac{1}{d^k}\bigg(k!\binom{d+k-1}{k}-d^k\bigg)\\
        \le&O\bigg(\frac{k^2}{d}\bigg),\label{eq:eq13_approx_formula}
    \end{align}
    which implies that the second term in \cref{eq:approx_formula} is at most $O(k^2/d)$.
    Here we have used
    \begin{itemize}
        \item \cref{eq:eq12_approx_formula} in the first equality;
        \item $|\Tr[(R_{\sigma,\regA'}\otimes I_\regB)\psi_{\regA'\regB}]|\le 1$ for all $\sigma\in S_k$ in the second inequality;
        \item $|\Tr[R_{\sigma,\regA}]|=\Tr[R_{\sigma,\regA}]$ for all $\sigma\in S_k$ in the second equality;
        \item \cref{lem:symmetric_subspace} in the third and fourth equality.
    \end{itemize}

    Therefore, by putting together \cref{eq:approx_formula,eq:eq11_approx_formula,eq:eq13_approx_formula}, we have the desired result;
    \begin{align}
        \left\|
       (\cM^{(k)}_{\text{\rm Haar},\regA}\otimes\identitymap_\regB)(\psi_{\regA\regB})-\sum_{\sigma\in S_k}\frac{1}{d^k}R_{\sigma,\regA}^\dag\otimes\Tr_{\regA'}[(R_{\sigma,\regA'}\otimes I_\regB)\psi_{\regA'\regB}]
        \right\|_1
        \le O\bigg(\frac{k^2}{d}\bigg).
    \end{align}
    
\end{proof}

\begin{remark}
    Our bound in \cref{lem:approximation_formula_for_Haar_k-fold} is optimal because $\rho=\frac{\Pi^{(d,k)}_{\text{sym}}}{\Tr[\Pi^{(d,k)}_{\text{sym}}]}$ achieve the upper bound. We can check this as follows: from the concrete expression of $\Pi^{(d,k)}_{\text{sym}}$ in \cref{lem:symmetric_subspace}, we have
    \begin{align}
        \cM^{(k)}_{\text{\rm Haar}}(\Pi^{(d,k)}_{\text{sym}})=
        \Exp_{U\gets\mu_d}U^{\otimes k}\Pi^{(d,k)}_{\text{sym}}U^{\dag\otimes k}=\frac{1}{k!}\sum_{\sigma\in S_k}\Exp_{U\gets\mu_d}U^{\otimes k}R_\sigma U^{\dag\otimes k}=\frac{1}{k!}\sum_{\sigma\in S_k}R_\sigma=\Pi^{(d,k)}_{\text{sym}},
    \end{align}
    where we have used $R_\sigma U^{\otimes k}=U^{\otimes k}R_\sigma$ for any $\sigma\in S_k$ and any $U\in\Unitaries(d)$. On the other hand,
    \begin{align}
        \sum_{\sigma\in S_k}\frac{1}{d^k}\Tr[\Pi^{(d,k)}_{\text{sym}}R_\sigma]R_\sigma^\dag&=\sum_{\sigma\in S_k}\frac{1}{d^k}\Tr[\Pi^{(d,k)}_{\text{sym}}]R_\sigma^\dag\\
        &=\frac{1}{d^k}\binom{d+k-1}{k}\sum_{\sigma\in S_k}R_\sigma^\dag\\
        &=(1+O(k^2/d))\frac{1}{k!}\sum_{\sigma\in S_k}R_\sigma^\dag\\
        &=(1+O(k^2/d))\Pi^{(d,k)}_{\text{sym}},
    \end{align}
    where we have used $\Pi^{(d,k)}_{\text{sym}}R_\sigma=\Pi^{(d,k)}_{\text{sym}}$ for any $\sigma\in S_k$ in the first equality, and \cref{lem:symmetric_subspace} in the second and the last equality.
    Therefore,
    \begin{align}
        \left\|
        \cM^{(k)}_{\text{\rm Haar}}\bigg(\frac{\Pi^{(d,k)}_{\text{sym}}}{\Tr[\Pi^{(d,k)}_{\text{sym}}]}\bigg)-\sum_{\sigma\in S_k}\frac{1}{d^k}\Tr\bigg[\frac{\Pi^{(d,k)}_{\text{sym}}}{\Tr[\Pi^{(d,k)}_{\text{sym}}]}R_\sigma\bigg]R_\sigma^\dag
        \right\|_1 =O\left(\frac{k^2}{d}\right)\left\|\frac{\Pi^{(d,k)}_{\text{sym}}}{\Tr[\Pi^{(d,k)}_{\text{sym}}]}\right\|_1 =O\left(\frac{k^2}{d}\right).
    \end{align}
\end{remark}
\section{PRFSGs and PRSGs in the non-adaptive inverseless QHRO Model}
\label{sec:proof_of_PRFSGs_in_QHRO_model}

In this section, we construct selective secure PRFSGs in the non-adaptive inverseless QHRO model. 
\begin{theorem}\label{thm:PRSG_and_PRFSG_in_QHRO}
    Selective secure PRFSGs exist in the non-adaptive inverseless QHRO model.
\end{theorem}
This is shown by the following theorem.

\begin{theorem}\label{thm:main_caluculation_of_PRFSG_in_QHRO}
    Let $d,m,t,\ell\in\N$ such that $d>\sqrt{6}(t\ell+m)^{7/4}$. Let $\nu$ be the unitary $1$-design over $\Unitaries(d)$. Then, for any quantum state $\rho$ and distinct $x_1,...,x_\ell\in[d]$,
    \begin{align}
        &\bigg\|
        \Exp_{\substack{ U\gets\mu_d,\\P\gets\nu}}\bigotimes^\ell_{i=1}(UP\ket{x_i}\bra{x_i}P^\dag U^\dag)^{\otimes t}_\regA\otimes U^{\otimes m}_\regB\rho_{\regB\regC} U^{\dag\otimes m}_\regB-
        \bigotimes^\ell_{i=1}\bigg(\Exp_{\ket{\psi_i}\gets\mu^{s}_d}\ket{\psi_i}\bra{\psi_i}^{\otimes t}\bigg)_\regA\otimes\Exp_{U\gets\nu_d} U^{\otimes m}_\regB\rho_{\regB\regC} U^{\dag\otimes m}_\regB 
        \bigg\|_1\notag\\
        \le& O\bigg(\sqrt{\frac{m\ell}{d}}\bigg)+O\bigg(\frac{(t\ell+m)^2}{d}\bigg),
    \end{align}
    where $\mu^{s}_d$ denotes the Haar mesure over all $d$-dimensional states.
\end{theorem}

Before proving \cref{thm:main_caluculation_of_PRFSG_in_QHRO}, we show \cref{thm:PRSG_and_PRFSG_in_QHRO} assuming it.

\begin{proof}[Proof of \cref{thm:PRSG_and_PRFSG_in_QHRO}]
    Since the proof is the same, we only show the existence of PRFSGs in the QHRO model. Let $\secp\in\N$ be a security parameter and $\cU\coloneqq\{U_n\}_{n\in\N}$ be a single common Haar random unitary. Then, the following QPT algorithm $G^\cU$ becomes a PRFSG:
    \begin{itemize}
        \item  For $k,x\in\bit^\secp$, $G^\cU(k,x)$ prepares $X^k\ket{x}$ and query it to $U$, then outputs $\ket{\phi_k(x)}\coloneqq U_\secp X^k\ket{x}$.
    \end{itemize}
    Let $\cA^\cU$ be an adversary that queries $\xi_{\regB\regC}$ to $U_{\secp,\regB}^{\otimes m(\secp)}\otimes(\bigotimes^{r(\secp)}_{i=1}U_{n(i)})_\regC$ before receiving challenge states from the challnger, where $r$ and $m$ are a polynomial of $\secp$, and $n(i)\neq\secp$ for all $i\in[r]$. Note that this does not lose the generality. Then, for any polynomial $t,\ell$, and any $x_1,...,x_\ell\in\bit^\secp$, the probability that $\cA^\cU$ outputs $1$ when it receives $\ket{\phi_k(x_1)}^{\otimes t}\otimes...\otimes \ket{\phi_k(x_\ell)}^{\otimes t}$ is
    \begin{align}
        &\Pr_{\substack{\cU\gets\mu,\\k\gets\bit^\secp}}[1\gets\cA^\cU(\ket{\phi_k(x_1)}\bra{\phi_k(x_1)}^{\otimes t}\otimes...\otimes \ket{\phi_k(x_\ell)}\bra{\phi_k(x_\ell)}^{\otimes t})]\\
        =&\Exp_{\substack{ \cU\gets\mu,\\ k\gets\bit^\secp}}\bigotimes^\ell_{i=1}(U_\secp X^k\ket{x_i}\bra{x_i}X^{ k\dag} U_\secp^\dag)^{\otimes t}_\regA\otimes (U^{\otimes m}_{\secp,\regB}\otimes \bigotimes^{r(\secp)}_{i=1}U_{n(i),\regC})\xi_{\regB\regC}(U^{\otimes m}_{\secp,\regB}\otimes \bigotimes^{r(\secp)}_{i=1}U_{n(i),\regC})^\dag\\
        =&\Exp_{\substack{ U_\secp\gets\mu_{2^\secp},\\ k\gets\bit^\secp}}\bigotimes^\ell_{i=1}(U_\secp X^k\ket{x_i}\bra{x_i}X^{ k\dag} U_\secp^\dag)^{\otimes t}_\regA\otimes U^{\otimes m}_{\secp,\regB}\rho_{\regB\regC}U^{\dag\otimes m}_{\secp,\regB}\\
        =&\Exp_{\substack{ U_\secp\gets\mu_{2^\secp},\\P\gets\nu}}\bigotimes^\ell_{i=1}(U_\secp P\ket{x_i}\bra{x_i}P^\dag U^\dag_\secp)^{\otimes t}_\regA\otimes U^{\otimes m}_{\secp,\regB}\rho_{\regB\regC} U^{\dag\otimes m}_{\secp,\regB},
    \end{align}
    where
    \begin{itemize}
        \item $\rho_{\regB\regC}$ is defined as
        \begin{align}
        \rho_{\regB\regC}\coloneqq(I_\regB\otimes\bigotimes^{r(\secp)}_{i=1}U_{n(i),\regC})\xi_{\regB\regC} (I_\regB\otimes\bigotimes^{r(\secp)}_{i=1}U_{n(i),\regC})^\dag,
        \end{align}
        and we have inserted it in the second equality since $n(i)\neq\secp$ for all $i\in[r]$;
        \item $\nu$ is the uniform distribution over all $\secp$-qubit Pauli operator, and, in the last equality, we have used
        \begin{align}
            \Exp_{k\gets\bit^\secp}X^k\ket{x}\bra{x}X^{k\dag}=\Exp_{k,k'\gets\bit^\secp}X^kZ^{k'}\ket{x}\bra{x}Z^{\dag k'}X^{k\dag}=\Exp_{P\gets\nu}P\ket{x}\bra{x}P^\dag
        \end{align}
        for any $x\in\bit^\secp$ since any Pauli $Z$ does not cahnge $\ket{x}$ except for a global phase.
    \end{itemize}
    On the other hand, we have
    \begin{align}
        &\Pr_{\substack{\cU\gets\mu_{2^\secp},\\ \ket{\psi_1},...,\ket{\psi_\ell}\gets\mu^{s}_{2^\secp}}}[1\gets\cA^\cU(\ket{\psi_1}\bra{\psi_1}^{\otimes t}\otimes...\otimes \ket{\psi_\ell}\bra{\psi_\ell}^{\otimes t})]\\
        =&\bigotimes^\ell_{i=1}\bigg(\Exp_{\ket{\psi_i}\gets\mu^{s}_{2^\secp}}\ket{\psi_i}\bra{\psi_i}^{\otimes t}\bigg)_\regA\otimes
        \Exp_{\cU\gets\mu}(U^{\otimes m}_{\secp,\regB}\otimes \bigotimes^{r(\secp)}_{i=1}U_{n(i),\regC})\xi_{\regB\regC}(U^{\otimes m}_{\secp,\regB}\otimes \bigotimes^{r(\secp)}_{i=1}U_{n(i),\regC})^\dag\\
        =&\bigotimes^\ell_{i=1}\bigg(\Exp_{\ket{\psi_i}\gets\mu^{s}_{2^\secp}}\ket{\psi_i}\bra{\psi_i}^{\otimes t}\bigg)_\regA\otimes\Exp_{U_\secp\gets\mu_{2^\secp}} U^{\otimes m}_{\secp,\regB}\rho_{\regB\regC} U^{\dag\otimes m}_{\secp,\regB},
    \end{align}
    where $\rho_{\regB\regC}$ is defined above.
    Therefore
    \begin{align}
        &\bigg|\Pr_{\substack{\cU\gets\mu,\\k\gets\bit^\secp}}[1\gets\cA^\cU(\ket{\phi_k(x_1)}\bra{\phi_k(x_1)}^{\otimes t}\otimes...\otimes \ket{\phi_k(x_\ell)}\bra{\phi_k(x_\ell)}^{\otimes t})]\notag\\
        &-\Pr_{\substack{\cU\gets\mu_{2^\secp},\\ \ket{\psi_1},...,\ket{\psi_\ell}\gets\mu^{s}_{2^\secp}}}[1\gets\cA^\cU(\ket{\psi_1}\bra{\psi_1}^{\otimes t}\otimes...\otimes \ket{\psi_\ell}\bra{\psi_\ell}^{\otimes t})]\bigg|\\
        \le&\bigg\|
        \Exp_{\substack{ U_\secp\gets\mu_{2^\secp},\\P\gets\nu}}\bigotimes^\ell_{i=1}(U_\secp P\ket{x_i}\bra{x_i}P^\dag U^\dag_\secp)^{\otimes t}_\regA\otimes U^{\otimes m}_{\secp,\regB}\rho_{\regB\regC} U^{\dag\otimes m}_{\secp,\regB}-\bigotimes^\ell_{i=1}\bigg(\Exp_{\ket{\psi_i}\gets\mu^{s}_{2^\secp}}\ket{\psi_i}\bra{\psi_i}^{\otimes t}\bigg)_\regA\otimes\Exp_{U_\secp\gets\mu_{2^\secp}} U^{\otimes m}_{\secp,\regB}\rho_{\regB\regC} U^{\dag\otimes m}_{\secp,\regB} 
        \bigg\|_1\\
        \le& O\bigg(\sqrt{\frac{m\ell}{2^\secp}}\bigg)+O\bigg(\frac{(t\ell+m)^2}{2^\secp}\bigg)\\
        \le&\negl(\secp),
    \end{align}
    which implies $G^\cU$ is a PRFSG in the non-adaptive QHRO model. Here, we have used \cref{lem:Pauli_is_1-design,thm:main_caluculation_of_PRFSG_in_QHRO} in the second inequality.
\end{proof}

Before a proof of \cref{thm:main_caluculation_of_PRFSG_in_QHRO}, we show the following lemma.

\begin{lemma}\label{lem:1-design_twirling_for_distinct}
    Let $d,k,\ell\in\N$ and $\nu$ be a unitary $1$-design over $\Unitaries(d)$. For $\ell$ distinct $x_1,...,x_\ell\in[d]$, define $x\coloneqq(x_1,...,x_\ell)$ and
    \begin{align}
        \Lambda^{(x)}\coloneqq\sum_{y_1,...,y_k\in[d]/\{x_1,...,x_\ell\}}\bigotimes^k_{i=1}\ket{y_i}\bra{y_i}.
    \end{align} 
    Then, for any $d^k$-dimensional state $\rho$,
    \begin{align}
        \Tr[\Lambda^{(x)}\Exp_{P\gets\nu}P^{\otimes k}\rho P^{\dag\otimes k}]\ge1-\frac{m\ell}{d}.\label{eq:goal_of_1-design_twirling}
    \end{align}
\end{lemma}

\begin{proof}[Proof of \cref{lem:1-design_twirling_for_distinct}]
    Note that
    \begin{align}
        I-\Lambda^{(x)}=\sum_{\substack{y_1,...,y_k\in[d]\\ y_i=x_j\text{ for some }i\in[k]\text{ and }j\in[\ell]}}\bigotimes^k_{i=1}\ket{y_i}\bra{y_i}
        \le\sum_{i\in[k],j\in[\ell]}\Lambda_{i,j},\label{eq:eq1_in_1-design_twirling}
    \end{align}
    where $\Lambda_{i,j}\coloneqq I^{\otimes i-1}\otimes \ket{x_j}\bra{x_j}\otimes I^{\otimes k-i}$. Then,
    \begin{align}
        \Tr[(I-\Lambda^{(x)})\Exp_{P\gets\nu}P^{\otimes k}\rho P^{\dag\otimes k}]
        &\le\sum_{i\in[k],j\in[\ell]}\Tr[\Lambda_{i,j}\Exp_{P\gets\nu}P^{\otimes k}\rho P^{\dag\otimes k}]\\
        &=\sum_{i\in[k],j\in[\ell]}\Tr[\rho\Exp_{P\gets\nu}P^{\dag\otimes k}\Lambda_{i,j} P^{\otimes k}]\\
        &=\sum_{i\in[k],j\in[\ell]}\Tr[\rho\frac{I^{\otimes k}}{d}]\\
        &=\frac{m\ell}{d},
    \end{align}
    which implies \cref{eq:goal_of_1-design_twirling}.
    Here, the inequality follows from \cref{eq:eq1_in_1-design_twirling}, and the second equality follows from the definition of $\Lambda_{i,j}$ and $\nu$ is a $1$-design as follows:
    \begin{align}
        \Exp_{P\gets\nu}P^{\dag\otimes k}\Lambda_{i,j} P^{\otimes k}=\Exp_{P\gets\nu}I^{\otimes i-1}\otimes P^\dag\ket{x_j}\bra{x_j}P\otimes I^{\otimes k-i}=\frac{1}{d}I^{\otimes k}.
    \end{align}
\end{proof}

Now we are reday to prove \cref{thm:main_caluculation_of_PRFSG_in_QHRO}.

\begin{proof}[Proof of \cref{thm:main_caluculation_of_PRFSG_in_QHRO}]
    Let $x\coloneqq(x_1,...,x_\ell)$ and
    \begin{align}
        \Lambda^{(x)}\coloneqq\sum_{y_1,...,y_k\in[d]/\{x_1,...,x_\ell\}}\bigotimes^k_{i=1}\ket{y_i}\bra{y_i}.
    \end{align}
    Define
    \begin{align}
        \xi_{\regB\regC}=\Exp_{P\gets\nu}P^{\otimes k}_\regB\rho_{\regB\regC} P^{\dag\otimes k}_\regB
    \end{align}
    and
    \begin{align}
        \xi'_{\regB\regC}\coloneqq\frac{\Lambda^{(x)}_\regB\xi_{\regB\regC}\Lambda^{(x)}_\regB}{\Tr[\Lambda^{(x)}_\regB\xi_{\regB\regC}]}.
    \end{align}
    Note that 
    \begin{align}
        \|\xi_{\regB\regC}-\xi'_{\regB\regC}\|_1\le O\bigg(\sqrt{\frac{m\ell}{d}}\bigg)\label{eq:xi_is_close_to_xi'}
    \end{align}
    from \cref{lem:1-design_twirling_for_distinct,lem:gentle_measurement}.
    Let us consider the following sequence of matrices:
    \begin{align}
        &\rho^{(0)}_{\regA\regB\regC}\coloneqq\Exp_{\substack{ U\gets\mu_d,\\P\gets\nu}}\bigotimes^\ell_{i=1}(UP\ket{x_i}\bra{x_i}P^\dag U^\dag)^{\otimes t}_\regA\otimes U^{\otimes m}_\regB\rho_{\regB\regC} U^{\dag\otimes m}_\regB\\
        &\rho^{(1)}_{\regA\regB\regC}\coloneqq\Exp_{ U\gets\mu_d}\bigotimes^\ell_{i=1}(U\ket{x_i}\bra{x_i} U^\dag)^{\otimes t}_\regA\otimes U^{\otimes m}_\regB\xi_{\regB\regC} U^{\dag\otimes m}_\regB\\
        &\rho^{(2)}_{\regA\regB\regC}\coloneqq\Exp_{ U\gets\mu_d}\bigotimes^\ell_{i=1}(U\ket{x_i}\bra{x_i} U^\dag)^{\otimes t}_\regA\otimes U^{\otimes m}_\regB\xi'_{\regB\regC} U^{\dag\otimes m}_\regB\\
        &\rho^{(3)}_{\regA\regB\regC}\coloneqq\sum_{\pi\in S_{t\ell+m}}\frac{1}{d^{t\ell+m}}R_{\pi,\regA\regB}\otimes\Tr_{\regA\regB}[R^\dag_{\pi,\regA\regB}(\bigotimes^\ell_{i=1}\ket{x_i}\bra{x_i}^{\otimes t}_\regA\otimes\xi'_{\regB\regC})]\\
        &\rho^{(4)}_{\regA\regB\regC}\coloneqq\bigg(\sum_{\sigma\in S_t}\frac{1}{d^{t}}R_{\sigma}\bigg)^{\otimes\ell}_\regA\otimes \sum_{\tau\in S_m}\frac{1}{d^m}R_{\tau,\regB}\otimes\Tr_{\regB}[R^\dag_{\tau,\regB}\xi'_{\regB\regC}]\\
        &\rho^{(5)}_{\regA\regB\regC}\coloneqq\bigotimes^\ell_{i=1}\bigg(\Exp_{\ket{\psi_i}\gets\mu^{s}_d}\ket{\psi_i}\bra{\psi_i}^{\otimes t}\bigg)_\regA\otimes\sum_{\tau\in S_m}\frac{1}{d^m}R_{\tau,\regB}\otimes\Tr_{\regB}[R^\dag_{\tau,\regB}\xi'_{\regB\regC}] \\
        &\rho^{(6)}_{\regA\regB\regC}\coloneqq\bigotimes^\ell_{i=1}\bigg(\Exp_{\ket{\psi_i}\gets\mu^{s}_d}\ket{\psi_i}\bra{\psi_i}^{\otimes t}\bigg)_\regA\otimes\Exp_{U\gets\nu_d} U^{\otimes m}_\regB\xi'_{\regB\regC} U^{\dag\otimes m}_\regB \\
        &\rho^{(7)}_{\regA\regB\regC}\coloneqq\bigotimes^\ell_{i=1}\bigg(\Exp_{\ket{\psi_i}\gets\mu^{s}_d}\ket{\psi_i}\bra{\psi_i}^{\otimes t}\bigg)_\regA\otimes\Exp_{U\gets\nu_d} U^{\otimes m}_\regB\xi_{\regB\regC} U^{\dag\otimes m}_\regB \\
        &\rho^{(8)}_{\regA\regB\regC}\coloneqq\bigotimes^\ell_{i=1}\bigg(\Exp_{\ket{\psi_i}\gets\mu^{s}_d}\ket{\psi_i}\bra{\psi_i}^{\otimes t}\bigg)_\regA\otimes\Exp_{U\gets\nu_d} U^{\otimes m}_\regB\rho_{\regB\regC} U^{\dag\otimes m}_\regB.
    \end{align}

    We argue $\rho^{(i-1)}$ is indistinguishable from $\rho^{(i)}$ for each $i\in[8]$:
    \begin{itemize}
        \item $\rho^{(0)}_{\regA\regB\regC}=\rho^{(1)}_{\regA\regB\regC}$ follows from the left and right invariance of the Haar measure:
        \begin{align}
            \rho^{(0)}_{\regA\regB\regC}
            &=\Exp_{\substack{ U\gets\mu_d,\\P\gets\nu}}\bigotimes^\ell_{i=1}(UP\ket{x_i}\bra{x_i}P^\dag U^\dag)^{\otimes t}_\regA\otimes U^{\otimes m}_\regB\rho_{\regB\regC} U^{\dag\otimes m}_\regB\\
            &=\Exp_{\substack{ U\gets\mu_d,\\P\gets\nu}}\bigotimes^\ell_{i=1}(U'\ket{x_i}\bra{x_i} U'^\dag)^{\otimes t}_\regA\otimes (U'P^\dag)^{\otimes k}_\regB\rho_{\regB\regC} (U'P^\dag)^{\dag\otimes k}_\regB\\
            &=\Exp_{U'\gets\nu_d}\bigotimes^\ell_{i=1}(U'\ket{x_i}\bra{x_i} U'^\dag)^{\otimes t}_\regA\otimes U'^{\otimes k}_\regB(\Exp_{P\gets\nu}P^{\otimes k}_\regB\rho_{\regB\regC}P^{\otimes k}_\regB) U^{\dag\otimes m}_\regB\\
            &=\Exp_{U'\gets\nu_d}\bigotimes^\ell_{i=1}(U'\ket{x_i}\bra{x_i} U'^\dag)^{\otimes t}_\regA\otimes U'^{\otimes k}_\regB\xi_{\regB\regC} U^{\dag\otimes m}_\regB\\
            &=\rho^{(1)}_{\regA\regB\regC},
        \end{align}
        where we replaced the expectation of $U$ with that of $U'\coloneqq UP$ in the second equality.

        \item $\|\rho^{(1)}_{\regA\regB\regC}-\rho^{(2)}_{\regA\regB\regC}\|_1\le O(\sqrt{m\ell/d})$ from \cref{eq:xi_is_close_to_xi'}.

        \item $\|\rho^{(2)}_{\regA\regB\regC}-\rho^{(3)}_{\regA\regB\regC}\|_1\le O((t\ell+m)^2/d)$ by \cref{lem:approximation_formula_for_Haar_k-fold}.

        \item $\rho^{(3)}_{\regA\regB\regC}=\rho^{(4)}_{\regA\regB\regC}$ from the following observation. Suppose that $\pi\in S_{t\ell+m}$ cannot be decomposed into $\pi=(\sigma_1,...,\sigma_\ell,\tau)$ for any $\sigma_1,...,\sigma_\ell\in S_t$ and $\tau\in S_m$. Then, for any state $\ket{\phi}_\regB$,
        \begin{align}
            (I_\regA\otimes\Lambda^{(x)}_\regB) R_{\pi,\regA\regB}(\bigotimes^\ell_{i=1}\ket{x_i}^{\otimes t})_\regA\ket{\phi}_\regB=0
        \end{align}
        from the definition of $\Lambda^{(x)}$, which implies
        \begin{align}
            \Tr_{\regA\regB}[R^\dag_{\pi,\regA\regB}(\bigotimes^\ell_{i=1}\ket{x_i}\bra{x_i}^{\otimes t}_\regA\otimes\xi'_{\regB\regC})]=\Tr_{\regA\regB}[R^\dag_{\pi,\regA\regB}(\bigotimes^\ell_{i=1}\ket{x_i}\bra{x_i}^{\otimes t}_\regA\otimes\xi'_{\regB\regC})\Lambda_\regB]=0
        \end{align}
        for such $\pi\in S_{t\ell+m}$. On the other hand, if $\pi\in S_{t\ell+m}$ can be decomposed into $\pi=(\sigma_1,...,\sigma_\ell,\tau)$ for some $\sigma_1,...,\sigma_\ell\in S_t$ and $\tau\in S_m$, we have
        \begin{align}
            \Tr_{\regA\regB}\bigg[R^\dag_{\pi,\regA\regB}\bigg(\bigotimes^\ell_{i=1}\ket{x_i}\bra{x_i}^{\otimes t}_\regA\otimes\xi'_{\regB\regC}\bigg)\bigg]
            &=\Tr_{\regA\regB}\bigg[R^\dag_{(\sigma_1,...,\sigma_\ell,\tau),\regA\regB}
            \bigg(\bigotimes^\ell_{i=1}\ket{x_i}\bra{x_i}^{\otimes t}_\regA\otimes\xi'_{\regB\regC}\bigg)\bigg]\\
            &=\Tr_{\regA\regB}\bigg[\bigg(\bigg(\bigotimes^\ell_{i=1}R_{\sigma_i}\bigg)_\regA\otimes R_{\tau,\regB}\bigg)^\dag
            \bigg(\bigotimes^\ell_{i=1}\ket{x_i}\bra{x_i}^{\otimes t}_\regA\otimes\xi'_{\regB\regC}\bigg)\bigg]\\
            &=\bigg(\prod_{i\in[\ell]}\Tr[R_{\sigma_i}^\dag\ket{x_i}\bra{x_i}^{\otimes t}]\bigg)\Tr_{\regB}[R_{\tau,\regB}^\dag\xi'_{\regB\regC}]\\
            &=\Tr_{\regB}[R_{\tau,\regB}^\dag\xi'_{\regB\regC}].
        \end{align}
        Thus,
        \begin{align}
            \rho^{(3)}_{\regA\regB\regC}&=\sum_{\pi\in S_{t\ell+m}}\frac{1}{d^{t\ell+m}}R_{\pi,\regA\regB}\otimes\Tr_{\regA\regB}\bigg[R^\dag_{\pi,\regA\regB}
            \bigg(\bigotimes^\ell_{i=1}\ket{x_i}\bra{x_i}^{\otimes t}_\regA\otimes\xi'_{\regB\regC}\bigg)\bigg]\\
            &=\sum_{\sigma_1,...,\sigma_\ell\in S_t,\tau\in S_m}\frac{1}{d^{t\ell+m}}R_{(\sigma_1,...,\sigma_\ell,\tau),\regA\regB}\otimes\Tr_{\regA\regB}
            \bigg[R^\dag_{(\sigma_1,...,\sigma_\ell,\tau),\regA\regB}
            \bigg(\bigotimes^\ell_{i=1}\ket{x_i}\bra{x_i}^{\otimes t}_\regA\otimes\xi'_{\regB\regC}\bigg)\bigg]\\
            &=\sum_{\sigma_1,...,\sigma_\ell\in S_t,\tau\in S_m}\frac{1}{d^{t\ell+m}}
            \bigg(\bigotimes^\ell_{i=1}R_{\sigma_i}\bigg)_\regA\otimes
            R_{\tau,\regB}\otimes\Tr_{\regB}[R^\dag_{\tau,\regB}\xi'_{\regB\regC}]\\
            &=\bigg(\sum_{\sigma\in S_t}\frac{1}{d^{t}}R_{\sigma}\bigg)^{\otimes\ell}_\regA\otimes \sum_{\tau\in S_m}\frac{1}{d^m}R_{\tau,\regB}\otimes\Tr_{\regB}[R^\dag_{\tau,\regB}\xi'_{\regB\regC}]\\
            &=\rho^{(4)}_{\regA\regB\regC}.
        \end{align}

        \item $\|\rho^{(4)}_{\regA\regB\regC}-\rho^{(5)}_{\regA\regB\regC}\|_1\le O(\ell t^2/d)$ as follows: for each $j\in[\ell+1]$, define
        \begin{align}
            \rho'^{(j)}_{\regA\regB\regC}\coloneqq\bigg(\bigotimes^{j-1}_{i=1}\bigg(\Exp_{\ket{\psi_i}\gets\mu^{s}_d}\ket{\psi_i}\bra{\psi_i}^{\otimes t}\bigg)\otimes\bigg(\sum_{\sigma\in S_t}\frac{1}{d^{t}}R_{\sigma}\bigg)^{\otimes\ell-j+1}\bigg)_\regA\otimes \sum_{\tau\in S_m}\frac{1}{d^m}R_{\tau,\regB}\otimes\Tr_{\regB}[R^\dag_{\tau,\regB}\xi'_{\regB\regC}].
        \end{align}
        From \cref{lem:symmetric_subspace},
        \begin{align}
            \Exp_{\ket{\psi_i}\gets\mu^{s}_d}\ket{\psi_i}\bra{\psi_i}^{\otimes t}=\sum_{\sigma\in S_t}\frac{1}{d(d-1)...(d-t+1)}R_{\sigma}=\bigg(1+O\bigg(\frac{t^2}{d}\bigg)\bigg)\sum_{\sigma\in S_t}\frac{1}{d^{t}}R_{\sigma},
        \end{align}
        which implies $\|\rho'^{(j)}_{\regA\regB\regC}-\rho'^{(j+1)}_{\regA\regB\regC}\|_1\le O(t^2/d)$ for all $j\in[\ell]$. Since $\rho'^{(1)}_{\regA\regB\regC}=\rho^{(4)}_{\regA\regB\regC}$ and $\rho'^{(\ell+1)}_{\regA\regB\regC}=\rho^{(5)}_{\regA\regB\regC}$, we have $\|\rho^{(4)}_{\regA\regB\regC}-\rho^{(5)}_{\regA\regB\regC}\|_1\le O(\ell t^2/d)$.

        \item $\|\rho^{(5)}_{\regA\regB\regC}-\rho^{(6)}_{\regA\regB\regC}\|_1\le O(m^2/d)$ from \cref{lem:approximation_formula_for_Haar_k-fold}.

        \item $\|\rho^{(6)}_{\regA\regB\regC}-\rho^{(7)}_{\regA\regB\regC}\|_1\le O(\sqrt{m\ell/d})$ from \cref{eq:xi_is_close_to_xi'}.

        \item $\rho^{(7)}_{\regA\regB\regC}=\rho^{(8)}_{\regA\regB\regC}$ follows from the left and right invariance of the Haar measure:
        \begin{align}
            \rho^{(7)}_{\regA\regB\regC}
            &=\bigotimes^\ell_{i=1}\bigg(\Exp_{\ket{\psi_i}\gets\mu^{s}_d}\ket{\psi_i}\bra{\psi_i}^{\otimes t}\bigg)_\regA\otimes\Exp_{U\gets\nu_d} U^{\otimes m}_\regB\xi_{\regB\regC} U^{\dag\otimes m}_\regB\\
            &=\bigotimes^\ell_{i=1}\bigg(\Exp_{\ket{\psi_i}\gets\mu^{s}_d}\ket{\psi_i}\bra{\psi_i}^{\otimes t}\bigg)_\regA\otimes\Exp_{U\gets\nu_d} U^{\otimes m}_\regB(\Exp_{P\gets\nu}P^{\otimes k}_\regB\rho_{\regB\regC} P^{\dag\otimes k}) U^{\dag\otimes m}_\regB\\
            &=\bigotimes^\ell_{i=1}\bigg(\Exp_{\ket{\psi_i}\gets\mu^{s}_d}\ket{\psi_i}\bra{\psi_i}^{\otimes t}\bigg)_\regA\otimes\Exp_{U\gets\nu_d,P\gets\nu} (UP)^{\otimes k}_\regB\rho_{\regB\regC} (UP)^{\dag\otimes k}_\regB\\
            &=\bigotimes^\ell_{i=1}\bigg(\Exp_{\ket{\psi_i}\gets\mu^{s}_d}\ket{\psi_i}\bra{\psi_i}^{\otimes t}\bigg)_\regA\otimes\Exp_{U'\gets\nu_d} U'^{\otimes k}_\regB\rho_{\regB\regC} U'^{\dag\otimes k}_\regB\\
            &=\rho^{(8)}_{\regA\regB\regC},
        \end{align}
        where we replaced the expectation of $U$ with that of $U'\coloneqq UP$ in the last equality.
    \end{itemize}
    Therefore, from the triangle inequality, we have
    \begin{align}
        \|\rho^{(0)}_{\regA\regB\regC}-\rho^{(7)}_{\regA\regB\regC}\|_1\le O\bigg(\sqrt{\frac{m\ell}{d}}\bigg)+O\bigg(\frac{(t\ell+m)^2}{d}\bigg),
    \end{align}
    which concludes the proof.
\end{proof}
\section{Adaptively-secure PRFSGs in the invertible QHRO Model}
In this section, we prove the following theorem.
\begin{theorem}\label{thm:adaptivePRFSGs_in_SQHRO}
    Classically-accessible adaptively-secure PRFSGs exist in the invertible QHRO model.
\end{theorem}

\subsection{Construction}
We consider the following construction:
\begin{align}\label{eqn:_PUP}
    V_k^U\ket{\phi} = X^{k_1} \circ U\circ X^{k_0} \ket{\phi}
\end{align}
where a key $k=(k_0,k_1)\in K$ specifies two independent operators $X^{k_0},X^{k_1} \in \Unitaries(d)$.\footnote{The independency of $k_0$ and $k_1$ are used in the proof of \cref{claim:j->j+1}.}
The following theorem shows that $V_k^U$ is a pseudorandom unitary if the queries are all \emph{classical}, meaning that the input register of $V_k^U$ is measured in a computational basis before every query.
This implies~\cref{thm:adaptivePRFSGs_in_SQHRO},
combining with the proof of \cite[Theorem 5.14]{TCC:AGQY22}, which says that if $U$ is a Haar random unitary, then $\ket{x}\mapsto \ket{x}\otimes U\ket{x}$ is adaptively-secure PRFSGs.
\begin{theorem}\label{thm:PUP_is_pPRU}
    Let $V_k$ be in~\cref{eqn:_PUP}.
    Suppose that $\{X^{k_b}\}_k$ is such that, for any $b\in\bit$ and $x\in\{0,...,d-1\}$, $X^{k_b}\ket{x}$ is uniformly distributed over $\{\ket{0},...,\ket{d-1}\}$ for random $k$.
    For any $\cA$ having access to two oracles that makes $p$ classical-input queries to the first oracle and $q$ queries\footnote{We do not count the queries used by $V_k^U$.} to the second oracle (including inverse queries),
    it holds that
    \begin{align}\label{eqn:PRFSGs_in_inv_QHRO_model}
            \left|\Pr_{U\gets \mu_d, k \gets K}\left[
                \cA^{V_k^U,(U,U^\dag)} \to 1
            \right] - \Pr_{U\gets \mu_d, W \gets \mu_d}\left[
                \cA^{W,(U,U^\dag)} \to 1
            \right]\right|
        = O\left(\sqrt{\frac{p^3+p^2q^2}{d}}\right).
    \end{align}
\end{theorem}


\subsection{Preparation I: Lemmas}
We use the following lemmas to prove the security of~\cref{eqn:_PUP}.

For two quantum states $\ket{a}$ and $\ket{b}$, define the generalized swap operation
\[
\SWAP_{\ket a,\ket b}:\alpha\ket{a} + \beta\ket{b} + \ket{c} \mapsto \alpha\ket{b} + \beta\ket{a} + \ket{c} 
\]
for any $\ket{c} \in \Span(\ket{a},\ket{b})^\perp$.\footnote{This map is well-defined unitary as a (rotated) reflection in $\Span(\ket{a},\ket{b}).$
} It is the swap operation between two states if $\braket{a|b}=0$.
We sometimes use $\SWAP_{a,b}$ for brevity.

Our two main lemmas are stated below. We give the proofs in~\cref{subsec:proof_reprogramming,subsec:proof_resampling}. 
\begin{lemma}[Unitary reprogramming lemma]\label{lem:unitary_reprogramming}
    Let $\cD$ be a distinguisher in the following experiment:
    \begin{description}
        \item[Phase 1:] $\cD$ outputs a unitary $F_0=F$ over $m$-qubit and a quantum algorithm $\cC$ whose output is a quantum state $\rho$ and a classical string that specifies a classical description of the following data: a set $S$ of $m$-qubit pure states and a unitary $U_S$
        such that, 
        for the span $\cS$ of all states in $S$, $U_S$ acts as the identity on the image of $I-\Pi_\cS$, where $\Pi_\cS$ is the projection to $\cS$.    
        Let 
        \begin{align}
            \epsilon:= \sup_{\ket{\phi}:m\text{-qubit state}} \Exp_{ \cC}\left[\left\|\Pi_{\cS}\ket{\phi}\right\|^2\right].
        \end{align}
        \item[Phase 2:] 
        $\cC$ is executed and outputs $\rho$, $S$ and $U_S$.
        Let $F_1:=F \circ U_S$.
        A bit $b$ is chosen uniformly at random, and $\cD$ is given $\rho$ and quantum access to $F_b$ and makes $q$ queries in expectation if $b=0$, and sends the quantum state $\nu_b$ to the next phase.
        \item[Phase 3:] $\cD$ loses access to $F_b$ and receives $\nu_b$ and the classical string specifying the classical descriptions $S$ and $U_S$ outputted by $\cC$ in the second phase. Finally, $\cD$ outputs a guess $b'$.
    \end{description}
    Then, it holds that 
    \begin{align}
    \left|
        \Pr\left[
            \cD \to 1 | b=1
        \right]-
        \Pr\left[
            \cD \to 1 | b=0
        \right]
    \right|\le q \cdot \sqrt{2\epsilon}.
    \end{align}
    In fact, the trace distance $\TD(\nu_0,\nu_1)$ between two cases after Phase 2 is at most $q\sqrt{2\epsilon}.$
\end{lemma}

\begin{lemma}[Unitary resampling lemma]\label{lem:unitary_resampling}
    Let $\cD$ be a distinguisher in the following experiment:
    \begin{description}
        \item[Phase 1:] $\cD$ specifies two distributions of $d$-dimensional qudit pure quantum states $D^{\mu}_0,D^{\mu}_1$ such that $\Exp{\ket{\mu_i}\bra{\mu_i}}=I/d$ for $i=0,1$.
        $\cD$ makes at most $q$ forward or inverse queries to a $d$-dimensional Haar random unitary $U^{(0)}:=U$, and sends the quantum state $\nu$ to the next phase.
        \item[Phase 2:] 
        Sample $\ket{\mu_0}\gets D^{\mu}_0,\ket{\mu_1}\gets D^{\mu}_1$.
        A bit $b\in\bit$ is uniformly chosen, and $\cD$, given $\nu$ and classical descriptions of $\ket{\mu_0},\ket{\mu_1}$, is allowed to make arbitrarily many (forward or inverse) queries to an oracle that is either $U^{(0)}$ if $b=0$ or $U^{(1)}:=U\circ \SWAP_{{\mu_0},{\mu_1}}$ if $b=1$. Finally, $\cD$ outputs a bit $b'$.
    \end{description}
    Then, the following holds:
    \begin{align}
        \left|
            \Pr\left[
                b'=1 | b=0
            \right] -
            \Pr\left[
                b'=1 | b=1
            \right]
        \right|\le 2\sqrt{\frac{6q}{d}}.
    \end{align}
    In fact, the trace distance between two distributions $(\nu,\ket{\mu_0},\ket{\mu_1},U^{(0)})$ and $(\nu,\ket{\mu_0},\ket{\mu_1},U^{(1)})$ is at most $2\sqrt{\frac{6q}{d}}$ where $U$ and $U'$ are perfectly given as their classical description.
\end{lemma}

The following fact is (implicitly) used in this section multiple times.

\subsection{Preparation II: Simulations}
We occasionally consider that the algorithms or oracles have perfect knowledge of quantum states or unitaries, \emph{without hurting the algorithm's behavior}. This section explains the classical simulation or descriptions of quantum objects.
Here and below, we fix a way to express (unnormalized) pure quantum states $\ket{\phi}$ and unitary $U$ by classical strings $\str{\ket{\phi}}$ and $\str{U}$---for example by the amplitudes of the state or matrix that describes the unitary.

We define the classical simulation of the unitary oracle $U$ with the classical-input access as follows.
\begin{definition}\label{def:simU}
    Let $U$ be a $d\times d$ unitary.
    We define a classical simulation oracle $\Sim(U)$ with the classical-input queries that are defined as follows:
    \begin{itemize}
        \item It maintains a list $T$ of tuples of two strings representing quantum states, initialized by $T=\emptyset$.
        \item For the $j$-th query $x_j\in\{0,...,d-1\}$, it does:
        \begin{itemize}
            \item 
            If there is \emph{no} $\ell<j$ such that $(x_\ell,\str{\ket{\psi_\ell}})$ in $T$, it defines $\ket{\psi_j}:= U \ket{x_j}$ and returns $\str{\ket{\psi_j}}$. It appends $(x_j,\str{\ket{\psi_j}})$ at the end of $T$.
            \item If there is $\ell<j$ such that $(x_\ell,\str{\ket{\psi_\ell}})$ in $T$, it returns $\str{\ket{\psi_\ell}}$. 
            It samples a new $x'\in \{0,...,d-1\}$ where there is no $\ell<j$ satisfying the above condition, and appends $(x',\str{U\ket{x'}})$ at the end of $T$.\footnote{This step is to maintain the same size of the list.}
        \end{itemize}
    \end{itemize}
    The list $T$ after the $j$-th query is denoted by
    \begin{align}\label{eqn:def_Tj}
        T_j = \left\{(x_j,\str{\ket{\psi_i}})\right\}_{i\in[j]}.
    \end{align}
    We define $X_j:=\Span(\ket{x_1},...,\ket{x_j})$ and $\Psi_j:=\Span(\ket{\psi_1},...,\ket{\psi_j})$.
\end{definition}

\begin{lemma}
    Let $\cA^{W}$ be an oracle algorithm that only makes classical-input queries to $W$. Then, there exists an oracle algorithm $\Sim(\cA)^{\Sim(W)}$ with the same number of queries whose output is identical to $\cA^W$. In particular, 
    $\Pr[\cA^{W}\to 1] = \Pr[\Sim(\cA)^{\Sim(W)}\to 1]$.
\end{lemma}
\begin{proof}
    We define $\cB:=\Sim(\cA)$ as follows:
    \begin{itemize}
        \item It runs $\cA$, but when $\cA$ makes the $j$-th query $x_j$ to $W$, $\cB$ makes query $x_j$ to $\Sim(W)$ and obtain $\str{\ket{\psi_j}}$. It recovers $\ket{\psi_j}$ 
        and returns to $\cA$ as the output of the $j$-th query.
        \item If $\cA$ terminates, $\cB$ outputs whatever $\cA$ outputs.
    \end{itemize}
    From the perspective of $\cA$, the oracle answers are always identical, proving the lemma.
\end{proof}

\subsection{Security Proof}
Given the lemmas, we will prove \cref{thm:PUP_is_pPRU}. 
In other words, we will prove \cref{eqn:PRFSGs_in_inv_QHRO_model}.


\begin{proof}[Proof of \cref{thm:PUP_is_pPRU}]
    We call the algorithms in the real world when accessing the oracles $V_k^U,(U,U^\dag)$, and in the ideal world when accessing $W,(U,U^\dag)$. Below, we occasionally write $U$ to denote the oracle access to both $U$ and $U^\dag$.
    We let $\cB:=\Sim(\cA)$ and prove the indistinguishability with respect to the simulation algorithm $\cB$ with the simulated oracle $\Sim(W)$ or $\Sim(V_k^U)$.
    
    Recall $\cB$ maintains the list $T$; after the $j$-th pure state query, we write
    \begin{align}
        T_j=\left\{x_j,\str{\ket{\psi_i}})\right\}_{i\in[j]}
    \end{align}
    to denote the current $T$,
    where it holds $\ket{\psi_i} = W \ket{x_i}$
    in the ideal world, and 
    \begin{align}
        \ket{\psi_i} = V_k^U \ket{x_i} =  X^{k_1} \circ U \circ X^{k_0} \ket{x_i}
    \end{align} 
    for some $k$ in the real-world experiment for $i\in [j]$. 
    
    We write $\Pi_{i=1}^j f_i$ to denote $f_1 \circ ... \circ f_j$; the order is important because we use the notation $\Pi$ for the product of unitary operations. Following \cite[pp. 9-10]{alagic2022post}, for any unitary $U$, we define:
    \begin{align}
        &\ovec{Q}_{T_j,U,k} := \prod_{i=1}^j \SWAP_{X^{k_0}\ket{x_i},U^{\dag} \circ (X^{k_1})^{\dag}\ket{\psi_i}}, 
        &\ovec{S}_{T_j,U,k} := \prod_{i=1}^j \SWAP_{UX^{k_0}\ket{x_i},(X^{k_1})^{\dag}\ket{\psi_i}},
    \end{align}
    and define
    \begin{align}
        U_{T_j,k}=[U]_{T_j,k}:=U \circ \ovec{Q}_{T_j,U,k}
    \end{align}
    for unitaries $U$.
    For any $\ket{x_1},\ket{x_2},\ket{y_1},\ket{y_2}$ and $U$, 
    $\SWAP_{U\ket{x_1},U\ket{y_1}}\circ U\circ\SWAP_{\ket{x_2},\ket{y_2}}$ equals to
    \begin{align}
        \SWAP_{U\ket{x_1},U\ket{y_1}} \circ\SWAP_{U\ket{x_2},U\ket{y_2}}  \circ U 
        =U\circ\SWAP_{\ket{x_1},\ket{y_1}}\circ \SWAP_{\ket{x_2},\ket{y_2}},
    \end{align}
    which gives
    \begin{align}\label{eqn:one-sided_UTj}
        U_{T_j,k} 
        =  \ovec{S}_{T_j,U,k} \circ U .
    \end{align}

    We divide the execution of $\cB$ into $p+1$ phases $P_0,...,P_p$ 
    where $P_i$ describes the execution between the $i$-th and $(i+1)$-st queries to the first oracle; 
    $P_0$ corresponds to the execution before the first pure state query. 
    Let $q_{j}$ denote the expected query number of $B$ to the second oracle during $P_j$. 
    It holds that $q=\sum_{j=0}^p q_j$.

    We define the following sequences of experiments:
    \begin{align}
        {\bf H}_{j,0}:& 
        \underbrace{U,W,U,\cdots ,W,U}_{P_0,...,P_j}, 
        &&
        \underbrace{V_k^U,U_{T_j,k},}_{(j+1)\text{-st pure state query and }P_{j+1}}
        && 
        \underbrace{V_k^U,\cdots,V_k^U,U_{T_j,k}}_{P_{j+2},...,P_p}
        \\
        {\bf H}_{j,1}:&
        \underbrace{U,W,U,\cdots ,W,U}_{P_0,...,P_j}, 
        &&
        \underbrace{V_k^{U_j},[U_j]_{T_j,k},}_{(j+1)\text{-st pure state query and }P_{j+1}}
        &&
        \underbrace{V_k^{U_j},\cdots,V_k^{U_j},[U_j]_{T_j,k}}_{P_{j+2},...,P_p}
        \\
        {\bf H}_{j,2}:&
        \underbrace{U,W,U,\cdots ,W,U}_{P_0,...,P_j}, 
        &&
        \underbrace{W,U_{T_{j+1},k},}_{(j+1)\text{-st pure state query and }P_{j+1}}
        &&
        \underbrace{V_k^{U_j},\cdots,V_k^{U_j},U_{T_{j+1},k}}_{P_{j+2},...,P_p}
        \\
        {\bf H}_{j,3}:&
        \underbrace{U,W,U,\cdots ,W,U}_{P_0,...,P_j}, 
        &&\underbrace{W,U_{T_{j+1},k},}_{(j+1)\text{-st pure state query and }P_{j+1}}&& \underbrace{V_k^U,\cdots,V_k^U,U_{T_{j+1},k}}_{P_{j+1},...,P_p}
        \\
        {\bf H}_{j+1,0}:&
        \underbrace{U,W,U,\cdots ,W,U}_{P_0,...,P_j}, 
        &&\underbrace{W,U,}_{(j+1)\text{-st pure state query and }P_{j+1}}&& \underbrace{V_k^U,\cdots,V_k^U,U_{T_{j+1},k}}_{P_{j+1},...,P_p}
    \end{align}
    where $U_j$ will be specified later in the proof of \cref{claim:j0->j1} where its actual definition is used. 
    The characters (with super/subscripts) in the descriptions of hybrids denote:
    \begin{description}
        \item[$U$:] It denotes the phases $P_0,...,P_p$ that may contain multiple queries to the unitary $U.$
        \item[$V,W$:] They denote a \emph{single} pure state query to the corresponding \emph{simulation} oracle $\Sim(V)$ or $\Sim(W)$. Recall the oracles also store the query list $T$. 
        If the $j$-th query input $\ket{x_j}$ coincides with the previous $\ell$-th query for some $\ell < j$, the simulation oracle returns $\ket{\psi_\ell}$ \emph{without making the actual query}.\footnote{For example, even if $\ket{\psi_\ell}=W\ket{x_\ell}$ holds and the $j$-th oracle query is to $V_k^{U}$, the oracle returns the stored output $\ket{\psi_\ell}$. This resembles the assumptions that no same queries are made in the (classical or post-quantum) random oracle/permutations.}
    \end{description}
    Note that ${\bf H}_{0,0}$ and ${\bf H}_{p,0}$ correspond to the real- and ideal-world experiments, respectively.
    
    We write $\cB({\bf H}_{j,k})$ to denote the algorithm with the hybrid experiment ${\bf H}_{j,k}.$    
    We will prove the following claims, which correspond (parts of) \cite[Lemma 6,7]{alagic2022post}. Note that in \cref{claim:j1->j2} we connect the hybrid ${\bf H}_{j,1}$ and ${\bf H}_{j,3}$ and ${\bf H}_{j,2}$ appears as an intermediate hybrid in between.
    \begin{myclaim}\label{claim:j->j+1}
        $\left| 
            \Pr\left[
                \cB({\bf H}_{j,3})\to 1
            \right]-
            \Pr\left[
                \cB({\bf H}_{j+1,0})\to 1
            \right]
        \right|
        \le 4q_{j+1}\sqrt{6p^2/d}
        $ for $j=0,...,p-1$.
    \end{myclaim}
    \begin{myclaim}\label{claim:j0->j1}
        $\left| 
            \Pr\left[
                \cB({\bf H}_{j,0})\to 1
            \right]-
            \Pr\left[
                \cB({\bf H}_{j,1})\to 1
            \right]
        \right|\le 2\sqrt{6q/d}
        $ for $j=0,...,p-1$.
    \end{myclaim}
    \begin{myclaim}\label{claim:j1->j2}
        $\left| 
            \Pr\left[
                \cB({\bf H}_{j,1})\to 1
            \right]-
            \Pr\left[
                \cB({\bf H}_{j,3})\to 1
            \right]
        \right|\le 4\sqrt{p/d}
        $ for $j=0,...,p-1$.
    \end{myclaim}

    We prove these claims later. Given these claims, we prove our main result as follows.
    \begin{align}
        &\left|\Pr[\cB({\bf H}_{0,0})\to 1] - \Pr[\cB({\bf H}_{p,0})\to 1]\right| 
        \\& 
        \le\sum_{j=0}^{p-1} \left(
            4q_{j+1} \sqrt{\frac{6p^2}{d}} + 2\sqrt{\frac{6q}{d}} + 4\sqrt{\frac{p}{d}}
        \right)
        \\&
        \le 4q \sqrt{\frac{6p^2}{d}} + 2p\sqrt{\frac{6q}{d}} + 4p\sqrt{\frac{p}{d}}\\&
        =O\left( \sqrt{\frac{p^3 + p^2q^2}{d}} \right).
    \end{align}
    This proves the desired result.
\end{proof}

    \begin{proof}[Proof of \cref{claim:j->j+1}]
        Recall this claim compares the following two hybrids:
        \begin{align}
        {\bf H}_{j,3}:&
        \underbrace{U,W,U,\cdots ,W,U}_{P_0,...,P_j}, 
        &&\underbrace{W,{\color{red}U_{T_{j+1},k}},}_{(j+1)\text{-st pure state query and }P_{j+1}}&& \underbrace{V_k^U,\cdots,V_k^U,U_{T_{j+1},k}}_{P_{j+1},...,P_p}
        \\
        {\bf H}_{j+1,0}:&
        \underbrace{U,W,U,\cdots ,W,U}_{P_0,...,P_j}, 
        &&\underbrace{W,{\color{red}U},}_{(j+1)\text{-st pure state query and }P_{j+1}}&& \underbrace{V_k^U,\cdots,V_k^U,U_{T_{j+1},k}}_{P_{j+1},...,P_p}
        \end{align}
        Note that the only difference between the two hybrids is the second oracle in the phase $P_{j+1}$. We define the following distinguisher $\cD$ using $\cB$ to invoke the unitary reprogramming lemma (\cref{lem:unitary_reprogramming}):
        \begin{description}
            \item[Phase 1:] 
            $\cD$ samples Haar random unitary $U$ and defines a unitary
            \begin{align}
                F_0=F=\ketbra{0} \otimes U + \ketbra{1} \otimes U^{\dag}.
            \end{align}
            It defines the following algorithm:
            \begin{description}
                \item[$\cC$:] It samples Haar random unitary $W$ and $k\gets \cK$. It runs $\cB$ by answering the queries to the second oracles using $U$ and one to the first oracles using $\Sim(W)$, until after answering the $(j+1)$-st query to the first oracle with the intermediate state $\rho$. Until this point, $\cD$ gets the list
                \begin{align}
                    T_{j+1}=\{(x_i,\str{\ket{\psi_i}})\}_{i\in[j+1]}
                \end{align}
                of the input-output states to $W.$
                It then computes $S$ and $U_S$ such that
                \begin{align}
                    F_1=F\circ U_S = \ketbra{0} \otimes U_{T_{j+1},k} + \ketbra{1} \otimes U_{T_{j+1},k}^{\dag}
                \end{align}
                holds by choosing $U_S\coloneqq\ketbra{0}{0}\otimes\ovec{Q}_{T_j,U,k}+\ketbra{1}{1}\otimes \ovec{Q}_{T_j,U,k}^\dag$.
                
                Explicitly, the following choice of $S$ suffices by definition of $U_{T_{j+1},k}$ and \cref{eqn:one-sided_UTj}:
                \begin{align}\label{eqn:explicit_S}
                    \{\ket{0}\otimes X^{k_0}\ket{x_i},\ket{0}\otimes U^{\dag}\circ (X^{k_1})^{\dag}\ket{\psi_i}, \ket{1}\otimes UX^{k_0}\ket{x_i}, \ket
                    1 \otimes (X^{k_1})^{\dag} \ket{\psi_i}\}_{i=1}^{j+1}.
                \end{align}
                We define the following subsets:
                \begin{align}\label{eqn:decomp_S}
                    &S_{00}:=\{\ket{0}\otimes X^{k_0}\ket{x_i}\}_{i=1}^{j+1},
                    &&S_{01}:=\{\ket{0}\otimes U^{\dag}\circ (X^{k_1})^{\dag}\ket{\psi_i}\}_{i=1}^{j+1},\\& S_{10}:=\{\ket{1}\otimes UX^{k_0}\ket{x_i}\}_{i=1}^{j+1}, 
                    &&S_{11}:=\{\ket
                    1 \otimes (X^{k_1})^{\dag} \ket{\psi_i}\}_{i=1}^{j+1}.
                \end{align}
            \end{description}
            \item[Phase 2:] 
            $\cC$ is executed and outputs $S$ and $\rho$, and $\cD$ is given quantum access to $F_b$. $\cD$ resumes running $\cB$ given $\rho$, by answering the queries using $F_b$. When $\cB$ makes the $(j+2)$-nd pure-state query, this phase is finished.
            \item[Phase 3:] $\cD$ is now given the classical string specifying $k,W,T_{j+1}$. $\cD$ resumes running $\cB$, answering the queries to the first oracle using $\Sim(V_k^U)$ (with the list $T_{j+1}$) and the queries to the second oracle using $U_{T_{j+1},k}$. Finally, $\cD$ outputs whatever $\cB$ outputs.
        \end{description}
        This distinguisher $\cD$ fits in \cref{lem:unitary_reprogramming}. Also, if $b=0$, the distinguisher accesses the oracle exactly as in ${\bf H}_{j+1,0}$, whereas $b=1$ gives ${\bf H}_{j,3}$. 

        The expected number of queries in Phase 2 is exactly the expected number of queries in $P_{j+1}$, that is, $q_{j+1}$. 
        To bound $\epsilon$, note that the last register of each vector in $S$, described in \cref{eqn:explicit_S}, is applied by $1$-design unitaries over random $k$. 
        We will show that the following inequality holds:
        \begin{align}\label{eqn:last_part_claim_j3j0}
            \Exp_{\cC}\left[\left\|\Pi_{\cS}\ket{\phi}\right\|^2\right] \le \frac{12p^2}{d}.
        \end{align}
        The claim is followed by the unitary reprogramming lemma.
    \end{proof}

        \begin{proof}[Proof of \cref{eqn:last_part_claim_j3j0}]
            For any $b,c\in \bit$, $S_{bc}$ is the set of $j+1$ orthonormal states. 
            Let $\Pi_{bc}$ and $\Pi_b$ be the projections to the span of $S_{bc}$ and $S_{b0} \cup S_{b1}$.
            Also note that the states in $S_{00}\cup S_{01}$ and the states in $S_{10}\cup S_{11}$ are orthogonal.
            Thus it holds that
            \begin{align}
                    \|\Pi_\cS \ket{\phi}\|^2 = \|\Pi_0\ket{\phi}\|^2 + \|\Pi_1\ket{\phi}\|^2
            \end{align}
            and we bound each term. Below, we give the upper bound of $\|\Pi_0\ket{\phi}\|^2 $, and the same argument gives the same upper bound for $\|\Pi_1\ket{\phi}\|^2 $.

            Note that $\Pi_0 \Pi_{0b} = \Pi_{0b}$. The following can be easily verified:
            \begin{align}\label{eqn:identity}
                \Pi_0^2 = (\Pi_0 - \Pi_{00} - \Pi_{01})^2 + (\Pi_{00} - \Pi_{01})^2
            \end{align}
            which gives
            \begin{align}\label{eqn:first}
                \|\Pi_0 \ket{\phi}\|^2 
                &= 
                \|(\Pi_0 - \Pi_{00} - \Pi_{01})\ket{\phi}\|^2 + \|(\Pi_{00} - \Pi_{01})\ket{\phi}\|^2
                \\
                &\le 
                \|\Pi_{00}(\Pi_0 - \Pi_{00} - \Pi_{01})\ket{\phi}\|^2 + 
                \|(\Pi_0-\Pi_{00})(\Pi_0 - \Pi_{00} - \Pi_{01})\ket{\phi}\|^2
                \\&~~+\left(\|\Pi_{00}\ket{\phi}\|+\|\Pi_{01}\ket{\phi}\|\right)^2\\
                &=\|\Pi_{00}\Pi_{01}\ket{\phi}\|^2 
                +\|(\Pi_0-\Pi_{00})(\Pi_0-\Pi_{01})\ket{\phi}\|^2
                +\left(\|\Pi_{00}\ket{\phi}\|+\|\Pi_{01}\ket{\phi}\|\right)^2
                \\&
                \le \|\Pi_{00}\Pi_{01}\|^2_2 
                +\|(\Pi_0-\Pi_{00})(\Pi_0-\Pi_{01})\|^2_2
                +2\|\Pi_{00}\ket{\phi}\|^2+2\|\Pi_{01}\ket{\phi}\|^2\label{eqn:_last}
            \end{align}
            where the first equality is obtained by considering $\bra{\phi}(\cdot) \ket{\phi}$ on \cref{eqn:identity}, and the inequality holds because 1) we decompose the first term by the range of $\Pi_{00}$ and $(\Pi_0-\Pi_{00})$, and 2) we apply the triangle inequality on the second term.
            In the last inequality, we use the property of the matrix 2-norm and $(a+b)^2\le 2a^2+2b^2$.
            Using this inequality, we will prove that, using $j\le p-1$,
            \begin{align}
                \Exp_{\cC} 
                \|\Pi_0 \ket{\phi}\|^2 =\frac{2(j+1)^2+4(j+1)}{d} \le \frac{6p^2}{d}.
            \end{align}
            By similarly bounding $\Exp \|\Pi_1 \ket{\phi}\|^2$, we have the inequality $\|\Pi_\cS \ket{\phi}\|^2\le \frac{12p^2}{d},$ which concludes the proof of \cref{eqn:last_part_claim_j3j0}.

            Recall that $\ket{\psi_i}=W\ket{x_i}$ for all $i.$
            For convenience, write 
            $S_{00}=\{\ket{a_1},...,\ket{a_{j+1}}\}$ and $S_{01}=\{\ket{b_1},...,\ket{b_{j+1}}\}$
            where $\ket{a_i}:=\ket{0}\otimes X^{k_0}\ket{x_i} \in S_{00}$ and $\ket{b_\ell}:=\ket{0}\otimes U^{\dag} (X^{k_1})^{\dag}W\ket{x_\ell} \in S_{01}$. 
            For any $i,\ell \in \{1,...,j+1\}$, it holds that
            \begin{align}
                \Exp_{k} \|\ketbra{a_i} \cdot \ketbra{b_\ell} \|_2^2
                =  \Exp_{k} |\braket{b_\ell|a_i} |^2 
                =  \Exp_{k} |\bra{x_\ell} W^\dag X^{k_1} U X^{k_0}\ket{x_i} |^2  = \frac{1}{d}
            \end{align}
            because $\Exp_k[{X^{k_0} \ketbra{x_i} (X^{k_0})^\dag}]=I/d$. 

            We now give the upper bounds on the terms in \cref{eqn:_last} in expectation over $W,k$, which are randomness chosen by $\cC$. 
            The third and last terms can be bounded by $2(j+1)/d$ each easily.
            
            Note that $S_{0b}$ is an orthonormal set for $b=0,1$, thus $\Pi_{00} = \sum \ketbra{a_i}, \Pi_{01}= \sum \ketbra{b_i}.$ It holds that
            \begin{align}
                \Exp_k\|\Pi_{00} \Pi_{01}\|_2^2
                &= \Exp_k \|\sum_{1\le i,\ell \le j+1} \ketbra{a_i} \cdot \ketbra{b_\ell} \|^2_2
                \\& \le \Exp_k \sum_{1\le i,\ell \le j+1}  \|\ketbra{a_i} \cdot \ketbra{b_\ell} \|^2_2 = \frac{(j+1)^2}{d}
            \end{align}
            where we use $\|A+B\|_2 \le \max(\|A\|_2,\|B\|_2) \le \sqrt{\|A\|_2^2+\|B\|_2^2}$ for two orthogonal projectors $A,B$ such that $AB=0$, and $\|AB\|_2=\|BA\|_2$ for two matrices. This gives an upper bound on the first term. 

            For the second term, note that $\|(\Pi_0-\Pi_{00})(\Pi_0-\Pi_{01})\|_2= \|\Pi_{01}\Pi_{00}\|_2$ is well-known\footnote{It can be proven as follows. Let $R_0,R_{00},R_{01}$ be the ranges of the projectors $\Pi_0,\Pi_{00},\Pi_{01}$. Then it holds that $R_0=R_{00}\oplus R_{01} = R_{01}^\perp \oplus R_{01} = R_{01}^\perp \oplus R_{00}^\perp$, which implies that there exist unitary $U_0, U_1$ such that $U_0:R_{00}\to R_{01}^\perp$ and $U_1:R_{01}\to R_{00}^\perp$. For $U=U_0 \oplus U_1$, it holds that $\Pi_0 - \Pi_{00} = U \Pi_{01}U^\dag$ and vice versa, which proves $\|(\Pi_0-\Pi_{00})(\Pi_0-\Pi_{01})\|_2 = \| U \Pi_{01} U^\dag U \Pi_{00} U^\dag \|_2  = \|\Pi_{01}\Pi_{00}\|.$} if their ranges have only a trivial intersection $\{0\}$, which happens with probability 1. This gives the same upper bound on the second term.
        \end{proof}

    \begin{proof}[Proof of \cref{claim:j0->j1}]
        Recall this claim compares the following two hybrids:
        \begin{align}
            {\bf H}_{j,0}:& 
            \underbrace{U,W,U,\cdots ,W,U}_{P_0,...,P_j}, 
            &&
            \underbrace{V_k^{{\color{red}U}},{\color{red}U}_{T_j,k},}_{(j+1)\text{-st pure state query and }P_{j+1}}
            && 
            \underbrace{V_k^{{\color{red}U}},\cdots,V_k^{{\color{red}U}},{\color{red}U}_{T_j,k}}_{P_{j+2},...,P_p}
            \\
            {\bf H}_{j,1}:&
            \underbrace{U,W,U,\cdots ,W,U}_{P_0,...,P_j}, 
            &&
            \underbrace{V_k^{{\color{red}U_j}},[{\color{red}U_j}]_{T_j,k},}_{(j+1)\text{-st pure state query and }P_{j+1}}
            &&
            \underbrace{V_k^{{\color{red}U_j}},\cdots,V_k^{{\color{red}U_j}},[{\color{red}U_j}]_{T_j,k}}_{P_{j+2},...,P_p}
        \end{align}
        Note that $U_j$ is yet to be defined; we will define $U_j$ by $U\circ \SWAP_{\mu_0,\mu_1}$ for some $d$-dimensional qudit states $\ket{\mu_0},\ket{\mu_1}$.
        The only difference between the two hybrids is that the unitary $U$ in ${\bf H}_{j,0}$ from the $(j+1)$-st phase is replaced by $U_j$ in ${\bf H}_{j,1}$.

        We define the following distinguisher $\cD$ using $\cB$ to invoke \cref{lem:unitary_resampling}:
        \begin{description}
            \item[Phase 1:] $\cD$ samples a $d$-dimensional Haar unitary $W$ and
            specifies the following two distributions of $\ket{\mu_0},\ket{\mu_1}$:
            \begin{align}
                D_0^\mu&:= \{\ket{0},...,\ket{d-1}\},\\
                D_1^\mu&:=\{d\text{-dimensional pure states} \}.
            \end{align}
            Given access to a $d$-dimensional Haar random unitary $U^{(0)}:=U$, $\cD$ runs $\cB^{\Sim(W),U}$ until $\cB$ asks the $(j+1)$-st query $x_{j+1}$ to $\Sim(W)$. This phase is finished before answering this query.
            Until this point, $\cD$ gets the list
                $T_{j}=((x_i,\str{\ket{\psi_i}}))_{i\in[j]}$
            of the input-output states to $W$.
            \item[Phase 2:] Samples $\ket{\mu_0}\gets D_0^\mu$, i.e., $\mu_0 \gets \{0,...,d-1\}$, and $\ket{\mu_1}\gets D_1^\mu$ and $b\gets \{0,1\}$.
            Now $\cD$ got the $(j+1)$-st query $\ket{x_{j+1}}$ for some $x_{j+1}\in\{0,...,d-1\}$. $\cD$ defines $k$ so that $X^{k_0}\ket{x_{j+1}}=\ket{\mu_0}$ and randomly chooses $k_1$.  Given oracle access to $U^{(b)}$, $\cD$ resumes running $\cB$, answering the remaining queries to the first oracle using $V_k^{U^{(b)}}$ and the remaining queries to the second oracle using $[U^{(b)}]_{T_j,k}$. In particular, the $(j+1)$-st query $x_{j+1}$ to the first oracle is answered by $U\ket{\mu_1}$.
        \end{description}
        
        The distinguisher fits in \cref{lem:unitary_resampling}. Also, by defining $U_j:=U^{(1)}$, the case of $b=0$ corresponds to ${\bf H}_{j,0}$ and $b=1$ corresponds to ${\bf H}_{j,1}$, respectively. Since $\Exp{\ketbra{\mu_i}}=I/d$ holds for $i=0,1$, the unitary resampling lemma proves the claim, where the number of queries in the first phase is $q_0+...+q_j \le q$.
    \end{proof}

    \begin{proof}[Proof of \cref{claim:j1->j2}]
        We consider the following variations of the second phase of the experiment in the proof of \cref{claim:j0->j1}, where we highlight the changed parts by red and the omitted parts are identical to Phase 2:
        \begin{description}
            \item[Phase 2-1:] Samples $\ket{\mu_0}\gets D_0^\mu$ and ${\color{red}\ket{\mu'_1}}\gets D_1^\mu$ and {\color{red}$b\gets 1$}. 
            {\color{red}Define $\ket{\mu_1}:=\frac{(I-\Pi_{U^{\dag}(\Psi_j)})\ket{\mu_1'}}{\|(I-\Pi_{U^{\dag}(\Psi_j)})\ket{\mu_1'}\|}$.} Here $\Phi_j=\Span(\ket{\psi_1},...,\ket{\psi_j})$ as defined in \cref{def:simU}.
            \item[Phase 2-2:] Samples $\ket{\mu_0}\gets D_0^\mu$ and {\color{red}$\ket{\mu_1}$ is defined by $U^{\dag}W\ket{x_{j+1}}$} and {$b\gets 1$}. ...
            In particular, 
            the $(j+1)$-st query $x_{j+1}$ to the first oracle is answered by $U\ket{\mu_1}{\color{red}  = W\ket{x_{j+1}}}$.
            \item[Phase 2-3:] Samples $\ket{\mu_0}\gets D_0^\mu$ and ${\ket{\mu_1}}$ is defined by $ U^{\dag}W\ket{x_{j+1}}$ and {$b\gets 1$}.
            ... 
            answering the remaining queries to the first oracle {\color{red} using $V_k^{U^{(1)}}{\color{red}(I-\Pi_{\Span(\ket{x_{j+1}},(X^{k_0})^{\dag}\ket{\mu_1})})}$} and the remaining queries to the second oracle using $[U^{(1)}]_{T_j,k}$.
            The $(j+1)$-st query $x_{j+1}$ to the first oracle is answered by $U\ket{\mu_1} = W\ket{x_{j+1}}$.
            \item[Phase 2-4:] Samples $\ket{\mu_0}\gets D_0^\mu$ and $\ket{\mu_1}$ is defined by $ U^{\dag}W\ket{x_{j+1}}$ and {$b\gets 1$}. 
            ... answering the remaining queries to the first oracle using ${\color{red}V_k^{U^{(0)}}}$ and the remaining queries to the second oracle using $[U^{(1)}]_{T_j,k}$.
        \end{description}
        We write ${\bf H}_{j,1}^{(1)}$,..., ${\bf H}_{j,1}^{(4)}$ to denote these changed hybrids. In Phase 2-3, we only change the output $\ket{\psi}$ according to the oracle's answers in the list $T$.
        \begin{itemize}
            \item Note that ${\bf H}_{j,1}$ coincides with {\bf Phase 2} with $b=0$, and ${\bf H}_{j,1}^{(1)}$ only differs the projective measurement.
            We think ${\bf H}_{j,1}^{(1)}$ as applying the projection $(I-\Pi_{U^{\dag}(\Psi_j)})$ on $\ket{\mu_1'}$ \emph{and proceeding conditioned on its success}.
            Note that $\ket{\mu_1'}$ is independent from $U$ and the states in $\Psi_j$ are orthogonal to each other. Thus, the projection fails with probability
            \begin{align}
                \|\Pi_{U^{\dag}(\Psi_j)}\ket{\mu_1'}\|^2 =\sum_{i=1}^j |\bra{\psi_i} U \ket{\mu_1'}|^2 = \frac{j}{d} \le {\frac{p}{d}}.
            \end{align}
            By the gentle measurement lemma, we have $|\Pr[\cB({\bf H}_{j,1})\to 1] - \Pr[\cB({\bf H}_{j,1}^{(1)})\to 1]|\le \sqrt{p/d}.$
            
            \item Phase 2-2 is identical to Phase 2-1 because $\ket{\mu_1}$ is uniformly distributed over $\Image(I-\Pi_{U^{\dag}(\Psi_j)})$ in both cases. 
            Note that this corresponds to 
            \begin{align}
                {\bf H}_{j,2}:\underbrace{U,W,U,\cdots ,W,U}_{P_0,...,P_j}, 
                &&
                \underbrace{{\color{red}W},U_{T_{j+1},k},}_{(j+1)\text{-st pure state query and }P_{j+1}}
                &&
                \underbrace{V_k^{U_j},\cdots,V_k^{U_j},U_{T_{j+1},k}}_{P_{j+2},...,P_p}  .
            \end{align}
            \item For Phase 2-3, for any $x\in\{0,...,d-1\}\setminus \{x_{j+1}\}$, the expectation over $U,W,k$ satisfies
            \begin{align}
                \Exp{\left|
                    \bra{\mu_1} X^{k_0} \ket{x}
                \right|^2}=
                \Exp{\left|
                    \bra{x_{j+1}}UW^{\dag} X^{k_0} \ket{x}
                \right|^2}=\frac{1}{d}
            \end{align}
            which represents that the projection on input $x$ fails.
            Therefore, by the quantum union bound (\cref{lem:quantum_union}), we have
            \begin{align}
                \left|
                    \Pr[\cB({\bf H}_{j,1}^{(2)})\to 1] - 
                    \Pr[\cB({\bf H}_{j,1}^{(3)})\to 1] 
                    \right | \le \sqrt{\frac{p}{d}}.
            \end{align}
            \item Finally, Phase 2-4 corresponds to
            \begin{align}
                {\bf H}_{j,3}:\underbrace{U,W,U,\cdots ,W,U}_{P_0,...,P_j}, 
                &&
                \underbrace{W,U_{T_{j+1},k},}_{(j+1)\text{-st pure state query and }P_{j+1}}
                &&
                \underbrace{V_k^{{\color{red}U}},\cdots,V_k^{{\color{red}U}},U_{T_{j+1},k}}_{P_{j+2},...,P_p}  .
            \end{align}
            Since $V_k^{{U}}$ and $V_k^{U_j}$ are identical in $\Image(\Pi_{\Span(\ket{x_{j+1}},(X^{k_0})^{\dag}\ket{\mu_1})})$, the same argument above shows that 
            \begin{align}
                \left|
                    \Pr[\cB({\bf H}_{j,1}^{(3)})\to 1] - 
                    \Pr[\cB({\bf H}_{j,1}^{(4)})\to 1] 
                    \right | \le \sqrt{\frac{p}{d}}.
            \end{align}
        \end{itemize}
        Combining the above, we prove the claim. 
    \end{proof}

\subsection{Proof of Unitary Reprogramming Lemma}\label{subsec:proof_reprogramming}
This section proves \cref{lem:unitary_reprogramming}. We recall the statement below for convenience.
\begin{lemma}[Unitary reprogramming lemma]
    Let $\cD$ be a distinguisher in the following experiment:
    \begin{description}
        \item[Phase 1:] $\cD$ outputs a unitary $F_0=F$ over $m$-qubit and a quantum algorithm $\cC$ whose output is a quantum state $\rho$ and a classical string that specifies a classical description of the following data: a set $S$ of $m$-qubit pure states and a unitary $U_S$
        such that, 
        for the span $\cS$ of all states in $S$, $U_S$ acts as the identity on the image of $I-\Pi_\cS$, where $\Pi_\cS$ is the projection to $\cS$.    
        Let 
        \begin{align}
            \epsilon:= \sup_{\ket{\phi}:m\text{-qubit state}} \Exp_{ \cC}\left[\left\|\Pi_{\cS}\ket{\phi}\right\|^2\right].
        \end{align}
        \item[Phase 2:] 
        $\cC$ is executed and outputs $\rho$, $S$ and $U_S$.
        Let $F_1:=F \circ U_S$.
        A bit $b$ is chosen uniformly at random, and $\cD$ is given $\rho$ and quantum access to $F_b$ and makes $q$ queries in expectation if $b=0$, and sends the quantum state $\nu_b$ to the next phase.
        \item[Phase 3:] $\cD$ loses access to $F_b$ and receives $\nu_b$ and the classical string specifying the classical descriptions $S$ and $U_S$ outputted by $\cC$ in the second phase. Finally, $\cD$ outputs a guess $b'$.
    \end{description}
    Then, it holds that 
    \begin{align}
    \left|
        \Pr\left[
            \cD \to 1 | b=1
        \right]-
        \Pr\left[
            \cD \to 1 | b=0
        \right]
    \right|\le q \cdot \sqrt{2\epsilon}.
    \end{align}
    In fact, the trace distance $\TD(\nu_0,\nu_1)$ between two cases after Phase 2 is at most $q\sqrt{2\epsilon}.$
\end{lemma}

\begin{proof}
    Let $\cM_\regC(\rho_{\regA\regC}) = \ketbra{0}_C \rho_{\regA\regC} \ketbra{0}_C + \ketbra{1}_C \rho_{\regA\regC} \ketbra{1}_C$ for a single-qubit register $C$ and arbitrary ancillary register $\regA$.
    Let $F$ be a unitary over $\bit^m$. The controlled version of $F$ is defined by
    \begin{align}
        cF = \ketbra{0} \otimes I + \ketbra{1} \otimes F 
    \end{align}
    so that $cF:\ket{c}\ket{x} \mapsto \ket{c} F^c \ket{x}.$\footnote{The controlled queries reflect the expected number of queries. See~\cite[Section 4.1]{alagic2022post} for a more detailed discussion.}
    The execution of $\cD$ can be described by
    \begin{align}
        (\Phi \circ cF \circ \cM_C)^{q_{\max}}
    \end{align}
    which is applied to some initial state $\rho$
    where $q_{\max}$ is the upper bound of the number of queries and $\Phi$ is an arbitrary quantum channel.\footnote{Each layer may have different channels $\Phi_1,...,\Phi_{q_{\max}}$. The standard argument with the counter, i.e. $\Phi=\sum_{j=1}^{q_{\max}} \ketbra{j}{j-1}\otimes \Phi_j$ allows us to use a single channel $\Phi$ without loss of generality.}
    Let $\Gamma_b= \Phi \circ cF_b \circ \cM_C$ and define
    \begin{align}
        \rho_k := (\Gamma_1^{q_{\max}-k} \circ \Gamma_0^k ) (\rho)
    \end{align}
    which corresponds to the final state where the first $k$ queries are answered by $cF_0$, and the remaining queries are answered by $cF_1$. The intermediate state after the $k$-th query is denoted by $\rho_k^{(0)}:=\Gamma_0^k (\rho).$ The final state of the algorithm using the $F_0$ (or $F_1$) oracle entirely is $\rho_0$ ($\rho_{q_{\max}}$, respectively).
    We also define $p_k:=\Tr\left[\ketbra 1_C \rho_{k}^{(0)}\right]$ to represent the probability that the oracle query is made in the $(k+1)$-th iteration for $0\le k <q_{\max}$.

    We give an upper bound
    \begin{align}
        \Exp_{r}\left[
            \TD\left(\ketbra{r} \otimes \rho_{q_{\max}},\ketbra{r} \otimes \rho_0\right)
        \right] \le q \sqrt{2\epsilon}
    \end{align}
    where $r$ is the randomness used by $\cC$. 
    Note that $cF_0^\dagger \circ cF_1 = cU_S.$
    By the monotonicity of $\TD$ under quantum channels, for any $r$ we have
    \begin{align}
        \TD\left(\ketbra{r} \otimes \rho_k , \ketbra{r}\otimes\rho_{k-1}\right) 
        &\le \TD\left(
            cF_0\circ \cM_C \left(\rho_{k-1}^{(0)}\right),
            cF_1\circ \cM_C \left(\rho_{k-1}^{(0)}\right)
        \right)\\
        &=\TD\left(
            \cM_C \left(\rho_{k-1}^{(0)}\right),
            cU_S\circ \cM_C \left(\rho_{k-1}^{(0)}\right)
        \right)\label{eqn:unitaryrep_TDbound}.
    \end{align}
    We can write
    \begin{align}
        c U_S \circ \cM_C\left(\rho_{k-1}^{(0)}\right) = 
        U_S \left(\ketbra{1}_C\rho_{k-1}^{(0)}\ketbra{1}_C\right)  + \ketbra 0_C\rho_{k-1}^{(0)}\ketbra 0_C,
    \end{align}
    thus \cref{eqn:unitaryrep_TDbound} can be written as
    \begin{align}
        &\TD\left(
            \ketbra 1_C\rho_{k-1}^{(0)}\ketbra 1_C,
           U_S \left(\ketbra 1_C\rho_{k-1}^{(0)}\ketbra 1_C\right)
        \right)\\
        &=p_{k-1} \cdot \TD\left(
            \sigma_{k-1},U_S(\sigma_{k-1})
        \right)
    \end{align}
    where we recall $p_{k-1}=\Tr\left[\ketbra 1_C \rho_{k-1}^{(0)}\right]$ and define $\sigma_{k-1}$ be the normalization of $\ketbra 1_C \rho_{k-1}^{(0)}\ketbra 1_C$. 
    This gives
    \begin{align}
        \Exp_{r}\left[
            \TD\left(\ketbra{r} \otimes \rho_{q_{\max}},\ketbra{r} \otimes \rho_0\right)
        \right] &
        \le \sum_{k=1}^{q_{\max}}
        \Exp_{r}\left[
            \TD\left(\ketbra{r} \otimes \rho_{q_{k}},\ketbra{r} \otimes \rho_{k-1}\right)\right]
        \\&
        \le \sum_{k=1}^{q_{\max}}
        p_{k-1}\cdot \Exp_{S}\left[ \TD\left(
            \sigma_{k-1},U_S (\sigma_{k-1})
        \right)
        \right]
        \\&\label{eqn:reprogramming_supsigmaupperbound}
        \le q\cdot \sup_\sigma \Exp_{S}\left[ \TD\left(
            \sigma,U_S (\sigma)
        \right)
        \right].
    \end{align}
    Using the fact that any mixed state is a convex combination of pure states and $\TD(\ket{\phi},\ket{\psi}) =\sqrt{1-\braket{\phi|\psi}^2} = \|\ket\phi - \ket\psi\|_2 /\sqrt{2}$, we have
    \begin{align}
        \sup_{\sigma} \Exp_{r}\left[\TD(\sigma,U_S (\sigma))\right]  
        \le\sup_{\ket{\phi}} \Exp_{S}\left[\TD(\ket{\phi},U_S\ket{\phi})\right]
        = 
        \frac{\sup_{\ket\phi} \Exp_{S}\left[\|\ket{\phi}-U_S\ket{\phi}\|_2\right]}{\sqrt{2}}
        \le \sqrt{2 \epsilon}.\label{eqn:reprogramming_supupperbound}
    \end{align}
    Plugging \cref{eqn:reprogramming_supupperbound} into \cref{eqn:reprogramming_supsigmaupperbound} concludes the proof.
    The last inequality of \cref{eqn:reprogramming_supupperbound} is derived by, recalling the map $U_S$ acts as the identity on the image of $I-\Pi_\cS$,
    \begin{align}
        &\Exp_{S}\left[\|\ket{\phi}-U_S\ket{\phi}\|_2\right]
        \\&
        =\Exp_{S}\left[\|\Pi_\cS\ket{\phi}-U_S\Pi_\cS\ket{\phi}\|_2\right]
        \\&
        \le \Exp_{S}\left[\|\Pi_\cS\ket{\phi}\|_2\right]
        +\Exp_S\left[\|U_S\Pi_\cS\ket{\phi}\|_2\right]
        \\&
        =2\Exp_{S}\left[\|\Pi_\cS\ket{\phi}\|_2\right]
        \\&\le 2\sqrt{\Exp_{S}\left[\|\Pi_\cS\ket{\phi}\|_2^2\right]} \le 2 \sqrt \epsilon
    \end{align}
    for any $\ket{\phi}$, where we use Jensen's inequality in the last step. 
\end{proof}
\subsection{Unitary Resampling Lemma}\label{subsec:proof_resampling}
This section proves \cref{lem:unitary_resampling}. We recall the statement below for convenience.
\begin{lemma}[Unitary resampling lemma]
    Let $\cD$ be a distinguisher in the following experiment:
    \begin{description}
        \item[Phase 1:] $\cD$ specifies two distributions of $d$-dimensional qudit pure quantum states $D^{\mu}_0,D^{\mu}_1$ such that $\Exp{\ket{\mu_i}\bra{\mu_i}}=I/d$ for $i=0,1$.
        $\cD$ makes at most $q$ forward or inverse queries to a $d$-dimensional Haar random unitary $U^{(0)}:=U$, and sends the quantum state $\nu$ to the next phase.
        \item[Phase 2:] 
        Sample $\ket{\mu_0}\gets D^{\mu}_0,\ket{\mu_1}\gets D^{\mu}_1$.
        A bit $c\in\bit$ is uniformly chosen, and $\cD$, given $\nu$, is allowed to make arbitrarily many (forward or inverse) queries to an oracle that is either $U$ if $b=0$ or $U':=U\circ \SWAP_{{\mu_0},{\mu_1}}$ if $b=1$. Finally, $\cD$ outputs a bit $b'$.
    \end{description}
    Then, the following holds:
    \begin{align}
        \left|
            \Pr\left[
                b'=1 | b=0
            \right] -
            \Pr\left[
                b'=1 | b=1
            \right]
        \right|\le 2\sqrt{\frac{2q}{d}}.
    \end{align}
    In fact, the trace distance between two distributions $(\nu,U)$ and $(\nu,U')$ is at most $2\sqrt{\frac{2q}{d}}$ where $U$ and $U'$ are perfectly given as their classical description.
\end{lemma}

\begin{proof}
    We assume that the execution of the first phase of $\cD$ can be described by\footnote{Technically, multiple projective measurements may exist between two oracle queries, which cannot be deferred because of the pure state queries. The general case can be proven exactly the same way, and considering the general case only makes the description of $\cD$ complicated.}
    \begin{align}
        \cD_1^U:=\Phi \circ U^{\pm1}\circ ... \circ U^{\pm1} \circ \Phi \circ U \circ \Phi
    \end{align}
    where $\Phi$ is an arbitrary quantum channel that may include the intermediate measurements.

    We consider the following two continuous distributions:
    \begin{description}
        \item[$D_0(\nu_{U},U)$:] It samples a Haar random unitary $U$, $\ket{\mu_0},\ket{\mu_1}$. 
        It runs $\cD_1^U$ on input $\ketbra{0}$ and obtains $\nu_{U}$. 
        Then it outputs $(\nu_{U},U)$ where $U$ specifies the full classical description of $U$.
        \item[$D_1(\nu_{U},U')$:] It samples a Haar random unitary $U$, $\ket{\mu_0},\ket{\mu_1}$.
        It runs $\cD_1^U$ on input $\ketbra{0}$ and obtains $\nu_{U}$. 
        Define $U':=U \circ \SWAP_{{\mu_0},{\mu_1}}$.
        Then it outputs $(\nu_U,U')$.
    \end{description}
    By the right-invariant property of Haar measure, the following distribution is identical to $D_1$:
    \begin{description}
        \item[$D_2(\nu_{U'},U)$:] 
        It samples a Haar random unitary $U$, $\ket{\mu_0},\ket{\mu_1}$.        
        Define $U':=U \circ \SWAP_{{\mu_0},{\mu_1}}$.
        It runs $\cD_1^{U'}$ on input $\ketbra{0}$ and obtains $\nu_{U'}$. 
        Then it outputs $(\nu_{U'},U)$.
    \end{description}
    We will prove that for any $U$, the following two mixed states are close:
    \begin{align}\label{eqn:td_resampling}
        \TD\left(\nu_U,\nu_{U'}\right) \le 2\sqrt{\frac{6q}{d}}.
    \end{align}
    Assuming this, we conclude the proof of the resampling lemma.

    Now we prove \cref{eqn:td_resampling}. 
    Let $S=\Span(\ket{\mu_0},\ket{\mu_1})$.
    We first define the projection
    $
        P_{+}:= I - \Pi_S.
    $
    By the same analysis as in \cref{eqn:first}, we can prove that
    \begin{align}
        \Exp_{\mu_0,\mu_1}\left[
            \|\Pi_S \ket{\phi}\|^2_2
        \right]\le \frac{6}{d}.
    \end{align}
    Also note that $U'=U\circ \SWAP_{{\mu_0},{\mu_1}} =  \SWAP_{U^{-1}\ket {\mu_0},U^{-1}\ket{\mu_1}} \circ U$. 
    We define $T=\Span(U^{-1}\ket {\mu_0},U^{-1}\ket{\mu_1})$ and $P_{-}:= I- \Pi_T$, which satisfies by the same reason:
    \begin{align}
        \Exp_{\mu_0,\mu_1}\left[
            \|\Pi_T \ket{\psi}\|^2_2
        \right]\le  \frac{6}{d}.
    \end{align}

    Observe that
    \begin{align}
        U P_+ \ket{\phi} = U' P_+ \ket{\phi} \text{ and }U^{-1}P_- \ket{\psi} = (U')^{-1}P_- \ket{\psi}
    \end{align}
    for any quantum states $\ket{\phi},\ket{\psi}$.
    Therefore, we have
    \begin{align}
        &
        \TD\left(\nu_U, \nu_{U'}\right)
        \\&
        \le \TD\left(
            \Phi \circ U^{\pm1}\circ ... \circ U^{\pm1} \circ \Phi(\ketbra{0}), \Phi \circ U^{\pm1}P_{\pm}\circ ... \circ U^{\pm1}P_{\pm}\circ \Phi(\ketbra{0})
        \right)
        \\&
        +\TD\left( \Phi \circ (U')^{\pm1}P_{\pm}\circ ... \circ (U')^{\pm1}P_{\pm}\circ \Phi(\ketbra{0}),
            \Phi \circ (U')^{\pm1}\circ ... \circ (U')^{\pm1} \circ \Phi(\ketbra{0})
        \right)
    \end{align}
    where the terms in $U^{\pm1}P_{\pm}$, $(U')^{\pm1}P_{\pm}$ always have the same signs. 
    The quantum union bound (\cref{lem:quantum_union}) ensures that each term is bounded above by $\sqrt{6q/d}$. Thus we have
    \begin{align}
        \TD\left(\nu_U, \nu_{U'}\right)\le 2\sqrt{\frac{6q}{d}}
    \end{align}
    for any $U$. This proves \cref{eqn:td_resampling}.
\end{proof}

\section{Breaking quantum-accessible PRFSG security}\label{sec:how_to_break_UP}
We prove that the constructions in this section are not quantum-accessible PRFSGs.
\begin{theorem}\label{thm:attack_XUX}
    Let $U$ be an $n$-qubit Haar random unitary given as an oracle and $a,b$ be random $n$-bit strings.
    Then, $X^aUX^b$ is not a quantum-accessible (nonadaptively-secure) PRFSGs in the QHRO model even without inverse access to the QHRO. More explicitly, a polynomial-time algorithm exists given non-adaptive oracle access to $U$ and $X^a U X^b$ that finds $a,b$ with overwhelming probability.
\end{theorem}

We also consider the random Pauli variant and prove the following theorem.
\begin{theorem}\label{thm:attack_UP}
    Let $U$ be an $n$-qubit Haar random unitary given as an oracle and $P$ be a random Pauli operator over $n$ qubits. 
    Then, $UP$ is not a quantum-accessible (nonadaptively-secure) PRFSGs in the QHRO model even without inverse access to the QHRO. More explicitly, a polynomial-time algorithm exists given non-adaptive oracle access to $U$ and $UP$ that finds $P$ with overwhelming probability.
\end{theorem}

Before proceeding to the attack, we use the following variant of Simon's algorithm for quantum states.

\begin{lemma}\label{lem:qHSP}
    Let $(\ket{\xi_x})_{x\in \bit^n}$ be quantum states.
    Suppose that there exists $t\in \bit^n\setminus \{0^n\}$ such that $\braket{\xi_x|\xi_{x\oplus t}}=1$ for any $x\in \bit^n$
    and there exists a constant $0\le c<1$ such that $|\braket{\xi_x|\xi_{x'}}| \le c$ if $x\oplus x' \notin \{0^n,t\}$.
    Suppose that there is an efficient algorithm $A$ that prepares
    \begin{align}\label{eqn:target_qHSP}
    \frac{\sum_{x\in \bit^n} \ket{x}\ket{\xi_x}^{\otimes t}}{\sqrt{2^n}}
    \end{align}
    for $t$ such that $c^t\le 2^{-2n+4}$.
    Then, there exists an algorithm that recovers $t$ using $O(n)$ calls to $A$ with overwhelming probability.
\end{lemma}
\begin{proof}
    Consider the following subroutine
    \begin{align}
    \frac{\sum_{x\in \bit^n} \ket{x}\ket{\xi_x}^{\otimes t}}{\sqrt{2^n}}
    &\mapsto
    \frac{\sum_{x,y\in \bit^n} (-1)^{x\cdot y}\ket{y}\ket{\xi_x}^{\otimes t}}{{2^n}}\\
    &=
    \frac{\sum_{y\in\bit^n} \ket{y} \sum_{x\in\bit^n} (-1)^{x\cdot y}\ket{\xi_x}^{\otimes t}}{{2^n}}\\
    &=
    \frac{\sum_{y\in\bit^n} \ket{y} \sum_{x\in X} ((-1)^{x\cdot y} + (-1)^{(x\oplus t) \cdot y})\ket{\xi_x}^{\otimes t}}{{2^n}}
    \end{align} 
    for $X\subset \bit^n$ of size $2^{n-1}$ such that $X\cup \{x\oplus t : x\in X\} = \bit^n$,
    where the first state is prepared by $A$ then we apply the inverse QFT on the first register. Measuring the first $n$ qubits, we obtain $y$ such that $y\cdot t=0$.
    Specifically, the probability of obtaining $y$ is,
    \begin{align}
        &\left|\frac{\| 2\sum_{x\in X} (-1)^{x\cdot y} \ket{\xi_x}^{\otimes t}\|^2}{4^n} - \frac{1}{2^{n-1}}\right|\\
        &=\left|\frac{2^{n+1} + \sum_{x,x'\in X, x\neq x'} (-1)^{(x\oplus x')\cdot y}\braket{\xi_x|\xi_{x'}}^{ t}}{4^n}-\frac{1}{2^{n-1}}\right|
        \\&
        \le\sum_{x,x'\in X, x\neq x'}\frac{|\braket{\xi_x|\xi_{x'}}^{ t}|}{4^n} \le \frac{2^{2n-2}c^t}{4^n} = \frac{c^t}{4} \le \frac{1}{2^{2n-2}}.
    \end{align}
    Therefore, the output is $(1/2^{n-1})$-close in the statistical distance to the uniform distribution over $y$ such that $y\cdot t = 0$. Repeating this procedure $O(n)$ times, we can recover $t$ with overwhelming probability.
\end{proof}

\subsection{Breaking XUX}
\begin{proof}[Proof of \cref{thm:attack_XUX}]
Let $U$ be Haar random unitary and $V:=X^a U X^b$ for random $n$-bit strings $a,b$. 
Let $\ket{\Phi}=\sum_{x\in \bit^n} \ket{x,x}/\sqrt{2^n}$ be the maximally entangled state. We have
\begin{align}\label{eqn:const1}
    &(X^x\otimes I) \cdot (U\otimes U) \cdot (X^y \otimes I)\ket{\Phi}
    \otimes(X^x\otimes I)\cdot 
    \left(
        { V \otimes U}
    \right) \cdot 
    (X^{ y} \otimes I) \ket{\Phi} \\&\label{eqn:const2}
    =
    (X^x U X^y \otimes U)\ket{\Phi}
    \otimes
    ({X^{a\oplus x} U X^{b\oplus y} \otimes U}) \ket{\Phi} .
\end{align}
We write $(X^x UIX^y \otimes U)\ket{\Phi}=: \ket{U_{x,y}}.$
We then consider the state
\begin{align}
    \ket{\xi_{x,y}}  &= \frac{
    (X^x U X^y \otimes U)
    \otimes
    ({X^{a\oplus x} U X^{b\oplus y} \otimes U}) 
    +
    ({X^{a\oplus x} U X^{b\oplus y} \otimes U}) 
    \otimes
    (X^x U X^y \otimes U)
    }{\sqrt{2}}\ket{\Phi,\Phi} \\
    \label{eqn:xixy}
    &= \frac{\ket{U_{x,y}}\otimes \ket{U_{a\oplus x , b\oplus y}} + \ket{U_{a\oplus x , b\oplus y}}\otimes\ket{U_{x,y}} }{\sqrt{2}}.
\end{align}
Note that $\ket{\xi_{x\oplus a,y\oplus b}}= \ket{\xi_{x,y}}$ holds. On the other hand, we later prove that $|\braket{\xi_{x,y}|\xi_{x',y'}}| \le {2n/2^{n/2}}$ holds for all pairs such that $(x,y)\oplus (x',y') \neq (0,0) \text{ or }(a,b)$ with an overwhelming probability over random $U$.
Assuming this, $t=O(1)$ satisfies the condition of \cref{lem:qHSP} with overwhelming probability.

To prepare the target state, we prepare
\begin{align}
    \sum_{x,y} \frac{1}{2^n} \ket{x,y} \otimes (\ket{\Phi}^{\otimes2})^{\otimes t}
    &\mapsto\sum_{x,y} \frac{1}{2^n} \ket{x,y} \otimes (\ket{U_{x,y}}\otimes \ket{U_{a\oplus x,b\oplus y}})^{\otimes t}
\end{align}
using \cref{eqn:const1,eqn:const2}.
Then compute the projection $I_{2^{2n}}\otimes\Pi_{\text{sym}}^{\otimes t}$ where $\Pi_{\text{sym}}:=\Pi^{(2^{2n},2)}_{\text{sym}}$ is the projection to the space
$ \left\{
    \frac{\ket{p,q} + \ket{q,p}}{\sqrt{2}}: p,q \in \bit^{2n}
\right\}$ as defined in \cref{lem:symmetric_subspace}. 
The probability of success is at least $1/2^t$, because
\begin{align}
    &\left\| (I_{2^{2n}}\otimes \Pi_{\text{sym}}^{\otimes t})
        \sum_{x,y}\frac{1}{2^n} \ket{x,y} \otimes(\ket{U_{x,y}}\otimes \ket{U_{a\oplus x,b\oplus y}})^{\otimes t}
    \right\|^2\\
    &=\frac{1}{2^{2n}} \sum_{x,y}\left\| \Pi_{\text{sym}}\ket{U_{x,y}}\otimes \ket{U_{a\oplus x,b\oplus y}}
    \right\|^{2t}
    \\&\ge\frac{1}{2^{2n}} \cdot\sum_{x,y} \left(\frac{1}{2}\right)^{t} = \frac{1}{2^t}
\end{align}
where we use the fact that $\ket{U_{x,y}}\otimes \ket{U_{a\oplus x,b\oplus y}}$ is identical to
\begin{align}
    \frac{\ket{U_{x,y}}\otimes \ket{U_{a\oplus x,b\oplus y}} + \ket{U_{a\oplus x,b\oplus y}}\otimes \ket{U_{x,y}}}{2} + \frac{\ket{U_{x,y}}\otimes \ket{U_{a\oplus x,b\oplus y}} -\ket{U_{a\oplus x,b\oplus y}}\otimes \ket{U_{x,y}}}{2}
\end{align}
which implies the projection onto the symmetric subspace succeeds with probability $1/2$ each. In particular, if all the projections succeed, the outcome becomes 
\begin{align}\label{eqn:targetXUX}
    \sum_{x,y}\frac{1}{2^n} \ket{x,y} \otimes
    \ket{\xi_{x,y}}^{\otimes t}.
\end{align}

It remains to prove that $|\braket{\xi_{x,y}|\xi_{x',y'}}|$ is small for $(x,y)\oplus (x',y') \neq (0,0) \text{ or }(a,b)$.
We use the following lemma.

\begin{lemma}\label{lem:X^aUX^bU^dagX^c_is_small}
    Let $a,b$ and $c$ be $n$ bit strings such that $a\neq c$. Then, 
    \begin{align}\label{eqn:XUXUX}
        |\bra{0}X^aUX^bU^\dag X^c\ket{0}|^2\le\frac{n}{\sqrt{2^n}}
    \end{align}
    with probability at least $1-e^{-O(n^2)}$ over the choice of $U$ with respect to $\mu_{2^n}$.
\end{lemma}
The proof is given below.
Assume that this lemma is true. Then, by the union bound, with probability at least $1-2^{2n}e^{-O(n^2)} = 1-\negl(n)$, \cref{eqn:XUXUX} holds for any $a\neq c.$
The inner product is, using \cref{eqn:xixy}, 
\begin{align}\label{eqn:UUUU}
    \braket{\xi_{x,y}|\xi_{x',y'}} = {\braket{U_{x,y}|U_{x',y'}}\braket{U_{a\oplus x,b\oplus y}|U_{a\oplus x',b\oplus y'}} + \braket{U_{x,y}|U_{a\oplus x',b\oplus y'}}\braket{U_{a\oplus x,b\oplus y}|U_{x',y'}} }.
\end{align}
For $y' \neq y $, we have
\begin{align}
    |\braket{U_{x,y}|U_{x',y'}}| &= |\bra{\Phi}(X^{x}U X^y \otimes U)^\dag (X^{x'}U X^{y'} \otimes U)\ket{\Phi}|
    \\&
    =|\bra{\Phi}(X^y U^\dag X^{x \oplus x'} U X^{y'} \otimes I)\ket{\Phi}|\\
    &=\left|\frac{\sum_{i,j\in \bit^n} \bra{i,i}(X^y U^\dag X^{x \oplus x'} U X^{y'} \otimes I)\ket{j,j}}{2^n}\right|\\
    &=\left|\frac{\sum_{i\in \bit^n} \bra{i}X^y U^\dag X^{x \oplus x'} U X^{y'} \ket{i}}{2^n}\right|\\
    &\le\sum_{i\in \bit^n} \frac{|\bra{i}X^y U^\dag X^{x \oplus x'} U X^{y'} \ket{i}|}{2^n}
    \le (\frac{n}{\sqrt{2^n}})^{1/2}.
\end{align}
The same inequality holds for $|\braket{U_{a\oplus x,b\oplus y}|U_{a\oplus x',b\oplus y'}}|$ and $|\braket{U_{x,y}|U_{a\oplus x',b\oplus y'}}|$ if $y'\notin \{y,y\oplus b\}$. Also, if $x\neq x'$, we have
\begin{align}
    |\braket{U_{x,y}|U_{x',y'}}| &= |\bra{\Phi}(X^{x}U X^y \otimes I)^\dag (X^{x'}U X^{y'} \otimes I)\ket{\Phi}|\\&
    = |\bra{\Phi}(I \otimes (X^{x}U X^y)^T)^\dag (I \otimes (X^{x'}U X^{y'})^T)\ket{\Phi}|
    \\&
    = |\bra{\Phi}(I \otimes X^{x}(U^T)^\dag X^{y\oplus y'}U^T X^{x'})\ket{\Phi}|
\end{align}
using the ricochet property of the maximally mixed state $(A \otimes I)\ket{\Phi} = (I\otimes A^T) \ket{\Phi}$. A simple calculation gives the same inequality holds for this case. If $(x',y') \notin \{(x,y) , (x\oplus a,y\oplus b)\}$, by the case-by-case analysis on each term of \cref{eqn:UUUU}, it must hold that
\begin{align}
    |\braket{\xi_{x,y}|\xi_{x',y'}} | \le \frac{2n}{2^{n/2}}.
\end{align}

Therefore, we can prepare the state in \cref{eqn:targetXUX} in polynomial time which satisfies the conditions of \cref{lem:qHSP}. Applying the attack in the lemma, we conclude the proof.    
\end{proof}

The proof of \cref{lem:X^aUX^bU^dagX^c_is_small} relies on the following lemma.
\begin{lemma}[\cite{emerson2005scalable}]\label{lem:E_UAU^dagCUBU^dag}
    Let $A,B$ and $C$ be $d\times d$ matrix. Then,
    \begin{align}
        \Exp_{U\gets\mu_d}UAU^\dag CUBU^\dag=\frac{\Tr[AB]\Tr[C]}{d}\frac{I}{d}+\frac{d\Tr[A]\Tr[B]-\Tr[AB]}{d(d^2-1)}\bigg(C-\Tr[C]\frac{I}{d}\bigg)
    \end{align}
\end{lemma}

\begin{proof}[Proof of \cref{lem:X^aUX^bU^dagX^c_is_small}]
    We show by the concentration inequality. To invoke it, we need the following expectation:
    \begin{align}
        &\Exp_{U\gets\mu_{2^n}}|\bra{0}X^aUX^bU^\dag X^c\ket{0}|^2\\
        =& \Exp_{U\gets\mu_{2^n}}|\bra{a}UX^bU^\dag \ket{c}|^2\\
        =&\Exp_{U\gets\mu_{2^n}}\bra{a}UX^bU^\dag\ketbra{c}{c}UX^bU^\dag\ket{a}\\
        =&\bra{a}\frac{\Tr[(X^b)^2]\Tr[\ket{c}\bra{c}]}{2^n}\frac{I}{2^n}+\frac{2^n\Tr[X^b]\Tr[X^b]-\Tr[(X^b)^2]}{2^n(2^{2n}-1)}\bigg(\ket{c}\bra{c}-\Tr[\ket{c}\bra{c}]\frac{I}{2^n}\bigg)\ket{a}\\
        =&\bra{a}\frac{I}{2^n}+\frac{2^n(2^{2n-2h(b)}-1)}{2^n(2^{2n}-1)}\bigg(\ketbra{c}{c}-\frac{I}{2^n}\bigg)\ket{a}\\
        =&\Theta(2^{-n}),
    \end{align}
    where we have used 
    \begin{itemize}
        \item \cref{lem:E_UAU^dagCUBU^dag} in the third equality;
        \item $\Tr[X^b]=2^{n-h(b)}$ in the fourth equality, where $h(b)$ is the hamming distance of $b$;
        \item $a\neq c$ and $0\le\frac{2^n(2^{2n-2h(b)}-1)}{2^n(2^{2n}-1)}\le1$ in the last equality.
    \end{itemize}
    Note that we can see $\bra{0}X^aUX^bU^\dag X^c\ket{0}|^2$ is the probability that some algorithm given access to $U$ and $U^\dag$ outputs 1. From this and \cref{lem:Lipschitz}, $\bra{0}X^aUX^bU^\dag X^c\ket{0}|^2$ is $4$-Lipshcitz for $U$.
    Therefore, from the concentration inequality \cref{thm:Haar_concentration},
    \begin{align}
        \Pr_{U\gets\mu_{2^n}}[|\bra{0}X^aUX^bU^\dag X^c\ket{0}|^2\le\frac{n}{\sqrt{2^n}}]
        \ge&\Pr_{U\gets\mu_{2^n}}\bigg[\bigg||\bra{0}X^aUX^bU^\dag X^c\ket{0}|^2-\Theta(2^{-n})\bigg|\le\frac{n}{2\sqrt{2^n}}\bigg]\\
        \ge&1-\exp\bigg(-O\bigg(2^n\frac{n^2}{2^n}\bigg)\bigg)\\
        \ge&1-e^{-O(n^2)}.
    \end{align}
\end{proof}

\subsection{Breaking UP}
\begin{proof}[Proof of \cref{thm:attack_UP}]
We construct the quantum states $\ket{\xi}$ given oracle access to $U$ and $UP=:V$ for random Pauli operator $P$ over $n$ qubits.
More concretely, for $(x,z)\in \bit^{n}\times \bit^n$, write $P_{x,z}$ to denote \begin{align}
    i^{x\cdot z }X^{\otimes x} Z^{\otimes z} = i^{x\cdot z } (X^{x_1} Z^{z_1})\otimes ... \otimes (X^{x_n} Z^{z_n}) .
\end{align}
We define $V=UP$ for random Pauli $P=P_{a,b}$. 
Note that $P_{x,z} \cdot P_{x',z'} = i^{x\cdot z' - x' \cdot z} P_{(x,z)\oplus(x',z')} = (-1)^{x\cdot z' - x' \cdot z}P_{x',z'}\cdot P_{x,z}$.
It well known that $\{\ket{P_{x,z}}:=(P_{x,z}\otimes I) \ket{\Phi}\}_{x,z}$ consists the orthonormal basis for the maximally mixed state $\ket{\Phi}$.

Let $(x',z')=(x,z)\oplus(a,b).$ It holds that
\begin{align}
    P_{a,b}  P_{x',z'} = i^{a\cdot z'- x'\cdot b} \cdot P_{x,z} = i^{a\cdot z- x\cdot b} \cdot P_{x,z}, \text{ and } P_{x',z'} = i^{a\cdot z - x\cdot b} \cdot P_{a,b}P_{x,z} .
\end{align}
Consider
\begin{align}\label{eqn:qHSP_concrete}
    \ket{\phi_{x,z}}&:= 
    \frac{1}{\sqrt 2}\left [(V\otimes I \otimes U\otimes I) + (U \otimes I \otimes V \otimes I)\right] (P_{x,z} \otimes I \otimes P_{x,z} \otimes I) \ket{\Phi,\Phi}
    \\
    &=(U\otimes I)^{\otimes 2} \left[\frac{(PP_{x,z} \otimes I) \ket{\Phi} \otimes(P_{x,z} \otimes I) \ket{\Phi} + (P_{x,z} \otimes I) \ket{\Phi} \otimes(PP_{x,z} \otimes I) \ket{\Phi} }{\sqrt 2}\right]
    \\&=i^{a\cdot z -x\cdot b}(U\otimes I)^{\otimes 2} \left[\frac{\ket{P_{x',z'},P_{x,z}}+\ket{P_{x,z},P_{x',z'}}}{\sqrt{2}}\right].
\end{align}
Similarly, $(U^{\dagger}\otimes I)^{\otimes 2}\ket{\phi_{x',z'}}$ can be expressed by
\begin{align}
    &\frac{(P_{a,b}P_{x',z'} \otimes I) \ket{\Phi} \otimes(P_{x',z'} \otimes I) \ket{\Phi} + (P_{x',z'} \otimes I) \ket{\Phi} \otimes(P_{a,b}P_{x',z'} \otimes I) \ket{\Phi} }{\sqrt 2}
    \\
    &=(-1)^{a\cdot z - x\cdot b}\cdot \frac{(P_{x,z} \otimes I) \ket{\Phi} \otimes(P_{a,b}P_{x,z} \otimes I) \ket{\Phi} + (P_{a,b}P_{x,z}  \otimes I) \ket{\Phi} \otimes(P_{x,z} \otimes I) \ket{\Phi} }{\sqrt 2}
    \\
    &= (-i)^{a\cdot z - x\cdot b}\frac{\ket{P_{x',z'},P_{x,z}}+\ket{P_{x,z},P_{x',z'}}}{\sqrt{2}}.
\end{align}
From this and the orthogonality of $\{(P_{x,z}\otimes I)\ket{\Phi}\}_{x,z}$, we derive that $\ket{\phi_{x,z}}^{\otimes 2} , \ket{\phi_{x',z'}}^{\otimes 2}$ are identical for $(x',z')\in \{(x,z),(x,z)\oplus(a,b)\}$ and otherwise orthogonal. Therefore, $\ket{\xi_{x,z}} := \ket{\phi_{x,z}}^{\otimes 2}$ be our target states. To construct $\sum_{x,z}\ket{x,z,\xi_{x,z}}$, we prepare
\begin{align}
    \sum_{x,z} \frac{1}{2^n} \ket{x,z} \otimes \ket{\Phi}^{\otimes4}
    &\mapsto\sum_{x,z} \frac{1}{2^n} \ket{x,z} \otimes \ket{P_{x,z}}^{\otimes 4}\\
    &\mapsto\sum_{x,z}\frac{1}{2^n} \ket{x,z} \otimes(U\otimes I)^{\otimes 4}  \ket{P_{x',z'},P_{x,z}}^{\otimes 2}
\end{align}
where the first step is to apply $(P_{x,z}\otimes I)^{\otimes4}$ and the second step apply $(V\otimes I \otimes U \otimes I)^{\otimes2}$.\footnote{We ignore the phase which becomes irrelevant.}
Then compute the projection $I_{2^{2n}}\otimes\Pi_{\text{sym}}^{\otimes 2}$ where $\Pi_{\text{sym}}:=\Pi^{(2^{2n},2)}_{\text{sym}}$ is the projection to the space
$ \left\{
    \frac{\ket{p,q} + \ket{q,p}}{\sqrt{2}}: p,q \in \bit^{2n}
\right\}$ as defined in \cref{lem:symmetric_subspace}.
It is not hard to see that the probability of success is at least $1/4$. 
This is because
\begin{align}
    &\left\| (I_{2^{2n}}\otimes \Pi_{\text{sym}}^{\otimes 2})
        \sum_{x,z}\frac{1}{2^n} \ket{x,z} \otimes(U\otimes I \otimes U\otimes I)^{\otimes2} \ket{P_{x',z'},P_{x,z}}^{\otimes 2}
    \right\|^2\\
    &=\frac{1}{2^{2n}} \sum_{x,z}\left\| \Pi_{\text{sym}}\ket{P_{x',z'},P_{x,z}}
    \right\|^4
    \\&\ge\frac{1}{2^{2n}} \cdot  \sum_{x,z}\left(\frac{1}{2}\right)^2 = \frac{1}{4}
\end{align}
where we use the invariant of the symmetric subspace under any unitary in the first equality, and 
\begin{align}
    \ket{P_{x',z'},P_{x,z}} = \frac{\ket{P_{x',z'},P_{x,z}}+\ket{P_{x,z},P_{x',z'}}}{2} +\frac{\ket{P_{x',z'},P_{x,z}}-\ket{P_{x,z},P_{x',z'}}}{2}.
\end{align}
Therefore, given $U,UP$ we can efficiently construct the state in~\Cref{eqn:target_qHSP} for $\ket{\xi_{x,z}} = \ket{\phi_{x,z}}^{\otimes2}$ for $\ket{\phi_{x,z}}$ defined in~\Cref{eqn:qHSP_concrete}. By~\Cref{lem:qHSP}, we can extract $(a,b)$ for $P=P_{a,b}$, thus $UP$ cannot be a secure quantum-accessible PRFSG.
\end{proof}

\section{Application of Haar Twirl Approximation: Alternative Proof of Non-Adaptive Security of PFC Ensemble}\label{subsec:alternative_proof_of_non-adaptive_security}

In this section, we give an alternative proof of the non-adaptive security of PFC ensemble \cite{metger2024simple}. They essentially use the Schur-Weyl duality in the proof of \cite{metger2024simple}. However, our proof does not invoke it and essentially uses the Weingarten calculus.

\subsection{Definitions and Lemmas}

First, we define the action of a permutation unitary and a binary phase unitary.
\begin{definition}[Permutation Unitaries on $\C^d$]
    Let $S_d$ be a set of all permutations over $d$ elements. For each $\pi\in S_d$, we define the permutation unitary $P_\pi$ on $\C^d$ that acts
    \begin{align}
        P_\pi\ket{x}=\ket{\pi(x)}
    \end{align}
    for all $x\in[d]$.
\end{definition}

\begin{definition}[Binary Phase Unitaries]
    For a function $f:[d]\to\bit$, we define the binary phase unitary $F_f$ on $\C^d$ that acts
    \begin{align}
        F_f\ket{x}=(-1)^{f(x)}\ket{x}
    \end{align}
    for all $x\in[d]$.
\end{definition}

\begin{definition}[$k$-wise twirl]
    Let $k,d\in\N$ and $\cF$ be a set of all functions $f:[d]\to\bit$. We define the $PF$ $k$-wise twirl $\cM^{(t)}_{PF}$ and $PFC$ $k$-wise twirl $\cM^{(t)}_{PFC}$ as follows:
    \begin{align}
        \cM^{(t)}_{PF}(\cdot)&\coloneqq\Exp_{
        \pi\gets S_d, f\gets\cF}(P_\pi F_f)^{\otimes k}(\cdot)(P_\pi F_f)^{\dag\otimes k},\\
        \cM^{(t)}_{PFC}(\cdot)&\coloneqq\Exp_{
        \substack{\pi\gets S_d, f\gets\cF,\\ C\gets\nu}}(P_\pi F_fC)^{\otimes k}(\cdot)(P_\pi F_fC)^{\dag\otimes k}.
    \end{align}
    Here, $S_d$ is the set of all permutations over $[d]$, $\cF$ is a set of all functions $f:[d]\to\bit$, and $\nu$ is any unitary 2-design. 
\end{definition}

The following two lemmas are shown in \cite{metger2024simple}, both of which are from the straightforward computation (without Schur-Weyl duality).

\begin{lemma}[Lemma 3.2 in \cite{metger2024simple}]\label{lem:Clifford-twril}
    Let $k,d\in\N$ and $\nu$ be any unitary 2-design. Define $\Lambda$ be the projection onto
    \begin{align}
        \text{span}\{\ket{x_1,...,x_k};x_1,...,x_k\in[d]\text{ and }x_1,...,x_k\text{ are distinct.}\}.
    \end{align}
    Then, for any quantum state $\rho$,
    \begin{align}
        \Tr[\Lambda\Exp_{C\gets\nu}C^{\otimes k}\rho C^{\dag\otimes k}]\ge1-O\bigg(\frac{k^2}{d}\bigg).
    \end{align}
\end{lemma}

\begin{lemma}[Immediate corollary of Lemma 3.8 of \cite{metger2024simple}]\label{lem:PF_twirl}
    Let $\regA$ be a $d^k$-dimentional register and $\regB$ be any register. Let $\Lambda$ be the projection defined in \cref{lem:Clifford-twril}. Then, for any state $\rho_{\regA\regB}$ such that $(\Lambda_\regA\otimes I_\regB)\rho_{\regA\regB}(\Lambda_\regA\otimes I_\regB)=\rho_{\regA\regB}$,
    \begin{align}
        (\cM^{(k)}_{PF,\regA}\otimes\id_{\regB})(\rho_{\regA\regB})=\sum_{\sigma\in S_k}\frac{\Lambda_\regA}{\Tr[\Lambda]}R^\dag_{\sigma,\regA}\otimes\Tr_{\regA'}[(R_{\sigma,\regA'}\otimes I_\regB)\rho_{\regA'\regB}].
    \end{align}
    Here $\regA'$ is a register whose size is the same as that of the register $\regA$.
\end{lemma}

\subsection{Proof}

Now by using \cref{lem:approximation_formula_for_Haar_k-fold}, we show the following theorem which is originally shown by invoking the Schur-Weyl duality in \cite{metger2024simple}.

\begin{theorem}\label{thm:other_proof_of_PFC}
    Let $k,d\in\N$ such that $d>\sqrt{6}k^{7/4}$. Then, 
    \begin{align}
        \bigg\|\cM^{(k)}_{\text{Haar}}-\cM^{(k)}_{PFC}\bigg\|_\diamond\le O\bigg(\frac{k}{\sqrt{d}}\bigg).
    \end{align}
\end{theorem}


\begin{proof}[Proof of \cref{thm:other_proof_of_PFC}]
    Let $\regA$ be a $d^k$-dimentional register and $\regB$ be any register. It suffices to show that for any state $\rho_{\regA\regB}$,
    \begin{align}
        \bigg\|(\cM^{(k)}_{\text{Haar},\regA}\otimes\id_{\regA\regB})(\rho_{\regA\regB})-(\cM^{(k)}_{PFC,\regA}\otimes\id_{\regB})(\rho_{\regA\regB})\bigg\|_1\le O\bigg(\frac{k}{\sqrt{d}}\bigg).
    \end{align}
    Define
    \begin{align}
        \xi_{\regA\regB}&\coloneqq\Exp_{C\gets\nu}(C^{\otimes k}_\regA\otimes I_\regB)\rho_{\regA\regB}(C^{\otimes k}_\regA\otimes I_\regB)^\dag,\text{ and}\\
        \xi'_{\regA\regB}&\coloneqq\frac{(\Lambda_\regA\otimes I_\regB)\xi_{\regA\regB}(\Lambda_\regA\otimes I_\regB)}{\Tr[(\Lambda_\regA\otimes I_\regB)\xi_{\regA\regB}]}.
    \end{align}
    From \cref{lem:Clifford-twril,lem:gentle_measurement}, we have
    \begin{align}
        \|\xi_{\regA\regB}-\xi'_{\regA\regB}\|_1
        \le& O\bigg(\frac{k}{\sqrt{d}}\bigg).\label{eq:eq1_other_proof_of_PFC}
    \end{align}
    Thus,
    \begin{align}
        &\bigg\|(\cM^{(k)}_{\text{Haar},\regA}\otimes\id_{\regA\regB})(\rho_{\regA\regB})-(\cM^{(k)}_{PFC,\regA}\otimes\id_{\regB})(\rho_{\regA\regB})\bigg\|_1\\
        =&\bigg\|(\cM^{(k)}_{\text{Haar},\regA}\otimes\id_{\regA\regB})(\xi_{\regA\regB})-(\cM^{(k)}_{PF,\regA}\otimes\id_{\regB})(\xi_{\regA\regB})\bigg\|_1\\
        \le&\bigg\|(\cM^{(k)}_{\text{Haar},\regA}\otimes\id_{\regA\regB})(\xi'_{\regA\regB})-(\cM^{(k)}_{PF,\regA}\otimes\id_{\regB})(\xi'_{\regA\regB})\bigg\|_1+O\bigg(\frac{k}{\sqrt{d}}\bigg),\\
    \end{align}
    where the equality follows from the right and left invariance of the Haar measure, and the inequality follows from \cref{eq:eq1_other_proof_of_PFC} and the triangle inequality.
    
    To conclude the proof, we show
    \begin{align}
        \bigg\|(\cM^{(k)}_{\text{Haar},\regA}\otimes\id_{\regA\regB})(\xi'_{\regA\regB})-(\cM^{(k)}_{PF,\regA}\otimes\id_{\regB})(\xi'_{\regA\regB})\bigg\|_1\le O\bigg(\frac{k}{\sqrt{d}}\bigg).\label{eq:main_eq_for_other_proof_of_PFC}
    \end{align}
    Let us consider the following hybrids of matrices:
    \begin{itemize}
        \item $\xi_{0,\regA\regB}\coloneqq(\cM^{(k)}_{\text{Haar},\regA}\otimes\id_{\regA\regB})(\xi'_{\regA\regB})$.
        \item $\xi_{1,\regA\regB}\coloneqq\sum_{\sigma\in S_k}\frac{1}{d^k}R^\dag_{\sigma,\regA}\otimes\Tr_{\regA'}[(R_{\sigma,\regA'}\otimes I_\regB)\xi'_{\regA'\regB}]$.
        \item $\xi_{2,\regA\regB}\coloneqq\sum_{\sigma\in S_k}\frac{\Lambda_\regA}{\Tr[\Lambda]}R^\dag_{\sigma,\regA}\otimes\Tr_{\regA'}[(R_{\sigma,\regA'}\otimes I_\regB)\xi'_{\regA'\regB}]$.
        \item $\xi_{3,\regA\regB}\coloneqq(\cM^{(k)}_{PF,\regA}\otimes\id_{\regB})(\xi'_{\regA\regB})$
    \end{itemize}
    From \cref{lem:approximation_formula_for_Haar_k-fold}, we have
    \begin{align}
        \|\xi_0-\xi_1\|_1\le O(k^2/d).\label{eq:eq2_other_proof_of_PFC}
    \end{align} 
    Moreover, we have 
    \begin{align}
        \xi_2=\xi_3 \label{eq:eq3_other_proof_of_PFC}
    \end{align} 
    from \cref{lem:PF_twirl}. Thus, it suffices to show $\|\xi_1-\xi_2\|_1\le O(k^2/d)$. Note that $\xi_{1,\regA\regB}$ is invariant under the action of 2-design twirl because
    \begin{align}
        \Exp_{C\gets\nu}(C^{\otimes k}_\regA\otimes I_\regB)\xi_{1,\regA\regB}(C^{\otimes k}_\regA\otimes I_\regB)^\dag
        &=\sum_{\sigma\in S_k}\frac{1}{d^k}\Exp_{C\gets\nu}(C^{\otimes k}R_\sigma^\dag C^{\dag\otimes k})_\regA\otimes\Tr_{\regA'}[(R_{\sigma,\regA'}\otimes I_\regB)\xi'_{\regA'\regB}]\\
        &=\sum_{\sigma\in S_k}\frac{1}{d^k}R^\dag_{\sigma,\regA}\otimes\Tr_{\regA'}[(R_{\sigma,\regA'}\otimes I_\regB)\xi'_{\regA'\regB}]\\
        &=\xi_{1,\regA\regB},
    \end{align}
    where we have used the fact $U^{\otimes k}R_\sigma U^{\dag\otimes k}=R_\sigma$ for any $\sigma\in S_k$ and $U\in\Unitaries(d)$ in the second equality. This and \cref{lem:Clifford-twril} imply
    \begin{align}
        \Tr[(\Lambda_\regA\otimes I_\regB)\xi_{1,\regA\regB}]\ge1-O\bigg(\frac{k^2}{d}\bigg).\label{eq:xi_1_has_large_overlap}
    \end{align}
    Thus, we have
    \begin{align}
        \|\xi_{1,\regA\regB}-\xi_{2,\regA\regB}\|_1
        &\le \|(\Lambda_\regA\otimes I_\regB)\xi_{1,\regA\regB}(\Lambda_\regA\otimes I_\regB)-\xi_{2,\regA\regB}\|_1+O\bigg(\frac{k}{\sqrt{d}}\bigg)\\
        &=\bigg\|\bigg(\frac{\Tr[\Lambda]}{d^k}-1\bigg)\xi_{2,\regA\regB}\bigg\|_1+O\bigg(\frac{k}{\sqrt{d}}\bigg)\\
        &=O\bigg(\frac{k^2}{d}\bigg)+O\bigg(\frac{k}{\sqrt{d}}\bigg)\\
        &\le O\bigg(\frac{k}{\sqrt{d}}\bigg),\label{eq:eq4_other_proof_of_PFC}
    \end{align}
    where we have used
    \begin{itemize}
        \item \cref{eq:xi_1_has_large_overlap,lem:gentle_measurement} in the first inequality and
        \item $\Tr[\Lambda]=d(d-1)\cdots(d-k+1)$ and $(\Lambda_\regA\otimes I_\regB)\xi_{1,\regA\regB}(\Lambda_\regA\otimes I_\regB)=\xi_{2,\regA\regB}$ in the second equality.
    \end{itemize}
    Therefore, \cref{eq:main_eq_for_other_proof_of_PFC} follows from \cref{eq:eq2_other_proof_of_PFC,eq:eq3_other_proof_of_PFC,eq:eq4_other_proof_of_PFC}, which concludes the proof.
\end{proof}

\ifnum\anonymous=1
\else
\paragraph{Acknowledgments.}
SY thanks Benoît Collins for lecturing about Weingarten calculus.
SY also thanks Tomoyuki Morimae for helpful discussions and for helping him to write the introduction.
\fi

\ifnum\submission=0
\bibliographystyle{alpha} 
\else
\bibliographystyle{splncs04}
\fi
\bibliography{abbrev3,crypto,reference}

\end{document}